\documentclass[copyright,creativecommons]{eptcs}

\usepackage[latin1]{inputenc}

\usepackage[english]{babel}

\usepackage{amsthm}

\usepackage[olditem]{paralist}

\usepackage{hyperref}

\usepackage{cite}

\usepackage[final]{microtype}

\usepackage{stringstrings}

\usepackage{cutwin}

\usepackage{thmtools}

\hypersetup
  {
  pdfsubject    = {},
  pdfkeywords   = {},
  pdfproducer   = {},
  pdfcreator    = {},
  pdfpagemode   = {UseNone},
  pdfstartview  = {FitH},
  colorlinks
  }

\usepackage{etoolbox}

\AtEndPreamble
  {

  \usepackage[pagewise,mathlines]{lineno}%\linenumbers

  }

\usepackage{xargs}
\usepackage{xspace}
\usepackage{xstring}
\usepackage{boolexpr}

\usepackage[cmex10]{mathtools}
\usepackage{amsxtra}
\usepackage{amssymb}
\usepackage{amsfonts}
\usepackage{mathrsfs}
\usepackage{mathdots}
\usepackage{stmaryrd}
\usepackage{latexsym}
\usepackage{wasysym}

\usepackage{textcomp}

\newcommand{\argemp}[2]
  {\if&#1&\else#2\fi}

\newcommand{\argdef}[2]
  {\if&#1&#2\else#1\fi}

\newcommand{\argint}[3]
  {\if&#2&\else#1#2#3\fi}

\newcommand{\argext}[3]
  {\if&#1&#3\else#1\if&#3&\else#2#3\fi\fi}

\newcommandx{\argsubsup}[3][2=, 3=]
  {\def\argsubscript{{#2}}\def\argsuperscript{{#3}}#1}

\newcommandx{\argind}[9][2=, 3=, 4=, 5=, 6=, 7=, 8=, 9=]
  {%
  \switch[#1=]%
    \case{0}#2%
    \case{1}#3%
    \case{2}#4%
    \case{3}#5%
    \case{4}#6%
    \case{5}#7%
    \case{6}#8%
    \case{7}#9%
    \otherwise\ensuremath{\clubsuit}%
  \endswitch%
  }

\newcommand{\arga}[1]
  {#1}

\newcommand{\argb}[2]
  {\argext{\arga{#1}}{, \allowbreak}{#2}}

\newcommand{\argc}[3]
  {\argext{\argb{#1}{#2}}{, \allowbreak}{#3}}

\newcommand{\argd}[4]
  {\argext{\argc{#1}{#2}{#3}}{, \allowbreak}{#4}}

\newcommand{\arge}[5]
  {\argext{\argd{#1}{#2}{#3}{#4}}{, \allowbreak}{#5}}

\newcommand{\argf}[6]
  {\argext{\arge{#1}{#2}{#3}{#4}{#5}}{, \allowbreak}{#6}}

\newcommand{\argg}[7]
  {\argext{\argf{#1}{#2}{#3}{#4}{#5}{#6}}{, \allowbreak}{#7}}

\newcommand{\argh}[8]
  {\argext{\argg{#1}{#2}{#3}{#4}{#5}{#6}{#7}}{, \allowbreak}{#8}}

\newcommand{\argi}[9]
  {\argext{\argh{#1}{#2}{#3}{#4}{#5}{#6}{#7}{#8}}{, \allowbreak}{#9}}

\newcommand{\argj}[9]
  {%
  \def\valarga{#1}%
  \def\valargb{#2}%
  \def\valargc{#3}%
  \def\valargd{#4}%
  \def\valarge{#5}%
  \def\valargf{#6}%
  \def\valargg{#7}%
  \def\valargh{#8}%
  \def\valargi{#9}%
  \argauxj%
  }
\newcommand{\argauxj}[1]
  {%
  \argext%
    {
    \argi
      {\valarga} {\valargb} {\valargc} {\valargd} {\valarge} {\valargf}
      {\valargg} {\valargh} {\valargi}
    }
    {, \allowbreak}{#1}%
  }

\newcommand{\argk}[9]
  {%
  \def\valarga{#1}%
  \def\valargb{#2}%
  \def\valargc{#3}%
  \def\valargd{#4}%
  \def\valarge{#5}%
  \def\valargf{#6}%
  \def\valargg{#7}%
  \def\valargh{#8}%
  \def\valargi{#9}%
  \argauxk%
  }
\newcommand{\argauxk}[2]
  {\argext{\argauxj{#1}}{, \allowbreak}{#2}}

\newcommand{\argl}[9]
  {%
  \def\valarga{#1}%
  \def\valargb{#2}%
  \def\valargc{#3}%
  \def\valargd{#4}%
  \def\valarge{#5}%
  \def\valargf{#6}%
  \def\valargg{#7}%
  \def\valargh{#8}%
  \def\valargi{#9}%
  \argauxl%
  }
\newcommand{\argauxl}[3]
  {\argext{\argauxk{#1}{#2}}{, \allowbreak}{#3}}

\newcommand{\argm}[9]
  {%
  \def\valarga{#1}%
  \def\valargb{#2}%
  \def\valargc{#3}%
  \def\valargd{#4}%
  \def\valarge{#5}%
  \def\valargf{#6}%
  \def\valargg{#7}%
  \def\valargh{#8}%
  \def\valargi{#9}%
  \argauxm%
  }
\newcommand{\argauxm}[4]
  {\argext{\argauxl{#1}{#2}{#3}}{, \allowbreak}{#4}}

\newcommand{\argn}[9]
  {%
  \def\valarga{#1}%
  \def\valargb{#2}%
  \def\valargc{#3}%
  \def\valargd{#4}%
  \def\valarge{#5}%
  \def\valargf{#6}%
  \def\valargg{#7}%
  \def\valargh{#8}%
  \def\valargi{#9}%
  \argauxn%
  }
\newcommand{\argauxn}[5]
  {\argext{\argauxm{#1}{#2}{#3}{#4}}{, \allowbreak}{#5}}

\newcommand{\argo}[9]
  {%
  \def\valarga{#1}%
  \def\valargb{#2}%
  \def\valargc{#3}%
  \def\valargd{#4}%
  \def\valarge{#5}%
  \def\valargf{#6}%
  \def\valargg{#7}%
  \def\valargh{#8}%
  \def\valargi{#9}%
  \argauxo%
  }
\newcommand{\argauxo}[6]
  {\argext{\argauxn{#1}{#2}{#3}{#4}{#5}}{, \allowbreak}{#6}}

\newcommand{\argp}[9]
  {%
  \def\valarga{#1}%
  \def\valargb{#2}%
  \def\valargc{#3}%
  \def\valargd{#4}%
  \def\valarge{#5}%
  \def\valargf{#6}%
  \def\valargg{#7}%
  \def\valargh{#8}%
  \def\valargi{#9}%
  \argauxp%
  }
\newcommand{\argauxp}[7]
  {\argext{\argauxo{#1}{#2}{#3}{#4}{#5}{#6}}{, \allowbreak}{#7}}

\newcommand{\argq}[9]
  {%
  \def\valarga{#1}%
  \def\valargb{#2}%
  \def\valargc{#3}%
  \def\valargd{#4}%
  \def\valarge{#5}%
  \def\valargf{#6}%
  \def\valargg{#7}%
  \def\valargh{#8}%
  \def\valargi{#9}%
  \argauxq%
  }
\newcommand{\argauxq}[8]
  {\argext{\argauxp{#1}{#2}{#3}{#4}{#5}{#6}{#7}}{, \allowbreak}{#8}}

\newcommand{\argr}[9]
  {%
  \def\valarga{#1}%
  \def\valargb{#2}%
  \def\valargc{#3}%
  \def\valargd{#4}%
  \def\valarge{#5}%
  \def\valargf{#6}%
  \def\valargg{#7}%
  \def\valargh{#8}%
  \def\valargi{#9}%
  \argauxr%
  }
\newcommand{\argauxr}[9]
  {\argext{\argauxq{#1}{#2}{#3}{#4}{#5}{#6}{#7}{#8}}{, \allowbreak}{#9}}

\newcommand{\txtfnt}[2][]
  {{%
  \IfStrEq{#1}{}
    {#2}
    {%
    \StrLeft{#1}{2}[\optbgn]%
    \StrGobbleLeft{#1}{2}[\optend]%
    \IfStrEqCase{\optbgn}
      {%
      {Rm}{\rmfamily\txtfnt[\optend]{#2}}%
      {Sf}{\sffamily\txtfnt[\optend]{#2}}%
      {Tt}{\ttfamily\txtfnt[\optend]{#2}}%
      {Up}{\upshape\txtfnt[\optend]{#2}}%
      {It}{\itshape\txtfnt[\optend]{#2}}%
      {Sl}{\slshape\txtfnt[\optend]{#2}}%
      {Sc}{\scshape\txtfnt[\optend]{#2}}%
      {Md}{\mdseries\txtfnt[\optend]{#2}}%
      {Bf}{\bfseries\txtfnt[\optend]{#2}}%
      {Em}{\emph{\txtfnt[\optend]{#2}}}%
      }
      [\ensuremath{\clubsuit}]%
    }%
  }}

\newcommand{\txtsub}[2][]
  {\argemp{#2}{\ensuremath{_{\text{\txtfnt[#1]{#2}}}}}}

\newcommand{\txtsup}[2][]
  {\argemp{#2}{\ensuremath{^{\text{\txtfnt[#1]{#2}}}}}}

\newcommandx{\txt}[4][1=, 3=, 4=]
  {\text{\txtfnt[#1]{#2}\ensuremath{\txtsub[#1]{#3}\txtsup[#1]{#4}}}\xspace}

\newcommandx{\txtarg}[5][1=, 3=, 4=]
  {{\txt[#1]{#2}[#3][#4]\argint{(}{#5}{)}}\xspace}

\newcommand{\txtstyname}{RmScMd}
\newcommand{\txtname}[1][]
  {\txt[\argdef{#1}{\txtstyname}]}
\newcommand{\txtargname}[1][]
  {\txtarg[\argdef{#1}{\txtstyname}]}

\newcommand{\txtstyabr}{Em}
\newcommand{\txtabr}[1][]
  {\txt[\argdef{#1}{\txtstyabr}]}

\newcommandx{\mthfnt}[3][1=, 2=0]
  {{%
  \IfStrEqCase{#1}
    {%
    {}%
      {#3}%
    {Name}%
      {%
      \IfStrEqCase{#2}
        {%
        {0}{\mathcal{#3}}%
        {1}{\mathscr{#3}}%
        {2}{\mathfrak{#3}}%
        {3}{\mathbf{#3}}%
        }
        [\ensuremath{\clubsuit}]%
      }%
    {Set}%
      {%
      \IfStrEqCase{#2}
        {%
        {0}{\mathrm{#3}}%
        {1}{\mathbb{#3}}%
        {2}{\mathsf{#3}}%
        {3}{\mathtt{#3}}%
        }
        [\ensuremath{\clubsuit}]%
      }%
    {Fun}%
      {%
      \IfStrEqCase{#2}
        {%
        {0}{\mathsf{#3}}%
        {1}{\mathrm{#3}}%
        }
        [\ensuremath{\clubsuit}]%
      }%
    {Rel}%
      {%
      \IfStrEqCase{#2}
        {%
        {0}{\mathit{#3}}%
        {1}{\mathtt{#3}}%
        }
        [\ensuremath{\clubsuit}]%
      }%
    {Sym}%
      {%
      \IfStrEqCase{#2}
        {%
        {0}{\mathtt{#3}}%
        {1}{\mathbf{#3}}%
        }
        [\ensuremath{\clubsuit}]%
      }%
    {Elm}%
      {\mathnormal{#3}}
    }
    [\ensuremath{\clubsuit}]%
  }}

\newcommand{\mthsub}[1]
  {\argemp{#1}{\ensuremath{_{\mathnormal{#1}}}}}

\newcommand{\mthsup}[1]
  {\argemp{#1}{\ensuremath{^{\mathnormal{#1}}}}}

\newcommandx{\mth}[5][1=, 2=0, 4=, 5=]
  {{\ensuremath{\mthfnt[#1][#2]{#3}\mthsub{#4}\mthsup{#5}}}}

\newcommandx{\mtharg}[6][1=, 2=0, 4=, 5=]
  {{\mth[#1][#2]{#3}[#4][#5]\ensuremath{\argint{\!\left(}{#6}{\right)}}}}

\newcommand{\mthempty}
  {\mth[][]}

\newcommand{\mthstyname}{0}
\newcommand{\mthname}[1][]
  {\mth[Name][\argdef{#1}{\mthstyname}]}

\newcommand{\mthstyset}{0}
\newcommand{\mthset}[1][]
  {\mth[Set][\argdef{#1}{\mthstyset}]}
\newcommand{\mthargset}[1][]
  {\mtharg[Set][\argdef{#1}{\mthstyset}]}

\newcommand{\mthstyfun}{0}
\newcommand{\mthfun}[1][]
  {\mth[Fun][\argdef{#1}{\mthstyfun}]}
\newcommand{\mthargfun}[1][]
  {\mtharg[Fun][\argdef{#1}{\mthstyfun}]}

\newcommand{\mthstyrel}{0}
\newcommand{\mthrel}[1][]
  {\mth[Rel][\argdef{#1}{\mthstyrel}]}

\newcommand{\mthstysym}{0}
\newcommand{\mthsym}[1][]
  {\mth[Sym][\argdef{#1}{\mthstysym}]}

\newcommand{\mthstyelm}{0}
\newcommand{\mthelm}[1][]
  {\mth[Elm][\argdef{#1}{\mthstyelm}]}

\newcommandx{\AName}[4][1=, 2=, 3=, 4=]{\mthname[#4]{A#3}[#1][#2]}
\newcommandx{\BName}[4][1=, 2=, 3=, 4=]{\mthname[#4]{B#3}[#1][#2]}
\newcommandx{\CName}[4][1=, 2=, 3=, 4=]{\mthname[#4]{C#3}[#1][#2]}
\newcommandx{\DName}[4][1=, 2=, 3=, 4=]{\mthname[#4]{D#3}[#1][#2]}
\newcommandx{\EName}[4][1=, 2=, 3=, 4=]{\mthname[#4]{E#3}[#1][#2]}
\newcommandx{\FName}[4][1=, 2=, 3=, 4=]{\mthname[#4]{F#3}[#1][#2]}
\newcommandx{\GName}[4][1=, 2=, 3=, 4=]{\mthname[#4]{G#3}[#1][#2]}
\newcommandx{\HName}[4][1=, 2=, 3=, 4=]{\mthname[#4]{H#3}[#1][#2]}
\newcommandx{\IName}[4][1=, 2=, 3=, 4=]{\mthname[#4]{I#3}[#1][#2]}
\newcommandx{\JName}[4][1=, 2=, 3=, 4=]{\mthname[#4]{J#3}[#1][#2]}
\newcommandx{\KName}[4][1=, 2=, 3=, 4=]{\mthname[#4]{K#3}[#1][#2]}
\newcommandx{\LName}[4][1=, 2=, 3=, 4=]{\mthname[#4]{L#3}[#1][#2]}
\newcommandx{\MName}[4][1=, 2=, 3=, 4=]{\mthname[#4]{M#3}[#1][#2]}
\newcommandx{\NName}[4][1=, 2=, 3=, 4=]{\mthname[#4]{N#3}[#1][#2]}
\newcommandx{\OName}[4][1=, 2=, 3=, 4=]{\mthname[#4]{O#3}[#1][#2]}
\newcommandx{\PName}[4][1=, 2=, 3=, 4=]{\mthname[#4]{P#3}[#1][#2]}
\newcommandx{\QName}[4][1=, 2=, 3=, 4=]{\mthname[#4]{Q#3}[#1][#2]}
\newcommandx{\RName}[4][1=, 2=, 3=, 4=]{\mthname[#4]{R#3}[#1][#2]}
\newcommandx{\SName}[4][1=, 2=, 3=, 4=]{\mthname[#4]{S#3}[#1][#2]}
\newcommandx{\TName}[4][1=, 2=, 3=, 4=]{\mthname[#4]{T#3}[#1][#2]}
\newcommandx{\UName}[4][1=, 2=, 3=, 4=]{\mthname[#4]{U#3}[#1][#2]}
\newcommandx{\VName}[4][1=, 2=, 3=, 4=]{\mthname[#4]{V#3}[#1][#2]}
\newcommandx{\WName}[4][1=, 2=, 3=, 4=]{\mthname[#4]{W#3}[#1][#2]}
\newcommandx{\XName}[4][1=, 2=, 3=, 4=]{\mthname[#4]{X#3}[#1][#2]}
\newcommandx{\YName}[4][1=, 2=, 3=, 4=]{\mthname[#4]{Y#3}[#1][#2]}
\newcommandx{\ZName}[4][1=, 2=, 3=, 4=]{\mthname[#4]{Z#3}[#1][#2]}

\newcommandx{\ASet}[4][1=, 2=, 3=, 4=]{\mthset[#4]{A#3}[#1][#2]}
\newcommandx{\BSet}[4][1=, 2=, 3=, 4=]{\mthset[#4]{B#3}[#1][#2]}
\newcommandx{\CSet}[4][1=, 2=, 3=, 4=]{\mthset[#4]{C#3}[#1][#2]}
\newcommandx{\DSet}[4][1=, 2=, 3=, 4=]{\mthset[#4]{D#3}[#1][#2]}
\newcommandx{\ESet}[4][1=, 2=, 3=, 4=]{\mthset[#4]{E#3}[#1][#2]}
\newcommandx{\FSet}[4][1=, 2=, 3=, 4=]{\mthset[#4]{F#3}[#1][#2]}
\newcommandx{\GSet}[4][1=, 2=, 3=, 4=]{\mthset[#4]{G#3}[#1][#2]}
\newcommandx{\HSet}[4][1=, 2=, 3=, 4=]{\mthset[#4]{H#3}[#1][#2]}
\newcommandx{\ISet}[4][1=, 2=, 3=, 4=]{\mthset[#4]{I#3}[#1][#2]}
\newcommandx{\JSet}[4][1=, 2=, 3=, 4=]{\mthset[#4]{J#3}[#1][#2]}
\newcommandx{\KSet}[4][1=, 2=, 3=, 4=]{\mthset[#4]{K#3}[#1][#2]}
\newcommandx{\LSet}[4][1=, 2=, 3=, 4=]{\mthset[#4]{L#3}[#1][#2]}
\newcommandx{\MSet}[4][1=, 2=, 3=, 4=]{\mthset[#4]{M#3}[#1][#2]}
\newcommandx{\NSet}[4][1=, 2=, 3=, 4=]{\mthset[#4]{N#3}[#1][#2]}
\newcommandx{\OSet}[4][1=, 2=, 3=, 4=]{\mthset[#4]{O#3}[#1][#2]}
\newcommandx{\PSet}[4][1=, 2=, 3=, 4=]{\mthset[#4]{P#3}[#1][#2]}
\newcommandx{\QSet}[4][1=, 2=, 3=, 4=]{\mthset[#4]{Q#3}[#1][#2]}
\newcommandx{\RSet}[4][1=, 2=, 3=, 4=]{\mthset[#4]{R#3}[#1][#2]}
\newcommandx{\SSet}[4][1=, 2=, 3=, 4=]{\mthset[#4]{S#3}[#1][#2]}
\newcommandx{\TSet}[4][1=, 2=, 3=, 4=]{\mthset[#4]{T#3}[#1][#2]}
\newcommandx{\USet}[4][1=, 2=, 3=, 4=]{\mthset[#4]{U#3}[#1][#2]}
\newcommandx{\VSet}[4][1=, 2=, 3=, 4=]{\mthset[#4]{V#3}[#1][#2]}
\newcommandx{\WSet}[4][1=, 2=, 3=, 4=]{\mthset[#4]{W#3}[#1][#2]}
\newcommandx{\XSet}[4][1=, 2=, 3=, 4=]{\mthset[#4]{X#3}[#1][#2]}
\newcommandx{\YSet}[4][1=, 2=, 3=, 4=]{\mthset[#4]{Y#3}[#1][#2]}
\newcommandx{\ZSet}[4][1=, 2=, 3=, 4=]{\mthset[#4]{Z#3}[#1][#2]}

\newcommandx{\aSet}[4][1=, 2=, 3=, 4=]{\mthset[#4]{a#3}[#1][#2]}
\newcommandx{\bSet}[4][1=, 2=, 3=, 4=]{\mthset[#4]{b#3}[#1][#2]}
\newcommandx{\cSet}[4][1=, 2=, 3=, 4=]{\mthset[#4]{c#3}[#1][#2]}
\newcommandx{\dSet}[4][1=, 2=, 3=, 4=]{\mthset[#4]{d#3}[#1][#2]}
\newcommandx{\eSet}[4][1=, 2=, 3=, 4=]{\mthset[#4]{e#3}[#1][#2]}
\newcommandx{\fSet}[4][1=, 2=, 3=, 4=]{\mthset[#4]{f#3}[#1][#2]}
\newcommandx{\gSet}[4][1=, 2=, 3=, 4=]{\mthset[#4]{g#3}[#1][#2]}
\newcommandx{\hSet}[4][1=, 2=, 3=, 4=]{\mthset[#4]{h#3}[#1][#2]}
\newcommandx{\iSet}[4][1=, 2=, 3=, 4=]{\mthset[#4]{i#3}[#1][#2]}
\newcommandx{\jSet}[4][1=, 2=, 3=, 4=]{\mthset[#4]{j#3}[#1][#2]}
\newcommandx{\kSet}[4][1=, 2=, 3=, 4=]{\mthset[#4]{k#3}[#1][#2]}
\newcommandx{\lSet}[4][1=, 2=, 3=, 4=]{\mthset[#4]{l#3}[#1][#2]}
\newcommandx{\mSet}[4][1=, 2=, 3=, 4=]{\mthset[#4]{m#3}[#1][#2]}
\newcommandx{\nSet}[4][1=, 2=, 3=, 4=]{\mthset[#4]{n#3}[#1][#2]}
\newcommandx{\oSet}[4][1=, 2=, 3=, 4=]{\mthset[#4]{o#3}[#1][#2]}
\newcommandx{\pSet}[4][1=, 2=, 3=, 4=]{\mthset[#4]{p#3}[#1][#2]}
\newcommandx{\qSet}[4][1=, 2=, 3=, 4=]{\mthset[#4]{q#3}[#1][#2]}
\newcommandx{\rSet}[4][1=, 2=, 3=, 4=]{\mthset[#4]{r#3}[#1][#2]}
\newcommandx{\sSet}[4][1=, 2=, 3=, 4=]{\mthset[#4]{s#3}[#1][#2]}
\newcommandx{\tSet}[4][1=, 2=, 3=, 4=]{\mthset[#4]{t#3}[#1][#2]}
\newcommandx{\uSet}[4][1=, 2=, 3=, 4=]{\mthset[#4]{u#3}[#1][#2]}
\newcommandx{\vSet}[4][1=, 2=, 3=, 4=]{\mthset[#4]{v#3}[#1][#2]}
\newcommandx{\wSet}[4][1=, 2=, 3=, 4=]{\mthset[#4]{w#3}[#1][#2]}
\newcommandx{\xSet}[4][1=, 2=, 3=, 4=]{\mthset[#4]{x#3}[#1][#2]}
\newcommandx{\ySet}[4][1=, 2=, 3=, 4=]{\mthset[#4]{y#3}[#1][#2]}
\newcommandx{\zSet}[4][1=, 2=, 3=, 4=]{\mthset[#4]{z#3}[#1][#2]}

\newcommandx{\AFun}[4][1=, 2=, 3=, 4=]{\mthfun[#4]{A#3}[#1][#2]}
\newcommandx{\BFun}[4][1=, 2=, 3=, 4=]{\mthfun[#4]{B#3}[#1][#2]}
\newcommandx{\CFun}[4][1=, 2=, 3=, 4=]{\mthfun[#4]{C#3}[#1][#2]}
\newcommandx{\DFun}[4][1=, 2=, 3=, 4=]{\mthfun[#4]{D#3}[#1][#2]}
\newcommandx{\EFun}[4][1=, 2=, 3=, 4=]{\mthfun[#4]{E#3}[#1][#2]}
\newcommandx{\FFun}[4][1=, 2=, 3=, 4=]{\mthfun[#4]{F#3}[#1][#2]}
\newcommandx{\GFun}[4][1=, 2=, 3=, 4=]{\mthfun[#4]{G#3}[#1][#2]}
\newcommandx{\HFun}[4][1=, 2=, 3=, 4=]{\mthfun[#4]{H#3}[#1][#2]}
\newcommandx{\IFun}[4][1=, 2=, 3=, 4=]{\mthfun[#4]{I#3}[#1][#2]}
\newcommandx{\JFun}[4][1=, 2=, 3=, 4=]{\mthfun[#4]{J#3}[#1][#2]}
\newcommandx{\KFun}[4][1=, 2=, 3=, 4=]{\mthfun[#4]{K#3}[#1][#2]}
\newcommandx{\LFun}[4][1=, 2=, 3=, 4=]{\mthfun[#4]{L#3}[#1][#2]}
\newcommandx{\MFun}[4][1=, 2=, 3=, 4=]{\mthfun[#4]{M#3}[#1][#2]}
\newcommandx{\NFun}[4][1=, 2=, 3=, 4=]{\mthfun[#4]{N#3}[#1][#2]}
\newcommandx{\OFun}[4][1=, 2=, 3=, 4=]{\mthfun[#4]{O#3}[#1][#2]}
\newcommandx{\PFun}[4][1=, 2=, 3=, 4=]{\mthfun[#4]{P#3}[#1][#2]}
\newcommandx{\QFun}[4][1=, 2=, 3=, 4=]{\mthfun[#4]{Q#3}[#1][#2]}
\newcommandx{\RFun}[4][1=, 2=, 3=, 4=]{\mthfun[#4]{R#3}[#1][#2]}
\newcommandx{\SFun}[4][1=, 2=, 3=, 4=]{\mthfun[#4]{S#3}[#1][#2]}
\newcommandx{\TFun}[4][1=, 2=, 3=, 4=]{\mthfun[#4]{T#3}[#1][#2]}
\newcommandx{\UFun}[4][1=, 2=, 3=, 4=]{\mthfun[#4]{U#3}[#1][#2]}
\newcommandx{\VFun}[4][1=, 2=, 3=, 4=]{\mthfun[#4]{V#3}[#1][#2]}
\newcommandx{\WFun}[4][1=, 2=, 3=, 4=]{\mthfun[#4]{W#3}[#1][#2]}
\newcommandx{\XFun}[4][1=, 2=, 3=, 4=]{\mthfun[#4]{X#3}[#1][#2]}
\newcommandx{\YFun}[4][1=, 2=, 3=, 4=]{\mthfun[#4]{Y#3}[#1][#2]}
\newcommandx{\ZFun}[4][1=, 2=, 3=, 4=]{\mthfun[#4]{Z#3}[#1][#2]}

\newcommandx{\aFun}[4][1=, 2=, 3=, 4=]{\mthfun[#4]{a#3}[#1][#2]}
\newcommandx{\bFun}[4][1=, 2=, 3=, 4=]{\mthfun[#4]{b#3}[#1][#2]}
\newcommandx{\cFun}[4][1=, 2=, 3=, 4=]{\mthfun[#4]{c#3}[#1][#2]}
\newcommandx{\dFun}[4][1=, 2=, 3=, 4=]{\mthfun[#4]{d#3}[#1][#2]}
\newcommandx{\eFun}[4][1=, 2=, 3=, 4=]{\mthfun[#4]{e#3}[#1][#2]}
\newcommandx{\fFun}[4][1=, 2=, 3=, 4=]{\mthfun[#4]{f#3}[#1][#2]}
\newcommandx{\gFun}[4][1=, 2=, 3=, 4=]{\mthfun[#4]{g#3}[#1][#2]}
\newcommandx{\hFun}[4][1=, 2=, 3=, 4=]{\mthfun[#4]{h#3}[#1][#2]}
\newcommandx{\iFun}[4][1=, 2=, 3=, 4=]{\mthfun[#4]{i#3}[#1][#2]}
\newcommandx{\jFun}[4][1=, 2=, 3=, 4=]{\mthfun[#4]{j#3}[#1][#2]}
\newcommandx{\kFun}[4][1=, 2=, 3=, 4=]{\mthfun[#4]{k#3}[#1][#2]}
\newcommandx{\lFun}[4][1=, 2=, 3=, 4=]{\mthfun[#4]{l#3}[#1][#2]}
\newcommandx{\mFun}[4][1=, 2=, 3=, 4=]{\mthfun[#4]{m#3}[#1][#2]}
\newcommandx{\nFun}[4][1=, 2=, 3=, 4=]{\mthfun[#4]{n#3}[#1][#2]}
\newcommandx{\oFun}[4][1=, 2=, 3=, 4=]{\mthfun[#4]{o#3}[#1][#2]}
\newcommandx{\pFun}[4][1=, 2=, 3=, 4=]{\mthfun[#4]{p#3}[#1][#2]}
\newcommandx{\qFun}[4][1=, 2=, 3=, 4=]{\mthfun[#4]{q#3}[#1][#2]}
\newcommandx{\rFun}[4][1=, 2=, 3=, 4=]{\mthfun[#4]{r#3}[#1][#2]}
\newcommandx{\sFun}[4][1=, 2=, 3=, 4=]{\mthfun[#4]{s#3}[#1][#2]}
\newcommandx{\tFun}[4][1=, 2=, 3=, 4=]{\mthfun[#4]{t#3}[#1][#2]}
\newcommandx{\uFun}[4][1=, 2=, 3=, 4=]{\mthfun[#4]{u#3}[#1][#2]}
\newcommandx{\vFun}[4][1=, 2=, 3=, 4=]{\mthfun[#4]{v#3}[#1][#2]}
\newcommandx{\wFun}[4][1=, 2=, 3=, 4=]{\mthfun[#4]{w#3}[#1][#2]}
\newcommandx{\xFun}[4][1=, 2=, 3=, 4=]{\mthfun[#4]{x#3}[#1][#2]}
\newcommandx{\yFun}[4][1=, 2=, 3=, 4=]{\mthfun[#4]{y#3}[#1][#2]}
\newcommandx{\zFun}[4][1=, 2=, 3=, 4=]{\mthfun[#4]{z#3}[#1][#2]}

\newcommandx{\ARel}[4][1=, 2=, 3=, 4=]{\mthrel[#4]{A#3}[#1][#2]}
\newcommandx{\BRel}[4][1=, 2=, 3=, 4=]{\mthrel[#4]{B#3}[#1][#2]}
\newcommandx{\CRel}[4][1=, 2=, 3=, 4=]{\mthrel[#4]{C#3}[#1][#2]}
\newcommandx{\DRel}[4][1=, 2=, 3=, 4=]{\mthrel[#4]{D#3}[#1][#2]}
\newcommandx{\ERel}[4][1=, 2=, 3=, 4=]{\mthrel[#4]{E#3}[#1][#2]}
\newcommandx{\FRel}[4][1=, 2=, 3=, 4=]{\mthrel[#4]{F#3}[#1][#2]}
\newcommandx{\GRel}[4][1=, 2=, 3=, 4=]{\mthrel[#4]{G#3}[#1][#2]}
\newcommandx{\HRel}[4][1=, 2=, 3=, 4=]{\mthrel[#4]{H#3}[#1][#2]}
\newcommandx{\IRel}[4][1=, 2=, 3=, 4=]{\mthrel[#4]{I#3}[#1][#2]}
\newcommandx{\JRel}[4][1=, 2=, 3=, 4=]{\mthrel[#4]{J#3}[#1][#2]}
\newcommandx{\KRel}[4][1=, 2=, 3=, 4=]{\mthrel[#4]{K#3}[#1][#2]}
\newcommandx{\LRel}[4][1=, 2=, 3=, 4=]{\mthrel[#4]{L#3}[#1][#2]}
\newcommandx{\MRel}[4][1=, 2=, 3=, 4=]{\mthrel[#4]{M#3}[#1][#2]}
\newcommandx{\NRel}[4][1=, 2=, 3=, 4=]{\mthrel[#4]{N#3}[#1][#2]}
\newcommandx{\ORel}[4][1=, 2=, 3=, 4=]{\mthrel[#4]{O#3}[#1][#2]}
\newcommandx{\PRel}[4][1=, 2=, 3=, 4=]{\mthrel[#4]{P#3}[#1][#2]}
\newcommandx{\QRel}[4][1=, 2=, 3=, 4=]{\mthrel[#4]{Q#3}[#1][#2]}
\newcommandx{\RRel}[4][1=, 2=, 3=, 4=]{\mthrel[#4]{R#3}[#1][#2]}
\newcommandx{\SRel}[4][1=, 2=, 3=, 4=]{\mthrel[#4]{S#3}[#1][#2]}
\newcommandx{\TRel}[4][1=, 2=, 3=, 4=]{\mthrel[#4]{T#3}[#1][#2]}
\newcommandx{\URel}[4][1=, 2=, 3=, 4=]{\mthrel[#4]{U#3}[#1][#2]}
\newcommandx{\VRel}[4][1=, 2=, 3=, 4=]{\mthrel[#4]{V#3}[#1][#2]}
\newcommandx{\WRel}[4][1=, 2=, 3=, 4=]{\mthrel[#4]{W#3}[#1][#2]}
\newcommandx{\XRel}[4][1=, 2=, 3=, 4=]{\mthrel[#4]{X#3}[#1][#2]}
\newcommandx{\YRel}[4][1=, 2=, 3=, 4=]{\mthrel[#4]{Y#3}[#1][#2]}
\newcommandx{\ZRel}[4][1=, 2=, 3=, 4=]{\mthrel[#4]{Z#3}[#1][#2]}

\newcommandx{\aRel}[4][1=, 2=, 3=, 4=]{\mthrel[#4]{a#3}[#1][#2]}
\newcommandx{\bRel}[4][1=, 2=, 3=, 4=]{\mthrel[#4]{b#3}[#1][#2]}
\newcommandx{\cRel}[4][1=, 2=, 3=, 4=]{\mthrel[#4]{c#3}[#1][#2]}
\newcommandx{\dRel}[4][1=, 2=, 3=, 4=]{\mthrel[#4]{d#3}[#1][#2]}
\newcommandx{\eRel}[4][1=, 2=, 3=, 4=]{\mthrel[#4]{e#3}[#1][#2]}
\newcommandx{\fRel}[4][1=, 2=, 3=, 4=]{\mthrel[#4]{f#3}[#1][#2]}
\newcommandx{\gRel}[4][1=, 2=, 3=, 4=]{\mthrel[#4]{g#3}[#1][#2]}
\newcommandx{\hRel}[4][1=, 2=, 3=, 4=]{\mthrel[#4]{h#3}[#1][#2]}
\newcommandx{\iRel}[4][1=, 2=, 3=, 4=]{\mthrel[#4]{i#3}[#1][#2]}
\newcommandx{\jRel}[4][1=, 2=, 3=, 4=]{\mthrel[#4]{j#3}[#1][#2]}
\newcommandx{\kRel}[4][1=, 2=, 3=, 4=]{\mthrel[#4]{k#3}[#1][#2]}
\newcommandx{\lRel}[4][1=, 2=, 3=, 4=]{\mthrel[#4]{l#3}[#1][#2]}
\newcommandx{\mRel}[4][1=, 2=, 3=, 4=]{\mthrel[#4]{m#3}[#1][#2]}
\newcommandx{\nRel}[4][1=, 2=, 3=, 4=]{\mthrel[#4]{n#3}[#1][#2]}
\newcommandx{\oRel}[4][1=, 2=, 3=, 4=]{\mthrel[#4]{o#3}[#1][#2]}
\newcommandx{\pRel}[4][1=, 2=, 3=, 4=]{\mthrel[#4]{p#3}[#1][#2]}
\newcommandx{\qRel}[4][1=, 2=, 3=, 4=]{\mthrel[#4]{q#3}[#1][#2]}
\newcommandx{\rRel}[4][1=, 2=, 3=, 4=]{\mthrel[#4]{r#3}[#1][#2]}
\newcommandx{\sRel}[4][1=, 2=, 3=, 4=]{\mthrel[#4]{s#3}[#1][#2]}
\newcommandx{\tRel}[4][1=, 2=, 3=, 4=]{\mthrel[#4]{t#3}[#1][#2]}
\newcommandx{\uRel}[4][1=, 2=, 3=, 4=]{\mthrel[#4]{u#3}[#1][#2]}
\newcommandx{\vRel}[4][1=, 2=, 3=, 4=]{\mthrel[#4]{v#3}[#1][#2]}
\newcommandx{\wRel}[4][1=, 2=, 3=, 4=]{\mthrel[#4]{w#3}[#1][#2]}
\newcommandx{\xRel}[4][1=, 2=, 3=, 4=]{\mthrel[#4]{x#3}[#1][#2]}
\newcommandx{\yRel}[4][1=, 2=, 3=, 4=]{\mthrel[#4]{y#3}[#1][#2]}
\newcommandx{\zRel}[4][1=, 2=, 3=, 4=]{\mthrel[#4]{z#3}[#1][#2]}

\newcommandx{\ASym}[4][1=, 2=, 3=, 4=]{\mthsym[#4]{A#3}[#1][#2]}
\newcommandx{\BSym}[4][1=, 2=, 3=, 4=]{\mthsym[#4]{B#3}[#1][#2]}
\newcommandx{\CSym}[4][1=, 2=, 3=, 4=]{\mthsym[#4]{C#3}[#1][#2]}
\newcommandx{\DSym}[4][1=, 2=, 3=, 4=]{\mthsym[#4]{D#3}[#1][#2]}
\newcommandx{\ESym}[4][1=, 2=, 3=, 4=]{\mthsym[#4]{E#3}[#1][#2]}
\newcommandx{\FSym}[4][1=, 2=, 3=, 4=]{\mthsym[#4]{F#3}[#1][#2]}
\newcommandx{\GSym}[4][1=, 2=, 3=, 4=]{\mthsym[#4]{G#3}[#1][#2]}
\newcommandx{\HSym}[4][1=, 2=, 3=, 4=]{\mthsym[#4]{H#3}[#1][#2]}
\newcommandx{\ISym}[4][1=, 2=, 3=, 4=]{\mthsym[#4]{I#3}[#1][#2]}
\newcommandx{\JSym}[4][1=, 2=, 3=, 4=]{\mthsym[#4]{J#3}[#1][#2]}
\newcommandx{\KSym}[4][1=, 2=, 3=, 4=]{\mthsym[#4]{K#3}[#1][#2]}
\newcommandx{\LSym}[4][1=, 2=, 3=, 4=]{\mthsym[#4]{L#3}[#1][#2]}
\newcommandx{\MSym}[4][1=, 2=, 3=, 4=]{\mthsym[#4]{M#3}[#1][#2]}
\newcommandx{\NSym}[4][1=, 2=, 3=, 4=]{\mthsym[#4]{N#3}[#1][#2]}
\newcommandx{\OSym}[4][1=, 2=, 3=, 4=]{\mthsym[#4]{O#3}[#1][#2]}
\newcommandx{\PSym}[4][1=, 2=, 3=, 4=]{\mthsym[#4]{P#3}[#1][#2]}
\newcommandx{\QSym}[4][1=, 2=, 3=, 4=]{\mthsym[#4]{Q#3}[#1][#2]}
\newcommandx{\RSym}[4][1=, 2=, 3=, 4=]{\mthsym[#4]{R#3}[#1][#2]}
\newcommandx{\SSym}[4][1=, 2=, 3=, 4=]{\mthsym[#4]{S#3}[#1][#2]}
\newcommandx{\TSym}[4][1=, 2=, 3=, 4=]{\mthsym[#4]{T#3}[#1][#2]}
\newcommandx{\USym}[4][1=, 2=, 3=, 4=]{\mthsym[#4]{U#3}[#1][#2]}
\newcommandx{\VSym}[4][1=, 2=, 3=, 4=]{\mthsym[#4]{V#3}[#1][#2]}
\newcommandx{\WSym}[4][1=, 2=, 3=, 4=]{\mthsym[#4]{W#3}[#1][#2]}
\newcommandx{\XSym}[4][1=, 2=, 3=, 4=]{\mthsym[#4]{X#3}[#1][#2]}
\newcommandx{\YSym}[4][1=, 2=, 3=, 4=]{\mthsym[#4]{Y#3}[#1][#2]}
\newcommandx{\ZSym}[4][1=, 2=, 3=, 4=]{\mthsym[#4]{Z#3}[#1][#2]}

\newcommandx{\aSym}[4][1=, 2=, 3=, 4=]{\mthsym[#4]{a#3}[#1][#2]}
\newcommandx{\bSym}[4][1=, 2=, 3=, 4=]{\mthsym[#4]{b#3}[#1][#2]}
\newcommandx{\cSym}[4][1=, 2=, 3=, 4=]{\mthsym[#4]{c#3}[#1][#2]}
\newcommandx{\dSym}[4][1=, 2=, 3=, 4=]{\mthsym[#4]{d#3}[#1][#2]}
\newcommandx{\eSym}[4][1=, 2=, 3=, 4=]{\mthsym[#4]{e#3}[#1][#2]}
\newcommandx{\fSym}[4][1=, 2=, 3=, 4=]{\mthsym[#4]{f#3}[#1][#2]}
\newcommandx{\gSym}[4][1=, 2=, 3=, 4=]{\mthsym[#4]{g#3}[#1][#2]}
\newcommandx{\hSym}[4][1=, 2=, 3=, 4=]{\mthsym[#4]{h#3}[#1][#2]}
\newcommandx{\iSym}[4][1=, 2=, 3=, 4=]{\mthsym[#4]{i#3}[#1][#2]}
\newcommandx{\jSym}[4][1=, 2=, 3=, 4=]{\mthsym[#4]{j#3}[#1][#2]}
\newcommandx{\kSym}[4][1=, 2=, 3=, 4=]{\mthsym[#4]{k#3}[#1][#2]}
\newcommandx{\lSym}[4][1=, 2=, 3=, 4=]{\mthsym[#4]{l#3}[#1][#2]}
\newcommandx{\mSym}[4][1=, 2=, 3=, 4=]{\mthsym[#4]{m#3}[#1][#2]}
\newcommandx{\nSym}[4][1=, 2=, 3=, 4=]{\mthsym[#4]{n#3}[#1][#2]}
\newcommandx{\oSym}[4][1=, 2=, 3=, 4=]{\mthsym[#4]{o#3}[#1][#2]}
\newcommandx{\pSym}[4][1=, 2=, 3=, 4=]{\mthsym[#4]{p#3}[#1][#2]}
\newcommandx{\qSym}[4][1=, 2=, 3=, 4=]{\mthsym[#4]{q#3}[#1][#2]}
\newcommandx{\rSym}[4][1=, 2=, 3=, 4=]{\mthsym[#4]{r#3}[#1][#2]}
\newcommandx{\sSym}[4][1=, 2=, 3=, 4=]{\mthsym[#4]{s#3}[#1][#2]}
\newcommandx{\tSym}[4][1=, 2=, 3=, 4=]{\mthsym[#4]{t#3}[#1][#2]}
\newcommandx{\uSym}[4][1=, 2=, 3=, 4=]{\mthsym[#4]{u#3}[#1][#2]}
\newcommandx{\vSym}[4][1=, 2=, 3=, 4=]{\mthsym[#4]{v#3}[#1][#2]}
\newcommandx{\wSym}[4][1=, 2=, 3=, 4=]{\mthsym[#4]{w#3}[#1][#2]}
\newcommandx{\xSym}[4][1=, 2=, 3=, 4=]{\mthsym[#4]{x#3}[#1][#2]}
\newcommandx{\ySym}[4][1=, 2=, 3=, 4=]{\mthsym[#4]{y#3}[#1][#2]}
\newcommandx{\zSym}[4][1=, 2=, 3=, 4=]{\mthsym[#4]{z#3}[#1][#2]}

\newcommandx{\AElm}[4][1=, 2=, 3=, 4=]{\mthelm[#4]{A#3}[#1][#2]}
\newcommandx{\BElm}[4][1=, 2=, 3=, 4=]{\mthelm[#4]{B#3}[#1][#2]}
\newcommandx{\CElm}[4][1=, 2=, 3=, 4=]{\mthelm[#4]{C#3}[#1][#2]}
\newcommandx{\DElm}[4][1=, 2=, 3=, 4=]{\mthelm[#4]{D#3}[#1][#2]}
\newcommandx{\EElm}[4][1=, 2=, 3=, 4=]{\mthelm[#4]{E#3}[#1][#2]}
\newcommandx{\FElm}[4][1=, 2=, 3=, 4=]{\mthelm[#4]{F#3}[#1][#2]}
\newcommandx{\GElm}[4][1=, 2=, 3=, 4=]{\mthelm[#4]{G#3}[#1][#2]}
\newcommandx{\HElm}[4][1=, 2=, 3=, 4=]{\mthelm[#4]{H#3}[#1][#2]}
\newcommandx{\IElm}[4][1=, 2=, 3=, 4=]{\mthelm[#4]{I#3}[#1][#2]}
\newcommandx{\JElm}[4][1=, 2=, 3=, 4=]{\mthelm[#4]{J#3}[#1][#2]}
\newcommandx{\KElm}[4][1=, 2=, 3=, 4=]{\mthelm[#4]{K#3}[#1][#2]}
\newcommandx{\LElm}[4][1=, 2=, 3=, 4=]{\mthelm[#4]{L#3}[#1][#2]}
\newcommandx{\MElm}[4][1=, 2=, 3=, 4=]{\mthelm[#4]{M#3}[#1][#2]}
\newcommandx{\NElm}[4][1=, 2=, 3=, 4=]{\mthelm[#4]{N#3}[#1][#2]}
\newcommandx{\OElm}[4][1=, 2=, 3=, 4=]{\mthelm[#4]{O#3}[#1][#2]}
\newcommandx{\PElm}[4][1=, 2=, 3=, 4=]{\mthelm[#4]{P#3}[#1][#2]}
\newcommandx{\QElm}[4][1=, 2=, 3=, 4=]{\mthelm[#4]{Q#3}[#1][#2]}
\newcommandx{\RElm}[4][1=, 2=, 3=, 4=]{\mthelm[#4]{R#3}[#1][#2]}
\newcommandx{\SElm}[4][1=, 2=, 3=, 4=]{\mthelm[#4]{S#3}[#1][#2]}
\newcommandx{\TElm}[4][1=, 2=, 3=, 4=]{\mthelm[#4]{T#3}[#1][#2]}
\newcommandx{\UElm}[4][1=, 2=, 3=, 4=]{\mthelm[#4]{U#3}[#1][#2]}
\newcommandx{\VElm}[4][1=, 2=, 3=, 4=]{\mthelm[#4]{V#3}[#1][#2]}
\newcommandx{\WElm}[4][1=, 2=, 3=, 4=]{\mthelm[#4]{W#3}[#1][#2]}
\newcommandx{\XElm}[4][1=, 2=, 3=, 4=]{\mthelm[#4]{X#3}[#1][#2]}
\newcommandx{\YElm}[4][1=, 2=, 3=, 4=]{\mthelm[#4]{Y#3}[#1][#2]}
\newcommandx{\ZElm}[4][1=, 2=, 3=, 4=]{\mthelm[#4]{Z#3}[#1][#2]}

\newcommandx{\aElm}[4][1=, 2=, 3=, 4=]{\mthelm[#4]{a#3}[#1][#2]}
\newcommandx{\bElm}[4][1=, 2=, 3=, 4=]{\mthelm[#4]{b#3}[#1][#2]}
\newcommandx{\cElm}[4][1=, 2=, 3=, 4=]{\mthelm[#4]{c#3}[#1][#2]}
\newcommandx{\dElm}[4][1=, 2=, 3=, 4=]{\mthelm[#4]{d#3}[#1][#2]}
\newcommandx{\eElm}[4][1=, 2=, 3=, 4=]{\mthelm[#4]{e#3}[#1][#2]}
\newcommandx{\fElm}[4][1=, 2=, 3=, 4=]{\mthelm[#4]{f#3}[#1][#2]}
\newcommandx{\gElm}[4][1=, 2=, 3=, 4=]{\mthelm[#4]{g#3}[#1][#2]}
\newcommandx{\hElm}[4][1=, 2=, 3=, 4=]{\mthelm[#4]{h#3}[#1][#2]}
\newcommandx{\iElm}[4][1=, 2=, 3=, 4=]{\mthelm[#4]{i#3}[#1][#2]}
\newcommandx{\jElm}[4][1=, 2=, 3=, 4=]{\mthelm[#4]{j#3}[#1][#2]}
\newcommandx{\kElm}[4][1=, 2=, 3=, 4=]{\mthelm[#4]{k#3}[#1][#2]}
\newcommandx{\lElm}[4][1=, 2=, 3=, 4=]{\mthelm[#4]{l#3}[#1][#2]}
\newcommandx{\mElm}[4][1=, 2=, 3=, 4=]{\mthelm[#4]{m#3}[#1][#2]}
\newcommandx{\nElm}[4][1=, 2=, 3=, 4=]{\mthelm[#4]{n#3}[#1][#2]}
\newcommandx{\oElm}[4][1=, 2=, 3=, 4=]{\mthelm[#4]{o#3}[#1][#2]}
\newcommandx{\pElm}[4][1=, 2=, 3=, 4=]{\mthelm[#4]{p#3}[#1][#2]}
\newcommandx{\qElm}[4][1=, 2=, 3=, 4=]{\mthelm[#4]{q#3}[#1][#2]}
\newcommandx{\rElm}[4][1=, 2=, 3=, 4=]{\mthelm[#4]{r#3}[#1][#2]}
\newcommandx{\sElm}[4][1=, 2=, 3=, 4=]{\mthelm[#4]{s#3}[#1][#2]}
\newcommandx{\tElm}[4][1=, 2=, 3=, 4=]{\mthelm[#4]{t#3}[#1][#2]}
\newcommandx{\uElm}[4][1=, 2=, 3=, 4=]{\mthelm[#4]{u#3}[#1][#2]}
\newcommandx{\vElm}[4][1=, 2=, 3=, 4=]{\mthelm[#4]{v#3}[#1][#2]}
\newcommandx{\wElm}[4][1=, 2=, 3=, 4=]{\mthelm[#4]{w#3}[#1][#2]}
\newcommandx{\xElm}[4][1=, 2=, 3=, 4=]{\mthelm[#4]{x#3}[#1][#2]}
\newcommandx{\yElm}[4][1=, 2=, 3=, 4=]{\mthelm[#4]{y#3}[#1][#2]}
\newcommandx{\zElm}[4][1=, 2=, 3=, 4=]{\mthelm[#4]{z#3}[#1][#2]}

\newcommand{\eg}
  {\txtabr{e.g.}}

\newcommand{\ie}
  {\txtabr{i.e.}}

\newcommand{\viceversa}
  {\txtabr{vice versa}}

\newcommand{\Mutatismutandis}
  {\txtabr{Mutatis mutandis}}

\newcommand{\aka}
  {\txtabr{a.k.a.}}

\renewcommand{\iff}
  {\txtabr{iff}}

\newcommand{\resp}
  {\txtabr{resp.}}

\newcommand{\wrt}
  {\txtabr{w.r.t.}}

\newcommand{\defeq}
  {\ensuremath{\triangleq}}

\newcommand{\lst}
  {\mthargfun{lst}}

\newcommand{\dual}[1]
  {\mthempty{\overline{#1}}}

\newcommand{\tuple}[1]
  {\ensuremath{\!\argint{\langle}{#1}{\rangle}}}

\newcommand{\tupleb}[2]
  {\tuple{\argb{#1}{#2}}}
\newcommand{\tuplec}[3]
  {\tuple{\argc{#1}{#2}{#3}}}
\newcommand{\tupled}[4]
  {\tuple{\argd{#1}{#2}{#3}{#4}}}
\newcommand{\tuplee}[5]
  {\tuple{\arge{#1}{#2}{#3}{#4}{#5}}}
\newcommand{\tuplef}[6]
  {\tuple{\argf{#1}{#2}{#3}{#4}{#5}{#6}}}
\newcommand{\tupleg}[7]
  {\tuple{\argg{#1}{#2}{#3}{#4}{#5}{#6}{#7}}}
\newcommand{\tupleh}[8]
  {\tuple{\argh{#1}{#2}{#3}{#4}{#5}{#6}{#7}{#8}}}
\newcommand{\tuplei}[9]
  {\tuple{\argi{#1}{#2}{#3}{#4}{#5}{#6}{#7}{#8}{#9}}}

\newcommand{\tuplej}[9]
  {%
  \def\defarga{#1}%
  \def\defargb{#2}%
  \def\defargc{#3}%
  \def\defargd{#4}%
  \def\defarge{#5}%
  \def\defargf{#6}%
  \def\defargg{#7}%
  \def\defargh{#8}%
  \def\defargi{#9}%
  \tupleauxj%
  }
\newcommand{\tuplek}[9]
  {%
  \def\defarga{#1}%
  \def\defargb{#2}%
  \def\defargc{#3}%
  \def\defargd{#4}%
  \def\defarge{#5}%
  \def\defargf{#6}%
  \def\defargg{#7}%
  \def\defargh{#8}%
  \def\defargi{#9}%
  \tupleauxk%
  }
\newcommand{\tuplel}[9]
  {%
  \def\defarga{#1}%
  \def\defargb{#2}%
  \def\defargc{#3}%
  \def\defargd{#4}%
  \def\defarge{#5}%
  \def\defargf{#6}%
  \def\defargg{#7}%
  \def\defargh{#8}%
  \def\defargi{#9}%
  \tupleauxl%
  }
\newcommand{\tuplem}[9]
  {%
  \def\defarga{#1}%
  \def\defargb{#2}%
  \def\defargc{#3}%
  \def\defargd{#4}%
  \def\defarge{#5}%
  \def\defargf{#6}%
  \def\defargg{#7}%
  \def\defargh{#8}%
  \def\defargi{#9}%
  \tupleauxm%
  }
\newcommand{\tuplen}[9]
  {%
  \def\defarga{#1}%
  \def\defargb{#2}%
  \def\defargc{#3}%
  \def\defargd{#4}%
  \def\defarge{#5}%
  \def\defargf{#6}%
  \def\defargg{#7}%
  \def\defargh{#8}%
  \def\defargi{#9}%
  \tupleauxn%
  }
\newcommand{\tupleo}[9]
  {%
  \def\defarga{#1}%
  \def\defargb{#2}%
  \def\defargc{#3}%
  \def\defargd{#4}%
  \def\defarge{#5}%
  \def\defargf{#6}%
  \def\defargg{#7}%
  \def\defargh{#8}%
  \def\defargi{#9}%
  \tupleauxo%
  }
\newcommand{\tuplep}[9]
  {%
  \def\defarga{#1}%
  \def\defargb{#2}%
  \def\defargc{#3}%
  \def\defargd{#4}%
  \def\defarge{#5}%
  \def\defargf{#6}%
  \def\defargg{#7}%
  \def\defargh{#8}%
  \def\defargi{#9}%
  \tupleauxp%
  }
\newcommand{\tupleq}[9]
  {%
  \def\defarga{#1}%
  \def\defargb{#2}%
  \def\defargc{#3}%
  \def\defargd{#4}%
  \def\defarge{#5}%
  \def\defargf{#6}%
  \def\defargg{#7}%
  \def\defargh{#8}%
  \def\defargi{#9}%
  \tupleauxq%
  }
\newcommand{\tupler}[9]
  {%
  \def\defarga{#1}%
  \def\defargb{#2}%
  \def\defargc{#3}%
  \def\defargd{#4}%
  \def\defarge{#5}%
  \def\defargf{#6}%
  \def\defargg{#7}%
  \def\defargh{#8}%
  \def\defargi{#9}%
  \tupleauxr%
  }

\newcommand{\tupleauxj}[1]
  {%
  \tuple{\argj{\defarga}{\defargb}{\defargc}{\defargd}{\defarge}{\defargf}%
    {\defargg}{\defargh}{\defargi}{#1}}%
  }
\newcommand{\tupleauxk}[2]
  {%
  \tuple{\argk{\defarga}{\defargb}{\defargc}{\defargd}{\defarge}{\defargf}%
    {\defargg}{\defargh}{\defargi}{#1}{#2}}%
  }
\newcommand{\tupleauxl}[3]
  {%
  \tuple{\argl{\defarga}{\defargb}{\defargc}{\defargd}{\defarge}{\defargf}%
    {\defargg}{\defargh}{\defargi}{#1}{#2}{#3}}%
  }
\newcommand{\tupleauxm}[4]
  {%
  \tuple{\argm{\defarga}{\defargb}{\defargc}{\defargd}{\defarge}{\defargf}%
    {\defargg}{\defargh}{\defargi}{#1}{#2}{#3}{#4}}%
  }
\newcommand{\tupleauxn}[5]
  {%
  \tuple{\argn{\defarga}{\defargb}{\defargc}{\defargd}{\defarge}{\defargf}%
    {\defargg}{\defargh}{\defargi}{#1}{#2}{#3}{#4}{#5}}%
  }
\newcommand{\tupleauxo}[6]
  {%
  \tuple{\argo{\defarga}{\defargb}{\defargc}{\defargd}{\defarge}{\defargf}%
    {\defargg}{\defargh}{\defargi}{#1}{#2}{#3}{#4}{#5}{#6}}%
  }
\newcommand{\tupleauxp}[7]
  {%
  \tuple{\argp{\defarga}{\defargb}{\defargc}{\defargd}{\defarge}{\defargf}%
    {\defargg}{\defargh}{\defargi}{#1}{#2}{#3}{#4}{#5}{#6}{#7}}%
  }
\newcommand{\tupleauxq}[8]
  {%
  \tuple{\argq{\defarga}{\defargb}{\defargc}{\defargd}{\defarge}{\defargf}%
    {\defargg}{\defargh}{\defargi}{#1}{#2}{#3}{#4}{#5}{#6}{#7}{#8}}%
  }
\newcommand{\tupleauxr}[9]
  {%
  \tuple{\argr{\defarga}{\defargb}{\defargc}{\defargd}{\defarge}{\defargf}%
    {\defargg}{\defargh}{\defargi}{#1}{#2}{#3}{#4}{#5}{#6}{#7}{#8}{#9}}%
  }

\newcommandx{\tupleauxbx}[2][1=, 2=]
  {%
  \tupleb
    {\argdef{#1}{\defarga[\argsubscript][\argsuperscript]}}
    {\argdef{#2}{\defargb[\argsubscript][\argsuperscript]}}%
  }
\newcommandx{\tupleauxcx}[3][1=, 2=, 3=]
  {%
  \tuplec
    {\argdef{#1}{\defarga[\argsubscript][\argsuperscript]}}
    {\argdef{#2}{\defargb[\argsubscript][\argsuperscript]}}
    {\argdef{#3}{\defargc[\argsubscript][\argsuperscript]}}%
  }
\newcommandx{\tupleauxdx}[4][1=, 2=, 3=, 4=]
  {%
  \tupled
    {\argdef{#1}{\defarga[\argsubscript][\argsuperscript]}}
    {\argdef{#2}{\defargb[\argsubscript][\argsuperscript]}}
    {\argdef{#3}{\defargc[\argsubscript][\argsuperscript]}}
    {\argdef{#4}{\defargd[\argsubscript][\argsuperscript]}}%
  }
\newcommandx{\tupleauxex}[5][1=, 2=, 3=, 4=, 5=]
  {%
  \tuplee
    {\argdef{#1}{\defarga[\argsubscript][\argsuperscript]}}
    {\argdef{#2}{\defargb[\argsubscript][\argsuperscript]}}
    {\argdef{#3}{\defargc[\argsubscript][\argsuperscript]}}
    {\argdef{#4}{\defargd[\argsubscript][\argsuperscript]}}
    {\argdef{#5}{\defarge[\argsubscript][\argsuperscript]}}%
  }
\newcommandx{\tupleauxfx}[6][1=, 2=, 3=, 4=, 5=, 6=]
  {%
  \tuplef
    {\argdef{#1}{\defarga[\argsubscript][\argsuperscript]}}
    {\argdef{#2}{\defargb[\argsubscript][\argsuperscript]}}
    {\argdef{#3}{\defargc[\argsubscript][\argsuperscript]}}
    {\argdef{#4}{\defargd[\argsubscript][\argsuperscript]}}
    {\argdef{#5}{\defarge[\argsubscript][\argsuperscript]}}
    {\argdef{#6}{\defargf[\argsubscript][\argsuperscript]}}%
  }
\newcommandx{\tupleauxgx}[7][1=, 2=, 3=, 4=, 5=, 6=, 7=]
  {%
  \tupleg
    {\argdef{#1}{\defarga[\argsubscript][\argsuperscript]}}
    {\argdef{#2}{\defargb[\argsubscript][\argsuperscript]}}
    {\argdef{#3}{\defargc[\argsubscript][\argsuperscript]}}
    {\argdef{#4}{\defargd[\argsubscript][\argsuperscript]}}
    {\argdef{#5}{\defarge[\argsubscript][\argsuperscript]}}
    {\argdef{#6}{\defargf[\argsubscript][\argsuperscript]}}
    {\argdef{#7}{\defargg[\argsubscript][\argsuperscript]}}%
  }
\newcommandx{\tupleauxhx}[8][1=, 2=, 3=, 4=, 5=, 6=, 7=, 8=]
  {%
  \tupleh
    {\argdef{#1}{\defarga[\argsubscript][\argsuperscript]}}
    {\argdef{#2}{\defargb[\argsubscript][\argsuperscript]}}
    {\argdef{#3}{\defargc[\argsubscript][\argsuperscript]}}
    {\argdef{#4}{\defargd[\argsubscript][\argsuperscript]}}
    {\argdef{#5}{\defarge[\argsubscript][\argsuperscript]}}
    {\argdef{#6}{\defargf[\argsubscript][\argsuperscript]}}
    {\argdef{#7}{\defargg[\argsubscript][\argsuperscript]}}
    {\argdef{#8}{\defargh[\argsubscript][\argsuperscript]}}%
  }
\newcommandx{\tupleauxix}[9][1=, 2=, 3=, 4=, 5=, 6=, 7=, 8=, 9=]
  {%
  \tuplei
    {\argdef{#1}{\defarga[\argsubscript][\argsuperscript]}}
    {\argdef{#2}{\defargb[\argsubscript][\argsuperscript]}}
    {\argdef{#3}{\defargc[\argsubscript][\argsuperscript]}}
    {\argdef{#4}{\defargd[\argsubscript][\argsuperscript]}}
    {\argdef{#5}{\defarge[\argsubscript][\argsuperscript]}}
    {\argdef{#6}{\defargf[\argsubscript][\argsuperscript]}}
    {\argdef{#7}{\defargg[\argsubscript][\argsuperscript]}}
    {\argdef{#8}{\defargh[\argsubscript][\argsuperscript]}}
    {\argdef{#9}{\defargi[\argsubscript][\argsuperscript]}}%
  }

\newcommandx{\tupleauxxjx}[9][1=, 2=, 3=, 4=, 5=, 6=, 7=, 8=, 9=]
  {%
  \def\optarga{#1}%
  \def\optargb{#2}%
  \def\optargc{#3}%
  \def\optargd{#4}%
  \def\optarge{#5}%
  \def\optargf{#6}%
  \def\optargg{#7}%
  \def\optargh{#8}%
  \def\optargi{#9}%
  \tupleauxxxjx%
  }
\newcommandx{\tupleauxxkx}[9][1=, 2=, 3=, 4=, 5=, 6=, 7=, 8=, 9=]
  {%
  \def\optarga{#1}%
  \def\optargb{#2}%
  \def\optargc{#3}%
  \def\optargd{#4}%
  \def\optarge{#5}%
  \def\optargf{#6}%
  \def\optargg{#7}%
  \def\optargh{#8}%
  \def\optargi{#9}%
  \tupleauxxxkx%
  }
\newcommandx{\tupleauxxlx}[9][1=, 2=, 3=, 4=, 5=, 6=, 7=, 8=, 9=]
  {%
  \def\optarga{#1}%
  \def\optargb{#2}%
  \def\optargc{#3}%
  \def\optargd{#4}%
  \def\optarge{#5}%
  \def\optargf{#6}%
  \def\optargg{#7}%
  \def\optargh{#8}%
  \def\optargi{#9}%
  \tupleauxxxlx%
  }
\newcommandx{\tupleauxxmx}[9][1=, 2=, 3=, 4=, 5=, 6=, 7=, 8=, 9=]
  {%
  \def\optarga{#1}%
  \def\optargb{#2}%
  \def\optargc{#3}%
  \def\optargd{#4}%
  \def\optarge{#5}%
  \def\optargf{#6}%
  \def\optargg{#7}%
  \def\optargh{#8}%
  \def\optargi{#9}%
  \tupleauxxxmx%
  }
\newcommandx{\tupleauxxnx}[9][1=, 2=, 3=, 4=, 5=, 6=, 7=, 8=, 9=]
  {%
  \def\optarga{#1}%
  \def\optargb{#2}%
  \def\optargc{#3}%
  \def\optargd{#4}%
  \def\optarge{#5}%
  \def\optargf{#6}%
  \def\optargg{#7}%
  \def\optargh{#8}%
  \def\optargi{#9}%
  \tupleauxxxnx%
  }
\newcommandx{\tupleauxxox}[9][1=, 2=, 3=, 4=, 5=, 6=, 7=, 8=, 9=]
  {%
  \def\optarga{#1}%
  \def\optargb{#2}%
  \def\optargc{#3}%
  \def\optargd{#4}%
  \def\optarge{#5}%
  \def\optargf{#6}%
  \def\optargg{#7}%
  \def\optargh{#8}%
  \def\optargi{#9}%
  \tupleauxxxox%
  }
\newcommandx{\tupleauxxpx}[9][1=, 2=, 3=, 4=, 5=, 6=, 7=, 8=, 9=]
  {%
  \def\optarga{#1}%
  \def\optargb{#2}%
  \def\optargc{#3}%
  \def\optargd{#4}%
  \def\optarge{#5}%
  \def\optargf{#6}%
  \def\optargg{#7}%
  \def\optargh{#8}%
  \def\optargi{#9}%
  \tupleauxxxpx%
  }
\newcommandx{\tupleauxxqx}[9][1=, 2=, 3=, 4=, 5=, 6=, 7=, 8=, 9=]
  {%
  \def\optarga{#1}%
  \def\optargb{#2}%
  \def\optargc{#3}%
  \def\optargd{#4}%
  \def\optarge{#5}%
  \def\optargf{#6}%
  \def\optargg{#7}%
  \def\optargh{#8}%
  \def\optargi{#9}%
  \tupleauxxxqx%
  }
\newcommandx{\tupleauxxrx}[9][1=, 2=, 3=, 4=, 5=, 6=, 7=, 8=, 9=]
  {%
  \def\optarga{#1}%
  \def\optargb{#2}%
  \def\optargc{#3}%
  \def\optargd{#4}%
  \def\optarge{#5}%
  \def\optargf{#6}%
  \def\optargg{#7}%
  \def\optargh{#8}%
  \def\optargi{#9}%
  \tupleauxxxrx%
  }

\newcommandx{\tupleauxxxjx}[1][1=]
  {%
  \tuplej
    {\argdef{\optarga}{\tuplearga[\argsubscript][\argsuperscript]}}
    {\argdef{\optargb}{\tupleargb[\argsubscript][\argsuperscript]}}
    {\argdef{\optargc}{\tupleargc[\argsubscript][\argsuperscript]}}
    {\argdef{\optargd}{\tupleargd[\argsubscript][\argsuperscript]}}
    {\argdef{\optarge}{\tuplearge[\argsubscript][\argsuperscript]}}
    {\argdef{\optargf}{\tupleargf[\argsubscript][\argsuperscript]}}
    {\argdef{\optargg}{\tupleargg[\argsubscript][\argsuperscript]}}
    {\argdef{\optargh}{\tupleargh[\argsubscript][\argsuperscript]}}
    {\argdef{\optargi}{\tupleargi[\argsubscript][\argsuperscript]}}
    {\argdef{#1}{\tupleargj[\argsubscript][\argsuperscript]}}%
  }
\newcommandx{\tupleauxxxkx}[2][1=, 2=]
  {%
  \tuplek
    {\argdef{\optarga}{\tuplearga[\argsubscript][\argsuperscript]}}
    {\argdef{\optargb}{\tupleargb[\argsubscript][\argsuperscript]}}
    {\argdef{\optargc}{\tupleargc[\argsubscript][\argsuperscript]}}
    {\argdef{\optargd}{\tupleargd[\argsubscript][\argsuperscript]}}
    {\argdef{\optarge}{\tuplearge[\argsubscript][\argsuperscript]}}
    {\argdef{\optargf}{\tupleargf[\argsubscript][\argsuperscript]}}
    {\argdef{\optargg}{\tupleargg[\argsubscript][\argsuperscript]}}
    {\argdef{\optargh}{\tupleargh[\argsubscript][\argsuperscript]}}
    {\argdef{\optargi}{\tupleargi[\argsubscript][\argsuperscript]}}
    {\argdef{#1}{\tupleargj[\argsubscript][\argsuperscript]}}
    {\argdef{#2}{\tupleargk[\argsubscript][\argsuperscript]}}
  }
\newcommandx{\tupleauxxxlx}[3][1=, 2=, 3=]
  {%
  \tuplel
    {\argdef{\optarga}{\tuplearga[\argsubscript][\argsuperscript]}}
    {\argdef{\optargb}{\tupleargb[\argsubscript][\argsuperscript]}}
    {\argdef{\optargc}{\tupleargc[\argsubscript][\argsuperscript]}}
    {\argdef{\optargd}{\tupleargd[\argsubscript][\argsuperscript]}}
    {\argdef{\optarge}{\tuplearge[\argsubscript][\argsuperscript]}}
    {\argdef{\optargf}{\tupleargf[\argsubscript][\argsuperscript]}}
    {\argdef{\optargg}{\tupleargg[\argsubscript][\argsuperscript]}}
    {\argdef{\optargh}{\tupleargh[\argsubscript][\argsuperscript]}}
    {\argdef{\optargi}{\tupleargi[\argsubscript][\argsuperscript]}}
    {\argdef{#1}{\tupleargj[\argsubscript][\argsuperscript]}}
    {\argdef{#2}{\tupleargk[\argsubscript][\argsuperscript]}}
    {\argdef{#3}{\tupleargl[\argsubscript][\argsuperscript]}}
  }
\newcommandx{\tupleauxxxmx}[4][1=, 2=, 3=, 4=]
  {%
  \tuplem
    {\argdef{\optarga}{\tuplearga[\argsubscript][\argsuperscript]}}
    {\argdef{\optargb}{\tupleargb[\argsubscript][\argsuperscript]}}
    {\argdef{\optargc}{\tupleargc[\argsubscript][\argsuperscript]}}
    {\argdef{\optargd}{\tupleargd[\argsubscript][\argsuperscript]}}
    {\argdef{\optarge}{\tuplearge[\argsubscript][\argsuperscript]}}
    {\argdef{\optargf}{\tupleargf[\argsubscript][\argsuperscript]}}
    {\argdef{\optargg}{\tupleargg[\argsubscript][\argsuperscript]}}
    {\argdef{\optargh}{\tupleargh[\argsubscript][\argsuperscript]}}
    {\argdef{\optargi}{\tupleargi[\argsubscript][\argsuperscript]}}
    {\argdef{#1}{\tupleargj[\argsubscript][\argsuperscript]}}
    {\argdef{#2}{\tupleargk[\argsubscript][\argsuperscript]}}
    {\argdef{#3}{\tupleargl[\argsubscript][\argsuperscript]}}
    {\argdef{#4}{\tupleargm[\argsubscript][\argsuperscript]}}
  }
\newcommandx{\tupleauxxxnx}[5][1=, 2=, 3=, 4=, 5=]
  {%
  \tuplen
    {\argdef{\optarga}{\tuplearga[\argsubscript][\argsuperscript]}}
    {\argdef{\optargb}{\tupleargb[\argsubscript][\argsuperscript]}}
    {\argdef{\optargc}{\tupleargc[\argsubscript][\argsuperscript]}}
    {\argdef{\optargd}{\tupleargd[\argsubscript][\argsuperscript]}}
    {\argdef{\optarge}{\tuplearge[\argsubscript][\argsuperscript]}}
    {\argdef{\optargf}{\tupleargf[\argsubscript][\argsuperscript]}}
    {\argdef{\optargg}{\tupleargg[\argsubscript][\argsuperscript]}}
    {\argdef{\optargh}{\tupleargh[\argsubscript][\argsuperscript]}}
    {\argdef{\optargi}{\tupleargi[\argsubscript][\argsuperscript]}}
    {\argdef{#1}{\tupleargj[\argsubscript][\argsuperscript]}}
    {\argdef{#2}{\tupleargk[\argsubscript][\argsuperscript]}}
    {\argdef{#3}{\tupleargl[\argsubscript][\argsuperscript]}}
    {\argdef{#4}{\tupleargm[\argsubscript][\argsuperscript]}}
    {\argdef{#5}{\tupleargn[\argsubscript][\argsuperscript]}}
  }
\newcommandx{\tupleauxxxox}[6][1=, 2=, 3=, 4=, 5=, 6=]
  {%
  \tupleo
    {\argdef{\optarga}{\tuplearga[\argsubscript][\argsuperscript]}}
    {\argdef{\optargb}{\tupleargb[\argsubscript][\argsuperscript]}}
    {\argdef{\optargc}{\tupleargc[\argsubscript][\argsuperscript]}}
    {\argdef{\optargd}{\tupleargd[\argsubscript][\argsuperscript]}}
    {\argdef{\optarge}{\tuplearge[\argsubscript][\argsuperscript]}}
    {\argdef{\optargf}{\tupleargf[\argsubscript][\argsuperscript]}}
    {\argdef{\optargg}{\tupleargg[\argsubscript][\argsuperscript]}}
    {\argdef{\optargh}{\tupleargh[\argsubscript][\argsuperscript]}}
    {\argdef{\optargi}{\tupleargi[\argsubscript][\argsuperscript]}}
    {\argdef{#1}{\tupleargj[\argsubscript][\argsuperscript]}}
    {\argdef{#2}{\tupleargk[\argsubscript][\argsuperscript]}}
    {\argdef{#3}{\tupleargl[\argsubscript][\argsuperscript]}}
    {\argdef{#4}{\tupleargm[\argsubscript][\argsuperscript]}}
    {\argdef{#5}{\tupleargn[\argsubscript][\argsuperscript]}}
    {\argdef{#6}{\tupleargo[\argsubscript][\argsuperscript]}}
  }
\newcommandx{\tupleauxxxpx}[7][1=, 2=, 3=, 4=, 5=, 6=, 7=]
  {%
  \tuplep
    {\argdef{\optarga}{\tuplearga[\argsubscript][\argsuperscript]}}
    {\argdef{\optargb}{\tupleargb[\argsubscript][\argsuperscript]}}
    {\argdef{\optargc}{\tupleargc[\argsubscript][\argsuperscript]}}
    {\argdef{\optargd}{\tupleargd[\argsubscript][\argsuperscript]}}
    {\argdef{\optarge}{\tuplearge[\argsubscript][\argsuperscript]}}
    {\argdef{\optargf}{\tupleargf[\argsubscript][\argsuperscript]}}
    {\argdef{\optargg}{\tupleargg[\argsubscript][\argsuperscript]}}
    {\argdef{\optargh}{\tupleargh[\argsubscript][\argsuperscript]}}
    {\argdef{\optargi}{\tupleargi[\argsubscript][\argsuperscript]}}
    {\argdef{#1}{\tupleargj[\argsubscript][\argsuperscript]}}
    {\argdef{#2}{\tupleargk[\argsubscript][\argsuperscript]}}
    {\argdef{#3}{\tupleargl[\argsubscript][\argsuperscript]}}
    {\argdef{#4}{\tupleargm[\argsubscript][\argsuperscript]}}
    {\argdef{#5}{\tupleargn[\argsubscript][\argsuperscript]}}
    {\argdef{#6}{\tupleargo[\argsubscript][\argsuperscript]}}
    {\argdef{#7}{\tupleargp[\argsubscript][\argsuperscript]}}
  }
\newcommandx{\tupleauxxxqx}[8][1=, 2=, 3=, 4=, 5=, 6=, 7=, 8=]
  {%
  \tupleq
    {\argdef{\optarga}{\tuplearga[\argsubscript][\argsuperscript]}}
    {\argdef{\optargb}{\tupleargb[\argsubscript][\argsuperscript]}}
    {\argdef{\optargc}{\tupleargc[\argsubscript][\argsuperscript]}}
    {\argdef{\optargd}{\tupleargd[\argsubscript][\argsuperscript]}}
    {\argdef{\optarge}{\tuplearge[\argsubscript][\argsuperscript]}}
    {\argdef{\optargf}{\tupleargf[\argsubscript][\argsuperscript]}}
    {\argdef{\optargg}{\tupleargg[\argsubscript][\argsuperscript]}}
    {\argdef{\optargh}{\tupleargh[\argsubscript][\argsuperscript]}}
    {\argdef{\optargi}{\tupleargi[\argsubscript][\argsuperscript]}}
    {\argdef{#1}{\tupleargj[\argsubscript][\argsuperscript]}}
    {\argdef{#2}{\tupleargk[\argsubscript][\argsuperscript]}}
    {\argdef{#3}{\tupleargl[\argsubscript][\argsuperscript]}}
    {\argdef{#4}{\tupleargm[\argsubscript][\argsuperscript]}}
    {\argdef{#5}{\tupleargn[\argsubscript][\argsuperscript]}}
    {\argdef{#6}{\tupleargo[\argsubscript][\argsuperscript]}}
    {\argdef{#7}{\tupleargp[\argsubscript][\argsuperscript]}}
    {\argdef{#8}{\tupleargq[\argsubscript][\argsuperscript]}}
  }
\newcommandx{\tupleauxxxrx}[9][1=, 2=, 3=, 4=, 5=, 6=, 7=, 8=, 9=]
  {%
  \tupler
    {\argdef{\optarga}{\tuplearga[\argsubscript][\argsuperscript]}}
    {\argdef{\optargb}{\tupleargb[\argsubscript][\argsuperscript]}}
    {\argdef{\optargc}{\tupleargc[\argsubscript][\argsuperscript]}}
    {\argdef{\optargd}{\tupleargd[\argsubscript][\argsuperscript]}}
    {\argdef{\optarge}{\tuplearge[\argsubscript][\argsuperscript]}}
    {\argdef{\optargf}{\tupleargf[\argsubscript][\argsuperscript]}}
    {\argdef{\optargg}{\tupleargg[\argsubscript][\argsuperscript]}}
    {\argdef{\optargh}{\tupleargh[\argsubscript][\argsuperscript]}}
    {\argdef{\optargi}{\tupleargi[\argsubscript][\argsuperscript]}}
    {\argdef{#1}{\tupleargj[\argsubscript][\argsuperscript]}}
    {\argdef{#2}{\tupleargk[\argsubscript][\argsuperscript]}}
    {\argdef{#3}{\tupleargl[\argsubscript][\argsuperscript]}}
    {\argdef{#4}{\tupleargm[\argsubscript][\argsuperscript]}}
    {\argdef{#5}{\tupleargn[\argsubscript][\argsuperscript]}}
    {\argdef{#6}{\tupleargo[\argsubscript][\argsuperscript]}}
    {\argdef{#7}{\tupleargp[\argsubscript][\argsuperscript]}}
    {\argdef{#8}{\tupleargq[\argsubscript][\argsuperscript]}}
    {\argdef{#9}{\tupleargr[\argsubscript][\argsuperscript]}}%
  }

\newcommand{\set}[2]
  {\ensuremath{\argint{\{}{\argext{#1}{\allowbreak:\allowbreak}{#2}}{\}}}}

\newcommand{\pow}[1]
  {\ensuremath{2^{#1}}}

\newcommand{\card}[1]
  {\mthempty{\argint{\vert}{#1}{\vert}}}

\newcommand{\cmp}
  {\ensuremath{\circ}}

\newcommand{\rst}
  {\mthempty{\upharpoonright}}

\newcommandx{\pto}[2][1=, 2=]
  {\ensuremath{\rightharpoonup}}

\newcommandx{\cto}[2][1=, 2=]
  {\:\mthempty{\to}[#1][#2]\:}
\newcommandx{\cpto}[2][1=, 2=]
  {\:\mthempty{\pto}[#1][#2]\:}

\newcommand{\SetN}
  {\mthset[1]{N}}

\newcommand{\SetZ}
  {\mthset[1]{Z}}

\newcommand{\numcc}[2]
  {\mthempty{[\argb{#1}{#2}]}}
\newcommand{\numco}[2]
  {\mthempty{[\argb{#1}{#2})}}

\DeclareRobustCommand{\min}
  {\mthfun{min}}
\DeclareRobustCommand{\max}
  {\mthfun{max}}

\DeclareRobustCommand{\arg}
  {\mthfun{arg}}

\newcommandx{\EF}[5][1=, 2=, 3=, 4=, 5=]
  {\txtargname{EF#5{\small\argint{$[$}{#1}{$]$}}}[#2][#3]{#4}}

\newcommandx{\BG}[5][1=, 2=, 3=, 4=, 5=]
  {\txtargname{BG#5{\small\argint{$[$}{#1}{$]$}}}[#2][#3]{#4}}

\newcommandx{\CG}[5][1=, 2=, 3=, 4=, 5=]
  {\txtargname{CG#5{\small\argint{$[$}{#1}{$]$}}}[#2][#3]{#4}}

\newcommandx{\PG}[5][1=, 2=, 3=, 4=, 5=]
  {\txtargname{PG#5{\small\argint{$[$}{#1}{$]$}}}[#2][#3]{#4}}

\newcommandx{\RG}[5][1=, 2=, 3=, 4=, 5=]
  {\txtargname{RG#5{\small\argint{$[$}{#1}{$]$}}}[#2][#3]{#4}}

\newcommandx{\SG}[5][1=, 2=, 3=, 4=, 5=]
  {\txtargname{SG#5{\small\argint{$[$}{#1}{$]$}}}[#2][#3]{#4}}

\newcommandx{\MG}[5][1=, 2=, 3=, 4=, 5=]
  {\txtargname{MG#5{\small\argint{$[$}{#1}{$]$}}}[#2][#3]{#4}}

\newcommandx{\EG}[5][1=, 2=, 3=, 4=, 5=]
  {\txtargname{EG#5{\small\argint{$[$}{#1}{$]$}}}[#2][#3]{#4}}

\newcommandx{\MPG}[5][1=, 2=, 3=, 4=, 5=]
  {\txtargname{MPG#5{\small\argint{$[$}{#1}{$]$}}}[#2][#3]{#4}}

\newcommandx{\DPG}[5][1=, 2=, 3=, 4=, 5=]
  {\txtargname{DPG#5{\small\argint{$[$}{#1}{$]$}}}[#2][#3]{#4}}

\renewcommandx{\SG}[5][1=, 2=, 3=, 4=, 5=]
  {\txtargname{SG#5{\small\argint{$[$}{#1}{$]$}}}[#2][#3]{#4}}

\newcommandx{\BF}[5][1=, 2=, 3=, 4=, 5=]
  {\txtargname{BF#5{\small\argint{$[$}{#1}{$]$}}}[#2][#3]{#4}}

\newcommandx{\QBF}
  {{\txtname{Q}}\BF}

\newcommandx{\FOL}[5][1=, 2=, 3=, 4=, 5=]
  {\txtargname{FOL#5{\small\argint{$[$}{#1}{$]$}}}[#2][#3]{#4}}

\newcommandx{\SOL}[5][1=, 2=, 3=, 4=, 5=]
  {\txtargname{SOL#5{\small\argint{$[$}{#1}{$]$}}}[#2][#3]{#4}}

\newcommandx{\TL}[5][1=, 2=, 3=, 4=, 5=]
  {\txtargname{TL#5{\small\argint{$[$}{#1}{$]$}}}[#2][#3]{#4}}

\newcommandx{\PL}[5][1=, 2=, 3=, 4=, 5=]
  {\txtargname{PL#5{\small\argint{$[$}{#1}{$]$}}}[#2][#3]{#4}}

\newcommandx{\ML}[5][1=, 2=, 3=, 4=, 5=]
  {\txtargname{ML#5{\small\argint{$[$}{#1}{$]$}}}[#2][#3]{#4}}

\newcommandx{\MC}[5][1=, 2=, 3=, 4=, 5=]
  {\txtargname{$\mu$Calculus#5{\small\argint{$[$}{#1}{$]$}}}[#2][#3]{#4}}

\newcommandx{\LTL}[5][1=, 2=, 3=, 4=, 5=]
  {\txtargname{LTL#5{\small\argint{$[$}{#1}{$]$}}}[#2][#3]{#4}}

\newcommandx{\PTL}[5][1=, 2=, 3=, 4=, 5=]
  {\txtargname{PTL#5{\small\argint{$[$}{#1}{$]$}}}[#2][#3]{#4}}

\newcommandx{\CTL}[5][1=, 2=, 3=, 4=, 5=]
  {\txtargname{CTL#5{\small\argint{$[$}{#1}{$]$}}}[#2][#3]{#4}}

\newcommandx{\CTLP}[5][1=, 2=, 3=, 4=, 5=]
  {\txtargname{CTL$^{+}$#5{\small\argint{$[$}{#1}{$]$}}}[#2][#3]{#4}}

\newcommandx{\CTLS}[5][1=, 2=, 3=, 4=, 5=]
  {\txtargname{CTL$^{\star}$#5{\small\argint{$[$}{#1}{$]$}}}[#2][#3]{#4}}

\newcommandx{\STL}[5][1=, 2=, 3=, 4=, 5=]
  {\txtargname{STL#5{\small\argint{$[$}{#1}{$]$}}}[#2][#3]{#4}}

\newcommandx{\STLP}[5][1=, 2=, 3=, 4=, 5=]
  {\txtargname{STL$^{+}$#5{\small\argint{$[$}{#1}{$]$}}}[#2][#3]{#4}}

\newcommandx{\STLS}[5][1=, 2=, 3=, 4=, 5=]
  {\txtargname{STL$^{\star}$#5{\small\argint{$[$}{#1}{$]$}}}[#2][#3]{#4}}

\newcommandx{\ATL}[5][1=, 2=, 3=, 4=, 5=]
  {\txtargname{ATL#5{\small\argint{$[$}{#1}{$]$}}}[#2][#3]{#4}}

\newcommandx{\ATLP}[5][1=, 2=, 3=, 4=, 5=]
  {\txtargname{ATL$^{+}$#5{\small\argint{$[$}{#1}{$]$}}}[#2][#3]{#4}}

\newcommandx{\ATLS}[5][1=, 2=, 3=, 4=, 5=]
  {\txtargname{ATL$^{\star}$#5{\small\argint{$[$}{#1}{$]$}}}[#2][#3]{#4}}

\newcommandx{\SL}[5][1=, 2=, 3=, 4=, 5=]
  {\txtargname{SL#5{\small\argint{$[$}{#1}{$]$}}}[#2][#3]{#4}}

\newcommandx{\LogTime}[4][1=, 2=, 3=, 4=]
  {\txtargname{LogTime#4}[#2][#3]{#1}}
\newcommandx{\LogTimeE}[4][1=, 2=, 3=, 4=]
  {\LogTime[#1][#2][#3][#4]-\EComplexity}
\newcommandx{\LogTimeH}[4][1=, 2=, 3=, 4=]
  {\LogTime[#1][#2][#3][#4]-\HComplexity}
\newcommandx{\LogTimeC}[4][1=, 2=, 3=, 4=]
  {\LogTime[#1][#2][#3][#4]-\CComplexity}

\newcommandx{\LogSpace}[4][1=, 2=, 3=, 4=]
  {\txtargname{LogSpace#4}[#2][#3]{#1}}
\newcommandx{\LogSpaceE}[4][1=, 2=, 3=, 4=]
  {\LogSpace[#1][#2][#3][#4]-\EComplexity}
\newcommandx{\LogSpaceH}[4][1=, 2=, 3=, 4=]
  {\LogSpace[#1][#2][#3][#4]-\HComplexity}
\newcommandx{\LogSpaceC}[4][1=, 2=, 3=, 4=]
  {\LogSpace[#1][#2][#3][#4]-\CComplexity}

\newcommandx{\PTime}[4][1=, 2=, 3=, 4=]
  {\txtargname{PTime#4}[#2][#3]{#1}}
\newcommandx{\PTimeE}[4][1=, 2=, 3=, 4=]
  {\PTime[#1][#2][#3][#4]-\EComplexity}
\newcommandx{\PTimeH}[4][1=, 2=, 3=, 4=]
  {\PTime[#1][#2][#3][#4]-\HComplexity}
\newcommandx{\PTimeC}[4][1=, 2=, 3=, 4=]
  {\PTime[#1][#2][#3][#4]-\CComplexity}

\newcommand{\UPTime}
  {\UComplexity\PTime}

\newcommand{\CoUPTime}
  {\CoComplexity\UPTime}

\newcommand{\NPTime}
  {\NComplexity\PTime}

\newcommand{\CoNPTime}
  {\CoComplexity\NPTime}

\newcommandx{\PSpace}[4][1=, 2=, 3=, 4=]
  {\txtargname{PSpace#4}[#2][#3]{#1}}
\newcommandx{\PSpaceE}[4][1=, 2=, 3=, 4=]
  {\PSpace[#1][#2][#3][#4]-\EComplexity}
\newcommandx{\PSpaceH}[4][1=, 2=, 3=, 4=]
  {\PSpace[#1][#2][#3][#4]-\HComplexity}
\newcommandx{\PSpaceC}[4][1=, 2=, 3=, 4=]
  {\PSpace[#1][#2][#3][#4]-\CComplexity}

\newcommandx{\ExpTime}[4][1=, 2=, 3=, 4=]
  {\txtargname{ExpTime#4}[#2][#3]{#1}}
\newcommandx{\ExpTimeE}[4][1=, 2=, 3=, 4=]
  {\ExpTime[#1][#2][#3][#4]-\EComplexity}
\newcommandx{\ExpTimeH}[4][1=, 2=, 3=, 4=]
  {\ExpTime[#1][#2][#3][#4]-\HComplexity}
\newcommandx{\ExpTimeC}[4][1=, 2=, 3=, 4=]
  {\ExpTime[#1][#2][#3][#4]-\CComplexity}

\newcommandx{\ExpSpace}[4][1=, 2=, 3=, 4=]
  {\txtargname{ExpSpace#4}[#2][#3]{#1}}
\newcommandx{\ExpSpaceE}[4][1=, 2=, 3=, 4=]
  {\ExpSpace[#1][#2][#3][#4]-\EComplexity}
\newcommandx{\ExpSpaceH}[4][1=, 2=, 3=, 4=]
  {\ExpSpace[#1][#2][#3][#4]-\HComplexity}
\newcommandx{\ExpSpaceC}[4][1=, 2=, 3=, 4=]
  {\ExpSpace[#1][#2][#3][#4]-\CComplexity}

\newcommandx{\NonElm}[4][1=, 2=, 3=, 4=]
  {\txtargname{NonElementary#4}[#2][#3]{#1}}
\newcommandx{\NonElmE}[4][1=, 2=, 3=, 4=]
  {\NonElm[#1][#2][#3][#4]-\EComplexity}
\newcommandx{\NonElmH}[4][1=, 2=, 3=, 4=]
  {\NonElm[#1][#2][#3][#4]-\HComplexity}
\newcommandx{\NonElmC}[4][1=, 2=, 3=, 4=]
  {\NonElm[#1][#2][#3][#4]-\CComplexity}

\newcommandx{\NonElmTime}[4][1=, 2=, 3=, 4=]
  {\txtargname{NonElementaryTime#4}[#2][#3]{#1}}
\newcommandx{\NonElmTimeE}[4][1=, 2=, 3=, 4=]
  {\NonElmTime[#1][#2][#3][#4]-\EComplexity}
\newcommandx{\NonElmTimeH}[4][1=, 2=, 3=, 4=]
  {\NonElmTime[#1][#2][#3][#4]-\HComplexity}
\newcommandx{\NonElmTimeC}[4][1=, 2=, 3=, 4=]
  {\NonElmTime[#1][#2][#3][#4]-\CComplexity}

\newcommandx{\NonElmSpace}[4][1=, 2=, 3=, 4=]
  {\txtargname{NonElementarySpace#4}[#2][#3]{#1}}
\newcommandx{\NonElmSpaceE}[4][1=, 2=, 3=, 4=]
  {\NonElmSpace[#1][#2][#3][#4]-\EComplexity}
\newcommandx{\NonElmSpaceH}[4][1=, 2=, 3=, 4=]
  {\NonElmSpace[#1][#2][#3][#4]-\HComplexity}
\newcommandx{\NonElmSpaceC}[4][1=, 2=, 3=, 4=]
  {\NonElmSpace[#1][#2][#3][#4]-\CComplexity}

\newcommandx{\DLHier}[4][2=, 3=, 4=]
  {\mthargset[0]{\Delta#4}[#1][#3]{#2}}
\newcommandx{\DLHierE}[4][2=, 3=, 4=]
  {\DLHier{#1}[#2][#3][#4]-\EComplexity}
\newcommandx{\DLHierH}[4][2=, 3=, 4=]
  {\DLHier{#1}[#2][#3][#4]-\HComplexity}
\newcommandx{\DLHierC}[4][2=, 3=, 4=]
  {\DLHier{#1}[#2][#3][#4]-\CComplexity}

\newcommandx{\ELHier}[4][2=, 3=, 4=]
  {\mthargset[0]{\Sigma#4}[#1][#3]{#2}}
\newcommandx{\ELHierE}[4][2=, 3=, 4=]
  {\ELHier{#1}[#2][#3][#4]-\EComplexity}
\newcommandx{\ELHierH}[4][2=, 3=, 4=]
  {\ELHier{#1}[#2][#3][#4]-\HComplexity}
\newcommandx{\ELHierC}[4][2=, 3=, 4=]
  {\ELHier{#1}[#2][#3][#4]-\CComplexity}

\newcommandx{\ULHier}[4][2=, 3=, 4=]
  {\mthargset[0]{\Pi#4}[#1][#3]{#2}}
\newcommandx{\ULHierE}[4][2=, 3=, 4=]
  {\ULHier{#1}[#2][#3][#4]-\EComplexity}
\newcommandx{\ULHierH}[4][2=, 3=, 4=]
  {\ULHier{#1}[#2][#3][#4]-\HComplexity}
\newcommandx{\ULHierC}[4][2=, 3=, 4=]
  {\ULHier{#1}[#2][#3][#4]-\CComplexity}

\newcommandx{\DBHier}[4][2=, 3=, 4=]
  {\mthargset[3]{\Delta#4}[#1][#3]{#2}}
\newcommandx{\DBHierE}[4][2=, 3=, 4=]
  {\DBHier{#1}[#2][#3][#4]-\EComplexity}
\newcommandx{\DBHierH}[4][2=, 3=, 4=]
  {\DBHier{#1}[#2][#3][#4]-\HComplexity}
\newcommandx{\DBHierC}[4][2=, 3=, 4=]
  {\DBHier{#1}[#2][#3][#4]-\CComplexity}

\newcommandx{\EBHier}[4][2=, 3=, 4=]
  {\mthargset[3]{\Sigma#4}[#1][#3]{#2}}
\newcommandx{\EBHierE}[4][2=, 3=, 4=]
  {\EBHier{#1}[#2][#3][#4]-\EComplexity}
\newcommandx{\EBHierH}[4][2=, 3=, 4=]
  {\EBHier{#1}[#2][#3][#4]-\HComplexity}
\newcommandx{\EBHierC}[4][2=, 3=, 4=]
  {\EBHier{#1}[#2][#3][#4]-\CComplexity}

\newcommandx{\UBHier}[4][2=, 3=, 4=]
  {\mthargset[3]{\Pi#4}[#1][#3]{#2}}
\newcommandx{\UBHierE}[4][2=, 3=, 4=]
  {\UBHier{#1}[#2][#3][#4]-\EComplexity}
\newcommandx{\UBHierH}[4][2=, 3=, 4=]
  {\UBHier{#1}[#2][#3][#4]-\HComplexity}
\newcommandx{\UBHierC}[4][2=, 3=, 4=]
  {\UBHier{#1}[#2][#3][#4]-\CComplexity}

\newcommandx{\DPolHier}[4][2=, 3=, 4=]
  {\DLHier{#1}[#2][\argb{\mathrm{P}}{#3}][#4]}
\newcommandx{\DPolHierE}[4][2=, 3=, 4=]
  {\DPolHier{#1}[#2][#3][#4]-\EComplexity}
\newcommandx{\DPolHierH}[4][2=, 3=, 4=]
  {\DPolHier{#1}[#2][#3][#4]-\HComplexity}
\newcommandx{\DPolHierC}[4][2=, 3=, 4=]
  {\DPolHier{#1}[#2][#3][#4]-\CComplexity}

\newcommandx{\EPolHier}[4][2=, 3=, 4=]
  {\ELHier{#1}[#2][\argb{\mathrm{P}}{#3}][#4]}
\newcommandx{\EPolHierE}[4][2=, 3=, 4=]
  {\EPolHier{#1}[#2][#3][#4]-\EComplexity}
\newcommandx{\EPolHierH}[4][2=, 3=, 4=]
  {\EPolHier{#1}[#2][#3][#4]-\HComplexity}
\newcommandx{\EPolHierC}[4][2=, 3=, 4=]
  {\EPolHier{#1}[#2][#3][#4]-\CComplexity}

\newcommandx{\UPolHier}[4][2=, 3=, 4=]
  {\ULHier{#1}[#2][\argb{\mathrm{P}}{#3}][#4]}
\newcommandx{\UPolHierE}[4][2=, 3=, 4=]
  {\UPolHier{#1}[#2][#3][#4]-\EComplexity}
\newcommandx{\UPolHierH}[4][2=, 3=, 4=]
  {\UPolHier{#1}[#2][#3][#4]-\HComplexity}
\newcommandx{\UPolHierC}[4][2=, 3=, 4=]
  {\UPolHier{#1}[#2][#3][#4]-\CComplexity}

\newcommandx{\DAriHier}[4][2=, 3=, 4=]
  {\DLHier{#1}[#2][\argb{0}{#3}][#4]}
\newcommandx{\DAriHierE}[4][2=, 3=, 4=]
  {\DAriHier{#1}[#2][#3][#4]-\EComplexity}
\newcommandx{\DAriHierH}[4][2=, 3=, 4=]
  {\DAriHier{#1}[#2][#3][#4]-\HComplexity}
\newcommandx{\DAriHierC}[4][2=, 3=, 4=]
  {\DAriHier{#1}[#2][#3][#4]-\CComplexity}

\newcommandx{\EAriHier}[4][2=, 3=, 4=]
  {\ELHier{#1}[#2][\argb{0}{#3}][#4]}
\newcommandx{\EAriHierE}[4][2=, 3=, 4=]
  {\EAriHier{#1}[#2][#3][#4]-\EComplexity}
\newcommandx{\EAriHierH}[4][2=, 3=, 4=]
  {\EAriHier{#1}[#2][#3][#4]-\HComplexity}
\newcommandx{\EAriHierC}[4][2=, 3=, 4=]
  {\EAriHier{#1}[#2][#3][#4]-\CComplexity}

\newcommandx{\UAriHier}[4][2=, 3=, 4=]
  {\ULHier{#1}[#2][\argb{0}{#3}][#4]}
\newcommandx{\UAriHierE}[4][2=, 3=, 4=]
  {\UAriHier{#1}[#2][#3][#4]-\EComplexity}
\newcommandx{\UAriHierH}[4][2=, 3=, 4=]
  {\UAriHier{#1}[#2][#3][#4]-\HComplexity}
\newcommandx{\UAriHierC}[4][2=, 3=, 4=]
  {\UAriHier{#1}[#2][#3][#4]-\CComplexity}

\newcommandx{\DAnaHier}[4][2=, 3=, 4=]
  {\DLHier{#1}[#2][\argb{1}{#3}][#4]}
\newcommandx{\DAnaHierE}[4][2=, 3=, 4=]
  {\DAnaHier{#1}[#2][#3][#4]-\EComplexity}
\newcommandx{\DAnaHierH}[4][2=, 3=, 4=]
  {\DAnaHier{#1}[#2][#3][#4]-\HComplexity}
\newcommandx{\DAnaHierC}[4][2=, 3=, 4=]
  {\DAnaHier{#1}[#2][#3][#4]-\CComplexity}

\newcommandx{\EAnaHier}[4][2=, 3=, 4=]
  {\ELHier{#1}[#2][\argb{1}{#3}][#4]}
\newcommandx{\EAnaHierE}[4][2=, 3=, 4=]
  {\EAnaHier{#1}[#2][#3][#4]-\EComplexity}
\newcommandx{\EAnaHierH}[4][2=, 3=, 4=]
  {\EAnaHier{#1}[#2][#3][#4]-\HComplexity}
\newcommandx{\EAnaHierC}[4][2=, 3=, 4=]
  {\EAnaHier{#1}[#2][#3][#4]-\CComplexity}

\newcommandx{\UAnaHier}[4][2=, 3=, 4=]
  {\ULHier{#1}[#2][\argb{1}{#3}][#4]}
\newcommandx{\UAnaHierE}[4][2=, 3=, 4=]
  {\UAnaHier{#1}[#2][#3][#4]-\EComplexity}
\newcommandx{\UAnaHierH}[4][2=, 3=, 4=]
  {\UAnaHier{#1}[#2][#3][#4]-\HComplexity}
\newcommandx{\UAnaHierC}[4][2=, 3=, 4=]
  {\UAnaHier{#1}[#2][#3][#4]-\CComplexity}

\newcommandx{\DBorHier}[4][2=, 3=, 4=]
  {\DBHier{#1}[#2][\argb{\mathrm{B}}{#3}][#4]}
\newcommandx{\DBorHierE}[4][2=, 3=, 4=]
  {\DBorHier{#1}[#2][#3][#4]-\EComplexity}
\newcommandx{\DBorHierH}[4][2=, 3=, 4=]
  {\DBorHier{#1}[#2][#3][#4]-\HComplexity}
\newcommandx{\DBorHierC}[4][2=, 3=, 4=]
  {\DBorHier{#1}[#2][#3][#4]-\CComplexity}

\newcommandx{\EBorHier}[4][2=, 3=, 4=]
  {\EBHier{#1}[#2][\argb{\mathrm{B}}{#3}][#4]}
\newcommandx{\EBorHierE}[4][2=, 3=, 4=]
  {\EBorHier{#1}[#2][#3][#4]-\EComplexity}
\newcommandx{\EBorHierH}[4][2=, 3=, 4=]
  {\EBorHier{#1}[#2][#3][#4]-\HComplexity}
\newcommandx{\EBorHierC}[4][2=, 3=, 4=]
  {\EBorHier{#1}[#2][#3][#4]-\CComplexity}

\newcommandx{\UBorHier}[4][2=, 3=, 4=]
  {\UBHier{#1}[#2][\argb{\mathrm{B}}{#3}][#4]}
\newcommandx{\UBorHierE}[4][2=, 3=, 4=]
  {\UBorHier{#1}[#2][#3][#4]-\EComplexity}
\newcommandx{\UBorHierH}[4][2=, 3=, 4=]
  {\UBorHier{#1}[#2][#3][#4]-\HComplexity}
\newcommandx{\UBorHierC}[4][2=, 3=, 4=]
  {\UBorHier{#1}[#2][#3][#4]-\CComplexity}

\newcommand{\EComplexity}
  {{\txtname{easy}}}

\newcommand{\HComplexity}
  {{\txtname{hard}}}

\newcommand{\CComplexity}
  {{\txtname{complete}}}

\newcommand{\UComplexity}
  {{\txtname{U}}}

\newcommand{\NComplexity}
  {{\txtname{N}}}

\newcommand{\CoComplexity}
  {{\txtname{Co}}}

\newtheorem{definition}{Definition}

\newtheorem{proposition}{Proposition}

\newtheorem{theorem}{Theorem}
\newtheorem{corollary}{Corollary}

% ...

\newcommandx{\qdmf}
  {\txtname{qdm}}
\newcommandx{\smf}
  {\txtname{sm}}

\newcommandx{\ParGam}
  {\txtname{ParityGames}}

\newcommandx{\PP}[5][1=, 2=, 3=, 4=, 5=]
  {\txtargname{PP#5{\small\argint{$[$}{#1}{$]$}}}[#2][#3]{#4}}
\newcommandx{\PR}[5][1=, 2=, 3=, 4=, 5=]
  {\txtargname{PR#5{\small\argint{$[$}{#1}{$]$}}}[#2][#3]{#4}}
\newcommandx{\BRIM}[5][1=, 2=, 3=, 4=, 5=]
  {\txtargname{SEPM#5{\small\argint{$[$}{#1}{$]$}}}[#2][#3]{#4}}
\newcommandx{\QDPM}[5][1=, 2=, 3=, 4=, 5=]
  {\txtargname{QDPM#5{\small\argint{$[$}{#1}{$]$}}}[#2][#3]{#4}}

\newcommand{\gamname}{\Game}
\newcommand{\GamName}
  {\mthname{\gamname}}

\newcommand{\movrel}{Mv}
\newcommandx{\MovRel}[3][1=, 2=, 3=]
  {\mthrel{\movrel#3}[#1][#2]}

\newcommand{\winset}{Win}
\newcommandx{\WinSet}[3][1=, 2=, 3=]
  {\mthset{\winset#3}[#1][#2]}

\newcommand{\prtfun}{pr}
\newcommandx{\prtFun}[4][1=, 2=, 3=, 4=]
  {\mthargfun{\prtfun#4}[#1][#2]{#3}}

\newcommand{\depfun}{dep}
\newcommandx{\depFun}[4][1=, 2=, 3=, 4=]
  {\mthargfun{\depfun#4}[#1][#2]{#3}}

\newcommand{\escfun}{esc}
\newcommandx{\escFun}[4][1=, 2=, 3=, 4=]
  {\mthargfun{\escfun#4}[#1][#2]{#3}}

\newcommand{\intfun}{int}
\newcommandx{\intFun}[4][1=, 2=, 3=, 4=]
  {\mthargfun{\intfun#4}[#1][#2]{#3}}

\newcommand{\styfun}{sty}
\newcommandx{\styFun}[4][1=, 2=, 3=, 4=]
  {\mthargfun{\styfun#4}[#1][#2]{#3}}

\newcommand{\srcfun}{src}
\newcommandx{\srcFun}[4][1=, 2=, 3=, 4=]
  {\mthargfun{\srcfun#4}[#1][#2]{#3}}

\newcommand{\regset}{Rg}
\newcommandx{\RegSet}[3][1=, 2=, 3=]
  {\mthset{\regset#3}[#1][#2]}

\newcommand{\qdset}{QD}
\newcommandx{\QDSet}[3][1=, 2=, 3=]
  {\mthset{\qdset#3}[#1][#2]}

\newcommand{\diselm}{\top}
\newcommandx{\disElm}[3][1=, 2=, 3=]
  {\mthelm{\diselm#3}[#1][#2]}

\newcommand{\ordrel}{\prec}
\newcommandx{\ordRel}[3][1=, 2=, 3=]
  {\mthrel{\ordrel#3}[#1][#2]}

\newcommand{\comrel}{\Yright}
\newcommandx{\comRel}[3][1=, 2=, 3=]
  {\mthrel{\comrel#3}[#1][#2]}

\newcommand{\qryfun}{\Re}
\newcommandx{\qryFun}[4][1=, 2=, 3=, 4=]
  {\mthargfun{\qryfun#4}[#1][#2]{#3}}

\newcommand{\sucfun}{\downarrow}
\newcommandx{\sucFun}[4][1=, 2=, 3=, 4=]
  {\mthargfun{\sucfun#4}[#1][#2]{#3}}

\newcommand{\strset}{Str}
\newcommandx{\StrSet}[3][1=, 2=, 3=]
  {\mthset{\strset#3}[#1][#2]}
\newcommand{\strsym}{\sigma}
\newcommandx{\strSym}[4][1=, 2=, 3=, 4=]
  {\mthargfun{\strsym#4}[#1][#2]{#3}}
\newcommand{\strelm}{\sigma}
\newcommandx{\strElm}[4][1=, 2=, 3=, 4=]
  {\mthargfun{\strelm#4}[#1][#2]{#3}}

\newcommand{\playset}{Plays}
\newcommandx{\PlaySet}[4][1=, 2=, 3=, 4=]
  {\mthset{\playset#4}[#1][#2]{#3}}
\newcommand{\playfun}{play}
\newcommandx{\playFun}[4][1=, 2=, 3=, 4=]
  {\mthargfun{\playfun#4}[#1][#2]{#3}}

\newcommand{\pthset}{Pth}
\newcommandx{\PthSet}[3][1=, 2=, 3=]
  {\mthset{\pthset#3}[#1][#2]}
\newcommand{\pthsym}{\pi}
\newcommandx{\pthSym}[3][1=, 2=, 3=]
  {\mthsym{\pthsym#3}[#1][#2]}
\newcommand{\pthelm}{\pi}
\newcommandx{\pthElm}[3][1=, 2=, 3=]
  {\mthelm{\pthelm#3}[#1][#2]}

\newcommand{\hstset}{Hst}
\newcommandx{\HstSet}[3][1=, 2=, 3=]
  {\mthset{\hstset#3}[#1][#2]}

\newcommand{\prfset}{Pf}
\newcommandx{\PrfSet}[3][1=, 2=, 3=]
  {\mthset{\prfset#3}[#1][#2]}
\newcommand{\prfsym}{f}
\newcommandx{\prfSym}[3][1=, 2=, 3=]
  {\mthsym{\prfsym#3}[#1][#2]}
\newcommand{\prfelm}{f}
\newcommandx{\prfElm}[3][1=, 2=, 3=]
  {\mthelm{\prfelm#3}[#1][#2]}

\newcommand{\refrel}{\trianglelefteq}
\newcommandx{\refRel}[3][1=, 2=, 3=]
  {\mthrel{\refrel#3}[#1][#2]}

\newcommand{\plrrel}{\leq}
\newcommandx{\plrRel}[3][1=, 2=, 3=]
  {\mthrel{\plrrel#3}[#1][#2]}

\newcommand{\gaufun}{\eth}
\newcommandx{\gauFun}[4][1=, 2=, 3=, 4=]
  {\mthargfun{\gaufun#4}[#1][#2]{#3}}

\newcommand{\noropr}{\upharpoonright}
\newcommandx{\norOpr}[4][1=, 2=, 3=, 4=]
  {\mthargfun{\noropr#4}[#1][#2]{#3}}

\newcommand{\bndfun}{\#}
\newcommandx{\bndFun}[4][1=, 2=, 3=, 4=]
  {\mthargfun{\bndfun#4}[#1][#2]{#3}}

\newcommand{\denot}[1]
  {\mthempty{\argint{\lVert}{#1}{\rVert}}}

% ...

\newcommand{\crtset}{Cr}
\newcommandx{\CrtSet}[3][1=, 2=, 3=]
  {\mthset{\crtset#3}[#1][#2]}

\newcommand{\ifpfun}{ifp}
\newcommandx{\ifpFun}[4][1=, 2=, 3=, 4=]
  {\mthargfun{\ifpfun#4}[#1][#2]{#3}}

\newcommand{\lfpfun}{lfp}
\newcommandx{\lfpFun}[4][1=, 2=, 3=, 4=]
  {\mthargfun{\lfpfun#4}[#1][#2]{#3}}

\newcommand{\rwaset}{RWA}
\newcommandx{\RWASet}[3][1=, 2=, 3=]
  {\mthset{\rwaset#3}[#1][#2]}
\newcommand{\nabfun}{\nabla}
\newcommandx{\nabFun}[4][1=, 2=, 3=, 4=]
  {\mthargfun{\nabfun#4}[#1][#2]{#3}}

\newcommand{\delfun}{\Delta}
\newcommandx{\delFun}[4][1=, 2=, 3=, 4=]
  {\mthargfun{\delfun#4}[#1][#2]{#3}}

% ...

\def\forallcmd#1{\ifx#1\forallcmd\else\defcmd#1\expandafter\forallcmd\fi}

\newcommandx{\DefMacroStructure}[5][2=, 3=, 4=, 5=]
  {
  \DefMacroName{#1}[#2]
  \DefMacroSet{#1}[#3][#4][#5]
  }
\newcommandx{\DefMacroName}[2][2=]
  {
  \csdef{#1Name}{\mthname{\argdef{#2}{#1}}}
  }
\newcommandx{\DefMacroSet}[4][2=, 3=, 4=]
  {
  \csdef{#1Set}{\mthset{\argdef{#2}{#1}}}
  \caselower[q]{#1}
  \DefMacroElm{\thestring}[#3]
  \DefMacroSym{\thestring}[#4]
  }
\newcommandx{\DefMacroElm}[2][2=]
  {
  \csdef{#1Elm}{\mthelm{\argdef{#2}{#1}}}
  \def\defcmd##1{\csdef{##1#1Elm}{\mthelm{##1}}}
  \forallcmd abcdefghijklmnopqrstuvwxyz\forallcmd
  \forallcmd ABCDEFGHIJKLMNOPQRSTUVWXYZ\forallcmd
  }
\newcommandx{\DefMacroSym}[2][2=]
  {
  \csdef{#1Sym}{\mthsym{\argdef{#2}{#1}}}
  \def\defcmd##1{\csdef{##1#1Sym}{\mthsym{##1}}}
  \forallcmd abcdefghijklmnopqrstuvwxyz\forallcmd
  \forallcmd ABCDEFGHIJKLMNOPQRSTUVWXYZ\forallcmd
  }
\newcommandx{\DefMacroFun}[2][2=]
  {
  \csdef{#1Fun}{\mthfun{\argdef{#2}{#1}}}
  }
\newcommandx{\DefMacroRel}[2][2=]
  {
  \csdef{#1Rel}{\mthrel{\argdef{#2}{#1}}}
  }

\DefMacroSym{Plr}[0]
\DefMacroSym{Opp}[1]
\DefMacroSet{Pos}[Ps][v][v]

\DefMacroSet{Prt}[Pr][p][p]

\DefMacroSet{Wgh}[Wg][w][w]
\DefMacroFun{wgh}[wg]

\DefMacroStructure{Msr}[M][Ms][\eta][\eta]
\DefMacroFun{pay}[pf]
\DefMacroFun{val}[c]

\DefMacroSet{GMsr}[GMs][\eta][\eta]
\DefMacroSet{SMsr}[SMs][\eta][\eta]
\DefMacroSet{SPth}[SPth][\pi][\pi]
\DefMacroSet{FPth}[FPth][\pi][\pi]

\DefMacroStructure{MF}[F][][\mu][\mu]
\DefMacroStructure{GMF}[F][GM][\mu][\mu]
\DefMacroStructure{SMF}[F][SM][\mu][\mu]
\DefMacroStructure{QDMF}[Q][QDM][\varrho][\varrho]
\DefMacroStructure{SQDMF}[Q][SQDM][\varrho][\varrho]

\DefMacroSet{Str}[Str][\sigma][\sigma]

\DefMacroFun{npp}

\DefMacroFun{win}

\DefMacroFun{sol}
\DefMacroFun{lift}

\DefMacroFun{nab}[\nabla]
\DefMacroFun{msr}[\eta]

\DefMacroFun{npm}[npm]
\DefMacroFun{nep}

\DefMacroFun{atr}
\DefMacroFun{pre}
\DefMacroFun{post}
\DefMacroFun{adj}

\DefMacroFun{qsi}[Q]
\DefMacroFun{dmn}[\Delta]

\DefMacroFun{unt}
\DefMacroFun{pre}

\DefMacroFun{extmin}
\DefMacroFun{prioqueue}

\DefMacroStructure{Evl}[C][CMs][\phi][\phi]

\DefMacroFun{bef}
\DefMacroFun{bep}
\DefMacroFun{prg}
\DefMacroFun{spr}
\DefMacroFun{chg}
\DefMacroFun{num}[n]
\DefMacroFun{bnd}[\flat]
\DefMacroFun{dst}[\#]

\DefMacroFun{chgmsr}
\DefMacroFun{chgspr}

% ...

\usepackage{tikz}
\usetikzlibrary{arrows,shapes,calc,patterns}

\usepackage{float}
\usepackage{wrapfig}

\usepackage{pgf}
\usepackage{pgfplots}
\usepackage{pgfplotstable}

\tikzstyle{every node} =
  [draw = none, fill = white, thin]
\tikzstyle{every edge} +=
  [black, thick]

\tikzstyle{noall} =
  [draw = none, fill = none]
\tikzstyle{nodraw} =
  [draw = none, fill = white]
\tikzstyle{nofill} =
  [draw = black, fill = none]

\tikzstyle{nonode} =
  [draw = white]
\tikzstyle{cnode} =
  [circle, draw = black]
\tikzstyle{snode} =
  [regular polygon, regular polygon sides = 4, draw = gray!50]
\tikzstyle{lnode} =
  [diamond, draw = gray!75]
\tikzstyle{pnode} =
  [regular polygon, regular polygon sides = 5, draw = gray]

\AfterEndPreamble
  {

  \newcommand{\figexm}
    {
    \begin{center}
    \footnotesize
    \scalebox{0.600}[0.680]
    {
    \begin{tikzpicture}
      [node distance = 5em, bend angle = 22.5, inner sep = 0.3em, minimum size =
      3.2em]

      \tikzset{every loop/.style = {max distance = 1.5em}}

      \node [cnode]
            (A)
            []
            {$\aSym/\mathbf{2}$};
      \node [snode]
            (B)
            [right of = A]
            {$\bSym/\mathbf{6}$};
      \node [snode]
            (C)
            [right of = B]
            {$\cSym/\mathbf{4}$};
      \node [cnode]
            (D)
            [right of = C]
            {$\dSym/\mathbf{1}$};
      \node [cnode]
            (E)
            [below of = A]
            {$\eSym/\mathbf{1}$};
      \node [snode]
            (F)
            [below of = B]
            {$\fSym/\mathbf{1}$};
      \node [cnode]
            (G)
            [below of = C]
            {$\gSym/\mathbf{3}$};
      \node [snode]
            (H)
            [below of = D]
            {$\hSym/\mathbf{2}$};
      \node [minimum size = 2.5em]
            (L)
            [below right of = B, xshift = -0.75em, yshift = -5.25em]
            {\large A parity game $\GamName$.};

      \path[->]
        (A) edge  [bend angle = 30, bend left]
                  (D)
            edge  [bend left]
                  (E)
        (B) edge  []
                  (F)
            edge  []
                  (G)
        (C) edge  []
                  (B)
            edge  []
                  (D)
        (D) edge  [loop right]
                  ()
            edge  [bend left]
                  (H)
        (E) edge  [bend left]
                  (A)
            edge  []
                  (B)
        (F) edge  [loop right]
                  ()
        (G) edge  []
                  (D)
            edge  []
                  (C)
        (H) edge  [bend left]
                  (D)
            edge  []
                  (G)
        ;

    \end{tikzpicture}
    }
  \end{center}
  }

  \newcommand{\figsima}
    {
    \begin{center}
    \footnotesize
    \scalebox{0.600}[0.680]
    {
    \begin{tikzpicture}
      [node distance = 5em, bend angle = 22.5, inner sep = 0.3em, minimum size =
      3.2em]

      \tikzset{every loop/.style = {max distance = 1.5em}}

      \node [noall]
            (A)
            []
            {\Large $\msrElm[\aSym]$};
      \node [noall]
            (B)
            [right of = A]
            {\Large $\msrElm[\bSym]$};
      \node [noall]
            (C)
            [right of = B]
            {\Large $\msrElm[\cSym]$};
      \node [noall]
            (D)
            [right of = C]
            {\Large $\bot$};
      \node [noall]
            (E)
            [below of = A]
            {\Large $\bot$};
      \node [noall]
            (F)
            [below of = B]
            {\Large $\bot$};
      \node [noall]
            (G)
            [below of = C]
            {\Large $\bot$};
      \node [noall]
            (H)
            [below of = D]
            {\Large $\msrElm[\hSym]$};
      \node [minimum size = 2.5em]
            (L)
            [below right of = B, xshift = -1.25em, yshift = -5.25em]
            {\large $\mathbf{(1):\mfElm[1] = \prgFun[\bot](\mfElm[\bot]) =
            \prgFun[+](\mfElm[1])}$};

      \path[->]
        (A) edge  [bend angle = 30, bend left, blue]
                  (D)
            edge  [bend left, blue]
                  (E)
        (B) edge  [blue]
                  (F)
            edge  [blue]
                  (G)
        (C) edge  [blue]
                  (D)
        (D) edge  [bend left, dashed, red]
                  (H)
        (E) edge  [bend left, dashed, red]
                  (A)
            edge  [dashed, red]
                  (B)
        (F) edge  [loop right, blue]
                  ()
        (G) edge  [dashed, red]
                  (C)
        (H) edge  [bend left, blue]
                  (D)
            edge  [blue]
                  (G)
        ;

      \begin{scope}
        [very thick, fill = gray, draw = gray, fill opacity = 0.075]

        \filldraw
          ($(A) + (0.00, 0.65)$)
            to [out = 0, in = 180]
          ($(C) + (0.00, 0.65)$)
            to [out = 0, in = 90]
          ($(C) + (0.65, 0.00)$)
            to [out = 270, in = 180]
          ($(H) + (0.00, 0.65)$)
            to [out = 0, in = 90]
          ($(H) + (0.50, 0.00)$)
            to [out = 270, in = 0]
          ($(H) + (-0.00, -0.50)$)
            to [out = 180, in = 270]
          ($(H) + (-0.65, 0.00)$)
            to [out = 90, in = 0]
          ($(C) + (0.00, -0.65)$)
            to [out = 180, in = 0]
          ($(A) + (0.00, -0.65)$)
            to [out = 180, in = 270]
          ($(A) + (-0.50, 0.0)$)
            to [out = 90, in = 180]
          ($(A) + (0.00, 0.65)$)
          ;

      \end{scope}

    \end{tikzpicture}
    }
  \end{center}
  }

  \newcommand{\figsimb}
    {
    \begin{center}
    \footnotesize
    \scalebox{0.600}[0.680]
    {
    \begin{tikzpicture}
      [node distance = 5em, bend angle = 22.5, inner sep = 0.3em, minimum size =
      3.2em]

      \tikzset{every loop/.style = {max distance = 1.5em}}

      \node [noall]
            (A)
            []
            {\Large $\msrElm[\aSym]$};
      \node [noall]
            (B)
            [right of = A]
            {\Large $\msrElm[\bSym]$};
      \node [noall]
            (C)
            [right of = B]
            {\Large $\msrElm[\cSym]$};
      \node [noall]
            (D)
            [right of = C]
            {\Large $\msrElm[\dSym]$};
      \node [noall]
            (E)
            [below of = A]
            {\Large $\msrElm[\eSym]$};
      \node [noall]
            (F)
            [below of = B]
            {\Large $\bot$};
      \node [noall]
            (G)
            [below of = C]
            {\Large $\msrElm[\gSym]$};
      \node [noall]
            (H)
            [below of = D]
            {\Large $\msrElm[\hSym]$};
      \node [minimum size = 2.5em]
            (L)
            [below right of = B, xshift = -1em, yshift = -5.25em]
            {\Large $\mathbf{(2):\mfElm[2] = \prgFun[\bot](\mfElm[1])}$};

      \path[->]
        (A) edge  [bend angle = 30, bend left, dashed, red]
                  (D)
            edge  [bend left, dashed, red]
                  (E)
        (B) edge  [blue]
                  (F)
        (C) edge  [blue]
                  (D)
        (D) edge  [bend left, blue]
                  (H)
        (E) edge  [blue]
                  (B)
        (F) edge  [loop right, blue]
                  ()
        (G) edge  [blue]
                  (C)
        (H) edge  [bend left, dashed, red]
                  (D)
            edge  [dashed, red]
                  (G)
        ;

      \begin{scope}
        [very thick, fill = gray, draw = gray, fill opacity = 0.075]

        \filldraw
          ($(A) + (0.00, 0.65)$)
            to [out = 0, in = 180]
          ($(D) + (0.00, 0.65)$)
            to [out = 0, in = 90]
          ($(D) + (0.50, 0.00)$)
            to [out = 270, in = 90]
          ($(H) + (0.50, 0.00)$)
            to [out = 270, in = 0]
          ($(H) + (-0.00, -0.65)$)
            to [out = 180, in = 0]
          ($(G) + (-0.00, -0.65)$)
            to [out = 180, in = 270]
          ($(G) + (-0.50, 0.00)$)
            to [out = 90, in = 0]
          ($(B) + (0.00, -0.65)$)
            to [out = 180, in = 90]
          ($(E) + (0.50, 0.00)$)
            to [out = 270, in = 0]
          ($(E) + (0.00, -0.65)$)
            to [out = 180, in = 270]
          ($(E) + (-0.50, 0.00)$)
            to [out = 90, in = 270]
          ($(A) + (-0.50, 0.00)$)
            to [out = 90, in = 180]
          ($(A) + (0.00, 0.65)$)
          ;

      \end{scope}

    \end{tikzpicture}
    }
  \end{center}
  }

  \newcommand{\figsimc}
    {
    \begin{center}
    \footnotesize
    \scalebox{0.600}[0.680]
    {
    \begin{tikzpicture}
      [node distance = 5em, bend angle = 22.5, inner sep = 0.3em, minimum size =
      3.2em]

      \tikzset{every loop/.style = {max distance = 1.5em}}

      \node [noall]
            (A)
            []
            {\Large $\msrElm[\aSym]'$};
      \node [noall]
            (B)
            [right of = A]
            {\Large $\msrElm[\bSym]$};
      \node [noall]
            (C)
            [right of = B]
            {\Large $\msrElm[\cSym]'$};
      \node [noall]
            (D)
            [right of = C]
            {\Large $\msrElm[\dSym]'$};
      \node [noall]
            (E)
            [below of = A]
            {\Large $\msrElm[\eSym]$};
      \node [noall]
            (F)
            [below of = B]
            {\Large $\bot$};
      \node [noall]
            (G)
            [below of = C]
            {\Large $\msrElm[\gSym]'$};
      \node [noall]
            (H)
            [below of = D]
            {\Large $\msrElm[\hSym]'$};
      \node [minimum size = 2.5em]
            (L)
            [below right of = B, xshift = -1em, yshift = -5.25em]
            {\large $\mathbf{(3):\mfElm[3] = \prgFun[+](\mfElm[2])}$};

      \path[->]
        (A) edge  [bend angle = 30, bend left, blue]
                  (D)
        (B) edge  [blue]
                  (F)
        (C) edge  [blue]
                  (B)
        (D) edge  [bend left, blue]
                  (H)
        (E) edge  [bend left, dashed, red]
                  (A)
            edge  [blue]
                  (B)
        (F) edge  [loop right, blue]
                  ()
        (G) edge  [blue]
                  (C)
        (H) edge  [blue]
                  (G)
        ;

      \begin{scope}
        [very thick, fill = gray, draw = gray, fill opacity = 0.075]

        \filldraw
          ($(A) + (0.00, 0.65)$)
            to [out = 0, in = 180]
          ($(D) + (0.00, 0.65)$)
            to [out = 0, in = 90]
          ($(D) + (0.50, 0.00)$)
            to [out = 270, in = 90]
          ($(H) + (0.50, 0.00)$)
            to [out = 270, in = 0]
          ($(H) + (-0.00, -0.65)$)
            to [out = 180, in = 0]
          ($(G) + (-0.00, -0.65)$)
            to [out = 180, in = 270]
          ($(G) + (-0.50, 0.00)$)
            to [out = 90, in = 0]
          ($(B) + (0.00, -0.65)$)
            to [out = 180, in = 90]
          ($(E) + (0.50, 0.00)$)
            to [out = 270, in = 0]
          ($(E) + (0.00, -0.65)$)
            to [out = 180, in = 270]
          ($(E) + (-0.50, 0.00)$)
            to [out = 90, in = 270]
          ($(A) + (-0.50, 0.00)$)
            to [out = 90, in = 180]
          ($(A) + (0.00, 0.65)$)
          ;

      \end{scope}

    \end{tikzpicture}
    }
  \end{center}
  }

  \newcommand{\figexp}
    {
      \begin{center}
      \footnotesize
      \scalebox{0.750}[0.750]
      {
        \begin{tikzpicture}
        \begin{axis}
          [
            width = 0.60\textwidth, height = 0.45\textwidth,
            xmin = 50, xmax = 500000,
            ymin = 0, ymax = 120,
            xmode=log,
            x axis line style = -, y axis line style = -,
            ymajorgrids = true,
            xlabel = Number of positions, ylabel = Time (s),
            y label style = {at = {(axis description cs:0.02,0.5)}},
            legend pos = outer north east,
            legend entries = {FPJ, TL, QDPM, SSPM, SPM, QPT, PP, Rec}
          ]

          \addplot [blue!60!black, solid, line width = 1pt, mark = pentagon,
            mark options = solid, mark size = 2.5] table [x index = 0, y index
            = 8] {clusteresultest};

          \addplot [green!50!black, solid, line width = 1pt, mark = o, mark
            options = solid, mark size = 2] table [x index = 0, y index = 7]
            {clusteresultest};

          \addplot [black, solid, line width = 0.75pt, mark = star, mark
            options = solid, mark size = 2.5] table [x index = 0, y index = 6]
            {clusteresultest};

          \addplot [red!50!black, solid, line width = 0.75pt, mark = square, mark
            options = solid, mark size = 2.5] table [x index = 0, y index = 5]
            {clusteresultest};

          \addplot [pink!50!black, solid, line width = 0.75pt, mark = triangle, mark
            options = solid, mark size = 2.5] table [x index = 0, y index = 4]
            {clusteresultest};

          \addplot [orange!50!black, solid, line width = 0.75pt, mark = diamond, mark
            options = solid, mark size = 2.5] table [x index = 0, y index = 3]
            {clusteresultest};

          \addplot [yellow!50!black, solid, line width = 0.75pt, mark = oplus, mark
            options = solid, mark size = 2.5] table [x index = 0, y index = 2]
            {clusteresultest};

          \addplot [brown!50!black, solid, line width = 0.75pt, mark = otimes, mark
            options = solid, mark size = 2.5] table [x index = 0, y index = 1]
            {clusteresultest};

        \end{axis}
      \end{tikzpicture}
      }
    \end{center}
    }

  }

\usepackage{multirow}

\usepackage{makecell}

% ...

\AfterEndPreamble
  {

  \newcommand{\tabsim}
    {
    \setlength{\tabcolsep}{0pt}
    \begin{tabular}{c|c|c|c}
      \begin{minipage}[t]{0.25\textwidth}
        \vspace{-.30em}
        \figexm
        \vspace{-0.5em}
      \end{minipage} &
      \begin{minipage}[t]{0.25\textwidth}
        \vspace{-.30em}
        \figsima
        \vspace{-0.5em}
      \end{minipage} &
      \begin{minipage}[t]{0.25\textwidth}
        \vspace{-0.30em}
        \figsimb
        \vspace{-0.5em}
      \end{minipage} &
      \begin{minipage}[t]{0.25\textwidth}
        \vspace{-0.30em}
        \figsimc
        \vspace{-1em}
      \end{minipage}
    \end{tabular}
    }

  \newcommand{\tabper}
    {
    \begin{tabular}{|l||r|r|r|r|r|r||r|r|}
      \hline
      \multicolumn{1}{|c||}{} & \multicolumn{6}{c||}{}            &
      \multicolumn{2}{c|}{} \\[-0.975em]
      \multicolumn{1}{|c||}{} & \multicolumn{6}{c||}{Exponential} &
      \multicolumn{2}{c|}{Quasi Polynomial} \\
      \hline & & & & & & & & \\[-0.975em]
      Benchmarks                & Rec & PP  & TL  & FPJ & SPM
      & {\bf QDPM}  & SSPM  & QPT \\
      \hline & & & & & & & & \\[-0.975em]
      Two Counters~\cite{Dij19} & 20  & 107 & 11  & 26  & 7
      & 85    & 4     & 5 \\
      QPT~\cite{FJSSW17}        & $>\!\!10^3\, [0s]$
                                      & $>\!\!10^3\, [0s]$
                                            & $>\!\!10^3\, [0s]$
                                                  & $>\!\!10^3\, [0s]$
                                                        & $>\!\!10^3\, [0s]$
      & $>\!\!10^3\, [0s]$
              & $>\!\!10^3\, [0.41s]$
                      & 30 \\
      Gazda's wc~\cite{Gaz16}   & 35  & 35  & $>\!\!10^3\, [0.01s]$
                                                  & 37  & 380
      & $>\!\!10^3\, [0.67s]$
              & 22    & 29 \\
      DP~\cite{BDM18a}          & 36  & 23  & $>\!\!10^3\, [0.19s]$
                                                  & 38  & 20
      & $>\!\!10^3\, [0s]$
              & $>\!\!10^3\, [0.17s]$
                      & 28 \\
      Divide\&Impera~\cite{BDM20} & 17  & 91  & $>\!\!10^3\, [27.75s]$
                                                  & 16  & 7
      & $>\!\!10^3\, [7.45s]$
              & 173   & 4 \\
      \hline
    \end{tabular}
    }

  }

\usepackage[ruled,vlined]{algorithm2e}

\DontPrintSemicolon

\SetKw{Signature}{signature}
\SetKwFor{Function}{function}{}{}
\SetKwFor{Procedure}{procedure}{}{}

\SetKwFor{Let}{let}{in}{}

\AfterEndPreamble
  {

  \newcommand{\algprg}
    {
    \begin{algorithm}[H]
      \caption{\label{alg:prg}Operator $\prgFun[+]$}
      \Signature{$\prgFun[+] \colon \MFSet \to \MFSet$} \;
      \Function{$\prgFun[+](\mfElm)$}
        {
        \nl $\QSet \gets \denot{\mfElm}[+]$ \;
        \nl \While{$\escFun(\mfElm, \QSet) \neq \emptyset$}
          {
          \nl $\ESet \gets \bepFun(\mfElm, \QSet)$ \;
          \nl $\mfElm \gets \liftFun(\mfElm, \ESet, \dual{\QSet})$ \;
          \nl $\QSet \gets \QSet \setminus \ESet$ \;
          }
        \nl $\mfElm \gets {\mfElm}[\QSet \mapsto \top]$ \;
        \nl \Return $\mfElm$ \;
        }
    \end{algorithm}
    }

  }

\hypersetup
  {
  pdftitle  = {From Quasi-Dominions to Progress Measures},
  pdfauthor = {M. Benerecetti, D. Dell'Erba, M.Faella, F. Mogavero}
  }

\begin{document}

  \title{From Quasi-Dominions to Progress Measures}
  \def\titlerunning{From Quasi-Dominions to Progress Measures}

  \author{Massimo Benerecetti$^{\mathsf{1}}$ \& Daniele Dell'Erba$^{\mathsf{2}}$
    \& Marco Faella$^{\mathsf{1}}$ \& Fabio Mogavero$^{\mathsf{1}}$
    \institute{$^{\mathsf{1}}$Universit\`a degli Studi di Napoli Federico II \&
    $^{\mathsf{2}}$University of Liverpool}}

  \def\authorrunning{M. Benerecetti, D. Dell'Erba, M. Faella, \& F. Mogavero}

  \maketitle

%%****************************************************************************%%
%%                                                                            %%
%% Article Title                                                              %%
%%                                                                            %%
%% Abstract.tex                                                               %%
%%                                                                            %%
%% Revision 0                                                                 %%
%%                                                                            %%
%% Copyright (C) 2020, Massimo Benerecetti, Daniele Dell'Erba, Marco Faella,  %%
%%                     and Fabio Mogavero.                                    %%
%% All rights reserved.                                                       %%
%%                                                                            %%
%%****************************************************************************%%

% Begin of file Abstract.tex

\begin{abstract}

We revisit the approaches to the solution of \emph{parity games} based on
\emph{progress measures} and show how the notion of \emph{quasi dominions} can
be integrated with those approaches.
The idea is that, while progress measure based techniques typically focus on one
of the two players, little information is gathered on the other player during
the solution process.
Adding quasi dominions provides additional information on this player that can
be leveraged to greatly accelerate convergence to a progress measure.
To accommodate quasi dominions, however, non trivial refinements of the approach
are necessary.
In particular, we need to introduce a novel notion of measure and a new method
to prove correctness of the resulting solution technique.

\end{abstract}

% End of file Abstract.tex

%%****************************************************************************%%
%%                                                                            %%
%% Article Title                                                              %%
%%                                                                            %%
%% Introduction.tex                                                           %%
%%                                                                            %%
%% Revision 0                                                                 %%
%%                                                                            %%
%% Copyright (C) 2020, Massimo Benerecetti, Daniele Dell'Erba, Marco Faella,  %%
%%                     and Fabio Mogavero.                                    %%
%% All rights reserved.                                                       %%
%%                                                                            %%
%%****************************************************************************%%

% Begin of file Introduction.tex

\begin{section}{Introduction}

Parity games are two-player infinite-duration games on graphs, which play a
crucial role in various fields of theoretical computer science.
These are games played on graphs, whose nodes, called \emph{positions}, are
labelled with natural numbers, called \emph{priorities}, and controlled by one
of two players: player $\PlrSym$ and player $\OppSym$.
Each player can choose edges, called \emph{moves}, when the game is at one of
its positions.
The goal of player $\OppSym$ is to force a play $\pthElm$, namely an infinite path in the
underlying graph, such that the maximal priority occurring infinitely often
along $\pthElm$ is of odd parity.
If such a play cannot be enforced by player $\OppSym$, player $\PlrSym$ wins the
game.
In this case, player $\PlrSym$ can indeed force a play whose maximal recurring
priority is even.

Finding efficient algorithms to solve these games in practice is a core problem
in formal verification and reactive synthesis, as it leads to efficient
solutions of the model-checking and satisfiability problems of expressive
temporal logics.
These algorithms can, indeed, be used as back-ends of satisfiability and
model-checking procedures~\cite{EJ88,EJS93,KVW00}.
In particular, the solution problem for these games has been proved linear-time
interreducible with the model-checking problem for the \emph{modal
\MC}~\cite{EJS93} and it is closely related to other games of infinite duration,
such as \emph{mean payoff}~\cite{EM79,GKK88}, \emph{discounted
payoff}~\cite{ZP96}, \emph{simple stochastic}~\cite{Con92}, and
\emph{energy}~\cite{CDHR10} games.
Parity games are also central to several techniques employed in automata
theory~\cite{Mos84,EJ91,KV98,GTW02}.
Besides the practical importance, parity games are also interesting from a
computational complexity point of view, since their solution problem is one of
the few inhabitants of the \UPTime $\cap$ \CoUPTime class~\cite{Jur98}.
That result improves the \NPTime $\cap$ \CoNPTime membership~\cite{EJS93}, which
easily follows from the property of \emph{memoryless
determinacy}~\cite{EJ91,Mos91}.

A number of quite different approaches to solve parity games have been proposed
in the literature that exhibit quite different characteristics.
Typically, the most efficient ones in practice are those based on game
decomposition, such as the Recursive Algorithm~\cite{McN93,Zie98}, Priority
Promotion~\cite{BDM18,BDM18a,BDM16b}, and Tangle Learning~\cite{Dij18a}, which,
however, usually suffer from poor worst-case bounds.
On the other hand, the approaches based on progress measures~\cite{Jur00,GW15}
often lead to good worst-case behaviours~\cite{JL17,FJSSW17}, but typically
perform very poorly in practice.
The main reason for this inefficiency resides in the fact that those algorithms
iteratively explore a space of functions assigning some value, called
\emph{measure}, to each position in the game.
At each iteration, the measures of some of the positions may increase and when
they become stationary for all the positions, a fixpoint is reached and a
solution can be extracted from the resulting measures.
In order to guarantee correctness, measures are allowed to increase very slowly,
which often leads to slow convergence to a solution and makes these approaches
less appealing in practical contexts.
The slow growth is the result of a uniform measure update policy for one of the
two players, specifically player $\OppSym$, which only allows for a minimal
measure increase for each of its positions that must be updated.

In this work we show that the update policy can be considerably improved upon,
without sacrificing correctness.
Instead of relying on the minimal increase policy to ensure soundness of measure
updates, we propose an approach that brings \emph{quasi dominions} into the
equation.
Note that the same idea has been implemented for Mean Payoff Games in~\cite{BDM20a}
where the resulting algorithm proved to be very efficient.
Informally, a quasi dominion is a set of positions from which one of the
players, say player $\PlrSym$, can win the game as long as the opponent, player
$\OppSym$, chooses not to escape from that set.
As such, the notion is not new and is at the very heart of the Priority
Promotion algorithms~\cite{BDM16}.
The idea is to leverage quasi dominions to justify a larger, but still
sound, increase in the measure for positions controlled by player $\OppSym$.
The crucial observation is that player $\OppSym$ surely loses from any position
of a quasi dominion for player $\PlrSym$, unless it can escape that set by
taking some exiting move.
Therefore, player $\OppSym$ can safely increase the measure of the escaping
position according to the exiting moves chosen, regardless of the fact that the
increase may not be minimal.
In this way, we are able to skip lower measures and jump directly to measures
that would be reached anyway, albeit with a number of iterations that is usually
much higher.

The integration of progress measures and quasi dominions, however, requires
\begin{inparaenum}[(i)]
  \item
    a richer form of measure, able to encode additional information that allows
    us to identify quasi dominions in the game, and
  \item
    a new update algorithm that takes this information into account when
    increasing the measures.
\end{inparaenum}

The main contributions of this paper can be summarised as follows:
\begin{inparaenum}[(a)]
  \item
    a novel solution algorithm for parity games based on the integration of
    progress measures and quasi dominions;
  \item
    the experimental results showing an improvement on the performance of orders
    of magnitude \wrt the classic and quasi-polynomial small progress measure
    algorithm;
  \item
    the present approach can pave the way to practically efficient
    quasi-polynomial algorithms based on the integration of succinct progress
    measures, such as those in~\cite{JL17,FJSSW17}.
\end{inparaenum}
%
%\indent
Due to space constraints, most proofs are reported in the appendix.

\end{section}

% End of file Introduction.tex

%%****************************************************************************%%
%%                                                                            %%
%% Article Title                                                              %%
%%                                                                            %%
%% Preliminaries.tex                                                          %%
%%                                                                            %%
%% Revision 0                                                                 %%
%%                                                                            %%
%% Copyright (C) 2020, Massimo Benerecetti, Daniele Dell'Erba, Marco Faella,  %%
%%                     and Fabio Mogavero.                                    %%
%% All rights reserved.                                                       %%
%%                                                                            %%
%%****************************************************************************%%

% Begin of file Preliminaries.tex

\begin{section}{Preliminaries}
  \label{sec:prl}

  A two-player turn-based \emph{arena} is a tuple $\AName = \tuplec
  {\PosSet[\PlrSym]} {\PosSet[\OppSym]} {\MovRel}$, with $\PosSet[\PlrSym] \cap
  \PosSet[\OppSym] = \emptyset$ and $\PosSet \defeq \PosSet[\PlrSym] \cup
  \PosSet[\OppSym]$, such that $\tupleb{\PosSet}{\MovRel}$ is a finite directed
  graph without sinks.
  $\PosSet[\PlrSym]$ (\resp, $\PosSet[\OppSym]$) is the set of positions of
  player $\PlrSym$ (\resp, $\OppSym$) and $\MovRel \subseteq \PosSet \times
  \PosSet$ is a left-total relation describing all possible moves.
  A \emph{path} in $\VSet \subseteq \PosSet$ is a finite or infinite sequence
  $\pthElm \in \PthSet(\VSet)$ of positions in $\VSet$ compatible with the move
  relation, \ie, $(\pthElm_{i}, \pthElm_{i + 1}) \in \MovRel$, for all $0 \leq i < \card{\pthElm} - 1$.
  The set $\FPthSet(\posElm)$ contains all the finite paths originating at the position $\posElm$.
  For a finite path $\pthElm$, with $\lst{\pthElm}$ we denote the last position
  of $\pthElm$.
  Finally, $\SPthSet(\posElm, \VSet)$ is the set of simple paths in
  $\FPthSet(\posElm)$ that are completely composed of positions in $\VSet$.
  A positional \emph{strategy} for player $\alpha \in \{ \PlrSym, \OppSym \}$ on
  $\VSet \subseteq \PosSet$ is a function $\strElm[\alpha] \in
  \StrSet[\alpha](\VSet) \defeq (\VSet \cap \PosSet[\alpha]) \to \PosSet$,
  mapping each $\alpha$-position $\posElm$ in $\VSet$ to a position
  $\strElm[\alpha](\posElm)$ compatible with the move relation, \ie, $(\posElm,
  \strElm[\alpha](\posElm)) \in \MovRel$.
  With $\StrSet[\alpha](\VSet)$ we denote the set of all $\alpha$-strategies
  on $\VSet$.
  Given an $\alpha$-strategy $\strElm[\alpha] \in \StrSet[\alpha](\VSet)$ and a
  set of positions $\USet \subseteq \PosSet$, the operator $\strElm[\alpha]
  \downarrow \USet$ restricts $\strElm[\alpha]$ to the positions in $\VSet \cap
  \USet$.
  A \emph{play} in $\VSet \subseteq \PosSet$ from a position $\posElm \in \VSet$
  \wrt a pair of strategies $(\strElm[\PlrSym], \strElm[\OppSym]) \in
  \StrSet[\PlrSym](\VSet) \times \StrSet[\OppSym](\VSet)$, called
  \emph{$((\strElm[\PlrSym], \strElm[\OppSym]), \posElm)$-play}, is a path
  $\pthElm \in \PthSet(\VSet)$ such that $(\pthElm)_{0} = \posElm$ and, for all
  $0 \leq i < \card{\pthElm} - 1$, if $(\pthElm)_{i} \in \PosSet[\PlrSym]$
  then $(\pthElm)_{i + 1} = \strElm[\PlrSym]((\pthElm)_{i})$ else $(\pthElm)_{i
  + 1} = \strElm[\OppSym]((\pthElm)_{i})$.
  The \emph{play function} $\playFun : (\StrSet[\PlrSym](\VSet) \times
  \StrSet[\OppSym](\VSet)) \times \VSet \to \PthSet(\VSet)$ returns, for each
  position $\posElm \in \VSet$ and pair of strategies $(\strElm[\PlrSym],
  \strElm[\OppSym]) \in \StrSet[\PlrSym](\VSet) \times \StrSet[\OppSym](\VSet)$,
  the maximal $((\strElm[\PlrSym], \strElm[\OppSym]), \posElm)$-play
  $\playFun((\strElm[\PlrSym], \strElm[\OppSym]), \posElm)$.
  A path $\pthElm \in \PthSet(\posElm)$ is called a \emph{$(\strElm[\alpha],
  \posElm)$-play in $\VSet$}, if $\pthElm = \playFun((\strElm[\PlrSym],
  \strElm[\OppSym]), \posElm)$, for some $\strElm[\dual{\alpha}] \in
  \StrSet[\dual{\alpha}](\VSet)$.

  A \emph{parity game} is a tuple $\GamName = \tuplec {\AName} {\PrtSet}
  {\prtFun} \in \PG$, where $\AName$ is an arena, $\PrtSet \subset \SetN$ is a
  finite set of priorities, and $\prtFun : \PosSet \to \PrtSet$ is a
  \emph{priority function} assigning a priority to each position.
  The priority function can be naturally extended to games and paths as follows:
  $\prtFun(\GamName) \defeq \max[\posElm \in \PosSet] \, \prtFun(\posElm)$; for
  a path $\pthElm \in \PthSet$, we set $\prtFun(\pthElm) \defeq \max_{0 \leq i < \card{\pthElm}}
  \, \prtFun((\pthElm)_{i})$, if $\pthElm$ is finite,
  and $\prtFun(\pthElm) \defeq \limsup_{i \in \SetN} \prtFun((\pthElm)_{i})$,
  otherwise.
  A set of positions $\VSet \subseteq \PosSet$ is an $\alpha$-\emph{dominion},
  with $\alpha \in \{ \PlrSym, \OppSym \}$, if there exists an $\alpha$-strategy
  $\strElm[\alpha] \in \StrSet[\alpha](\VSet)$ such that, for all
  $\dual{\alpha}$-strategies $\strElm[\dual{\alpha}] \in
  \StrSet[\dual{\alpha}](\VSet)$ and positions $\posElm \in \VSet$, the
  induced play $\pthElm = \playFun((\strElm[\PlrSym], \strElm[\OppSym]),
  \posElm)$ is infinite and $\prtFun(\pthElm) \equiv_{2} \alpha$.
  In other words, $\strElm[\alpha]$ only induces on $\VSet$ infinite plays whose
  maximal priority visited infinitely often has parity $\alpha$.
  By $\GamName \!\setminus\! \VSet$ we denote the maximal subgame of $\GamName$
  with set of positions $\PosSet'$ contained in $\PosSet \!\setminus\! \VSet$
  and move relation $\MovRel'$ equal to the restriction of $\MovRel$ to
  $\PosSet'$.

\end{section}

% End of file Preliminaries.tex

%%****************************************************************************%%
%%                                                                            %%
%% Article Title                                                              %%
%%                                                                            %%
%% SectionI.tex                                                               %%
%%                                                                            %%
%% Revision 0                                                                 %%
%%                                                                            %%
%% Copyright (C) 2020, Massimo Benerecetti, Daniele Dell'Erba, Marco Faella,  %%
%%                     and Fabio Mogavero.                                    %%
%% All rights reserved.                                                       %%
%%                                                                            %%
%%****************************************************************************%%

% Begin of file SectionI.tex

\begin{section}{Solving Parity Games via Progress Measures}
  \label{sec:prgmsr}

  The abstract notion of \emph{progress measure}~\cite{Kla91} has been
  introduced as a way to encode \emph{global properties} on paths of a graph by
  means of simpler \emph{local properties} of adjacent vertexes, \ie, of edges.
  In particular, this notion has been successfully employed in the literature,
  \eg, for the solution of automata theory~\cite{SE89,Kla91a,KK91,Kla92,Kla94}
  and program verification~\cite{Kla92a,Var96a} problems.

  In the context of parity games~\cite{Jur00}, the graph property of interest,
  called \emph{parity property}, asserts that, along every path in the graph,
  the maximal priority occurring infinitely often is of odd parity.
  More precisely, in game theoretic terms, a \emph{parity progress measure}
  witnesses the existence of a strategy $\strElm$ for one of the two players,
  from now on player $\OppSym$, such that each path in the graph induced by
  fixing $\strElm$ satisfies the desired property, where the graph induced by
  that strategy is obtained from the game arena by removing all the moves
  exiting from position owned by player $\OppSym$, except those ones specified
  by $\strElm$ itself.
  A parity progress measure associates with each vertex of the underlying graph
  a value, called \emph{parity measure} (or simply \emph{measure}, for short),
  taken from some totally-ordered set.
  Measures are thus associated with positions in the game and the measure
  $\msrElm$ of a position $\posElm$ can intuitively be interpreted as a local
  assessment of \emph{how far} $\posElm$ is from satisfying the parity property,
  with the maximal value $\msrElm = \top$ denoting failure in the satisfaction
  of the property for $\posElm$.
  More precisely, a progress measure implicitly identifies a strategy $\strElm$
  with the following characteristic: in the graph induced by $\strElm$, along
  every path, measures cannot increase and they strictly decrease when passing
  through an even-priority position.
  This ensures that every path eventually gets trapped into a cycle whose
  maximal priority is odd.

  In order to obtain a progress measure, we start from some well-behaved
  assignment of measures to positions of the game.  The local information
  encoded by these measures is, then, propagated back along the edges of the
  underlying graph so as to associate with each position the information on the
  priorities occurring along the plays of some finite length starting at that
  position.
  The propagation process is performed by means of a low-level measure-update
  operator, called \emph{stretch} operator $+$.
  The operator computes the contribution that a given position $\posElm$ would
  provide to a given measure $\msrElm$.
  Consider, for instance, a position $\posElm$ that has an adjacent position
  $\uposElm$ with measure $\msrElm$.
  Then $\msrElm + \posElm$ is the measure that $\posElm$ would obtain by
  choosing the move leading to $\uposElm$, namely, the one obtained by
  augmenting $\eta$ with the contribution of (the priority of) position
  $\posElm$.
  The intuition is that when $\posElm$ is an even-priority position, the
  augmented measure $\msrElm + \posElm$ strictly increases, moving further away
  from the priority condition.

  The process described above terminates when no position can be pushed further
  away from the property.
  More specifically, each even position %of player $\PlrSym$ 
  has to strictly dominate
  the measures obtainable through all, respectively one of, its adjacent
  positions, depending on whether that position belongs to the player $\PlrSym$
  or to the player $\OppSym$, respectively.
  Similarly, each odd position %of player $\OppSym$ 
  must have measures no lower than
  those obtainable through all, respectively one of, its adjacent moves, again
  depending on the owner of the position.
  When this happens, the positions with measure $\top$ are the ones from which
  player $\PlrSym$ wins the parity game, while the remaining ones are those from
  which the opponent can win, by simply forcing plays of non-increasing
  measures.
  The measures currently associated with this second set of positions form a
  progress measure for the game.

  Different notions of parity measures have been proposed in the literature,
  see, \eg, \cite{Jur00,GW15,JL17,FJSSW17}.
  In this section we introduce an abstract concept of measure space and progress
  measure.
  All the progress measure based approaches in the literature, including the one
  presented in this article, can be viewed as instantiations of this abstract
  schema.

  \begin{subsection}{Measure-Function Spaces}

    As mentioned above, techniques based on progress measures rely on attaching,
    at each step of the computation, suitable information to all positions in
    the game and updating it until a fixpoint is reached.
    The piece of information associated with every single position is called the
    current \emph{measure} of that position, whereas the set of all possible
    measures is called a \emph{measure space}.
    Such a space is a totally ordered set, with minimum and maximum elements,
    and provides the two special operations of \emph{truncation} and
    \emph{stretch} that evaluate and update the measure of a given position \wrt
    another position.
    Intuitively, the truncation operator $\rst_{\posElm}$ disregards the
    contribution to a given measure that is due to positions with priority lower
    than that of $\posElm$ along the explored finite plays.
    The stretch operator $+$, introduced in the previous paragraph, propagates the
    contribution that the position would provide to a given measure.

    These two operators essentially embed the
    semantics of the parity property into the propagation operation that sits at
    the basis of the computation of a progress measure.
    At the abstract level, canonical instances of these operators can be any
    functions that preserve the maximum element and the order, except for
    possibly mapping different measures onto the same one.
    All these requirements are formalised by the following definition.

    \begin{definition}[Measure Space]
      \label{def:msrspc}
      A \emph{measure space} is a mathematical structure $\MsrName \defeq
      \tuplef {\MsrSet} {<} {\bot} {\top} {\rst} {+}$, whose components enjoy
      the following properties:
      \begin{enumerate}
        \item\label{def:msrspc(ord)}
          $\tupled {\MsrSet} {<} {\bot} {\top}$ is a \emph{strict total order}
          with \emph{minimum} and \emph{maximum} on elements referred to as
          \emph{measures};
        \item\label{def:msrspc(trn)}
          the function $\rst \colon \MsrSet \times \PosSet \to \MsrSet$, called
          \emph{truncation operator}, maps a measure $\msrElm \in \MsrSet$ and
          a position $\posElm \in \PosSet$ to a measure $\msrElm \rst_{\posElm}
          \in \MsrSet$; this operator is \emph{canonical} whenever
          \begin{inparaenum}[(i)]
            \item\label{def:msrspc(trn:top)}
              $\msrElm = \top$ \iff $\msrElm \rst_{\posElm} = \top$ and
            \item\label{def:msrspc(trn:mon)}
              if $\msrElm \leq \msrElm[][\star]$ then $\msrElm \rst_{\posElm}
              \leq \msrElm[][\star] \rst_{\posElm}$, for all $\msrElm[][\star]
              \in \MsrSet$;
          \end{inparaenum}
        \item\label{def:msrspc(str)}
          the function $+ \colon \MsrSet \times \PosSet \to \MsrSet$, called
          \emph{stretch operator}, maps a measure $\msrElm \in \MsrSet$ and a
          position $\posElm \in \PosSet$ to a measure $\msrElm + \posElm \in
          \MsrSet$; this operator is \emph{canonical} whenever
          \begin{inparaenum}[(i)]
            \item\label{def:msrspc(str:top)}
              $\msrElm = \top$ \iff $\msrElm + \posElm = \top$ and
            \item\label{def:msrspc(str:mon)}
              if $\msrElm \leq \msrElm[][\star]$ then $\msrElm + \posElm \leq
              \msrElm[][\star] + \posElm$, for all $\msrElm[][\star] \in
              \MsrSet$.
          \end{inparaenum}
      \end{enumerate}
    \end{definition}

    Notice that both operators are canonical if they are \emph{monotone}
    in their first argument and preserve the distinction between the measure $\top$
    and the other measures, in the sense that $\top$ cannot be
    obtained by truncating or stretching a non-top measure and, \viceversa, no
    non-$\top$ measure is derivable by truncating or stretching $\top$.
    \footnote{Readers familiar with the research published in~\cite{Jur00,JL17}
    might find interesting to observe that both the small progress
    measure~\cite{Jur00} and the succinct progress measure~\cite{JL17}
    algorithms make implicit use of a measure space with a canonical truncation
    operator, but a non-canonical stretch operator.
    In more detail, the \emph{$d/2$-tuples} associated with the positions during
    an execution form a totally ordered set with minimum and maximum, once
    extended with the value $\top$ and where the value $\bot$ is identified with
    the unique all-zero $d/2$-tuple.
    Moreover, the truncation operator is represented by the function that zeros
    all components of a $d/2$-tuple with indexes smaller than the priority of
    the position given as second argument.
    Finally, the maps over $d/2$-tuples induced by the ternary functions
    $\mthfun{Prog}(\cdot, \cdot, \cdot)$~\cite{Jur00} and $\mthfun{lift}(\cdot,
    \cdot, \cdot)$~\cite{JL17}, used in the definition of the function
    $\mthfun{Lift}(\cdot, \cdot)$ at the core of the algorithms, implement the
    corresponding stretch operators.
    Such operators are, however, not canonical, since they can map some of the
    non-$\top$ measures to $\top$, failing so to satisfy the \emph{if} direction
    of Condition~\ref{def:msrspc(str:top)} of Definition~\ref{def:msrspc}.}
    %The lack of this feature is precisely the technical expedient which ensures
    %the termination of these approaches, by forcing the reachable part of the
    %measure space to be finite.
    %In this work, we show that termination can be achieved by other means.
    %without necessarily requiring the finiteness of the space.

    Given a measure space $\MsrName$, a \emph{measure function} $\mfElm$
    formalises the intuitive association discussed above by mapping each
    position $\posElm$ in the game to a measure $\mfElm(\posElm)$ in $\MsrName$.
    In addition, the order relation $<$ between measures declared in $\MsrName$
    induces a pointwise partial order $\sqsubseteq$ on the set measure functions
    $\MFSet$ defined in the usual way.
    This set together with its induced order form what we call a
    \emph{measure-function space}.

    \begin{definition}[Measure-Function Space]
      \label{def:msrfunspc}
      The \emph{measure-function space} induced by a given measure space
      $\MsrName$ is the \emph{partial order} $\MFName \defeq \tupleb {\MFSet}
      {\sqsubseteq}$, whose components are defined as prescribed in the
      following:
      \begin{enumerate}
        \item\label{def:msrfunspc(set)}
          $\MFSet \defeq \PosSet \to \MsrSet$ is the set of all functions
          $\mfElm \in \MFSet$, named \emph{measure functions}, mapping each
          position $\posElm \in \PosSet$ to a measure $\mfElm(\posElm) \in
          \MsrSet$;
        \item\label{def:msrfunspc(ord)}
          for all $\mfElm[1], \mfElm[2] \in \MFSet$, it holds that $\mfElm[1]
          \sqsubseteq \mfElm[2]$ \iff $\mfElm[1](\posElm) \leq
          \mfElm[2](\posElm)$, for all positions $\posElm \in \PosSet$.
      \end{enumerate}
    \end{definition}

    By taking $\mfElm[\bot]$ as the measure function associating measure $\bot$
    with every position, the following property of measure-function spaces
    immediately follows.

    \begin{proposition}
      \label{prp:unqmin}
      Every measure-function space $\MFName$ contains a unique minimal element
      $\mfElm[\bot] \in \MFSet$.
    \end{proposition}

    The \emph{$\PlrSym$-denotation} (\resp, \emph{$\OppSym$-denotation}) of a
    measure function $\mfElm \in \MFSet$ is the set $\denot{\mfElm}[\PlrSym]
    \defeq \set{ \posElm \in \PosSet }{ \mfElm(\posElm) \rst_{\posElm} = \top}$
    (\resp, $\denot{\mfElm}[\OppSym] \defeq \dual{\denot{\mfElm}[\PlrSym]}$) of
    all positions having maximal (\resp, non-maximal) measures associated with
    them in $\mfElm$, once truncated.
    Similarly, the \emph{$\bot$-denotation} (\resp, \emph{$+$-denotation}) of
    $\mfElm$ is the set $\denot{\mfElm}[\bot] \defeq \set{ \posElm \in \PosSet
    }{ \mfElm(\posElm) \rst_{\posElm} = \bot}$ (\resp, $\denot{\mfElm}[+] \defeq
    \dual{\denot{\mfElm}[\bot]}$) of all positions having minimal (\resp,
    non-minimal) measures.

    According to the intuition reported at the beginning of this section, the
    measure associated with a given position $\posElm$ is meant to encode
    information about the priorities encountered along the plays starting at
    that positions.
    More specifically, each measure contains the information gathered along some
    finite path and can be obtained by repeatedly applying the stretch operator
    backwards from the last position of that path.
    To formalise this intuition, we introduce the notion of measure
    $\msrFun(\pthElm)$ of a finite path $\pthElm$, including the empty one
    $\varepsilon$, that can be recursively computed via the function $\msrFun
    \colon \FPthSet \to \MFSet$ as follows:
    \[
      \msrFun(\pthElm) \defeq
      \begin{cases}
        \bot,
        & \text{if } \pthElm = \varepsilon; \\
        \msrFun(\pthElm') + \posElm,
        & \text{otherwise, where } \pthElm = \posElm \cdot \pthElm', \text{ for
        some unique } \posElm \in \PosSet \text { and } \pthElm' \in \FPthSet.
      \end{cases}
    \]
    At this point, we can constrain a measure function $\mfElm$, by requiring
    the measure $\mfElm(\posElm)$ of a position $\posElm$ to be witnessed by
    some finite path $\pi$ starting at $\posElm$, \ie, $\mfElm(\posElm) =
    \msrFun(\pthElm)$.
    By doing this, we obtain a \emph{ground measure function}.

    \begin{definition}[Ground Measure-Function Space]
      \label{def:grnmsrfunspc}
      The \emph{ground measure-function space} induced by a given measure space
      $\MsrName$ is the \emph{subspace} $\tupleb {\GMFSet} {\sqsubseteq}$ of
      the measure-function space $\MFName$, where $\GMFSet \defeq \set{ \mfElm 
      \in \MFSet }{ \forall \posElm \in \PosSet \ldotp \mfElm(\posElm) \in 
      \GMsrSet(\posElm) }$
      %$\GMFSet \defeq \PosSet \cto[\posElm] \GMsrSet(\posElm)$
      with $\GMsrSet(\posElm) \defeq \set{
      \msrFun(\pthElm) }{ \pthElm \in \FPthSet(\posElm) }$.
    \end{definition}

    All progress-measure approaches proposed in the literature implicitly work
    by updating ground measure functions only.
    The same holds true for the algorithm proposed in the current work, which
    actually runs on the even more restricted set of \emph{simple measure
    functions} introduced in Section~\ref{sec:qsidommsr}. %the second half of
%     Notice however that the property of being ground for a measure function is
%     not strictly required to prove soundness and completeness of the approaches,
%     as shown in the following.

  \end{subsection}

  \begin{subsection}{Progress-Measure Functions}

    The following definition, which tightly connects the truncation and stretch
    operators, formalises the essential semantic features of a measure space
    that are required to provide a meaningful notion of progress measure, as
    proven in Theorem~\ref{thm:prgmsr}.

    \begin{definition}[Progress Measure Space]
      \label{def:prgmsrspc}
      A measure space $\MsrName$ is a \emph{progress measure space} if the following
      properties hold true, for each measure $\msrElm \in \MsrSet$ and position
      $\posElm \in \PosSet$:
      \begin{enumerate}
        \item\label{def:prgmsrspc(grt)}
          $\msrElm \rst_{\posElm} < (\msrElm + \posElm) \rst_{\posElm}$, if
          $\msrElm \rst_{\posElm} < \top$ and $\prtFun(\posElm)$ is even;
        \item\label{def:prgmsrspc(geq)}
          $\msrElm \rst_{\uposElm} \leq (\msrElm + \posElm) \rst_{\uposElm}$,
          for all even-priority positions $\uposElm$ with $\prtFun(\posElm) \leq
          \prtFun(\uposElm)$.
      \end{enumerate}
    \end{definition}

    The first condition requires that the contribution to a measure due to an
    even position $\posElm$ cannot be cancelled out by truncating at $\posElm$
    itself and that such contribution is always meaningful, namely strictly
    increasing.
    This matches the intuition that even priorities tend to move away from the
    parity condition, therefore increasing the associated measure.
    The second property, instead, ensures that no lower-priority position
    $\posElm$ can overcome a higher even priority position $\uposElm$, in the
    sense that the contribution of $\posElm$ to the measure cannot move closer
    to the parity condition, once the stretched measure is analysed \wrt
    $\uposElm$.
    Technically, this means that the stretch forced by a lower priority position
    is always viewed as a non-strict improvement by any even position with
    higher priority, regardless of the parity of the first one.

    We can now turn our attention to the notion of \emph{progress measure}.
    Intuitively, a measure function over a progress mesure space is a progress 
    measure if it guarantees that every position with a non-$\top$ measure can 
    progress toward the parity property, namely toward lower measures.
    In other words, this establishes a type of stability property on the
    positions of a game according to the following intuition, which takes into
    account the opposite attitude of the two players.
    While player $\PlrSym$ tries to push toward higher measures, the opponent
    will try to keep the measure as low as possible.
    If player $\PlrSym$ cannot increase the measures of its positions and the
    opponent is not forced to increase the measures of its own positions, then
    player $\PlrSym$ cannot prevent the player $\OppSym$ from winning the game
    from all those positions whose measure did not reach value $\top$.
    This corresponds to requiring that player $\PlrSym$ cannot increase the
    measure of its positions by stretching the measure of any adjacent move,
    while the opponent can always choose a move whose corresponding stretch
    does not force the increment of the measure.
    The following definition precisely formalises this intuitive explanation.

    \begin{definition}[Progress Measure]
      \label{def:prgmsr}
      A measure function $\mfElm \in \MFSet$ is a \emph{progress measure} if the
      following conditions hold, for all positions $\posElm \in \PosSet$:
      \begin{enumerate}
        \item\label{def:prgmsr(plr)}
          $\mfElm(\wposElm) + \posElm \leq \mfElm(\posElm)$, for all adjacents
          $\wposElm \in \MovRel(\posElm)$ of $\posElm$, if $\posElm \in
          \PosSet[\PlrSym]$;
        \item\label{def:prgmsr(opp)}
          $\mfElm(\wposElm) + \posElm \leq \mfElm(\posElm)$, for some adjacent
          $\wposElm \in \MovRel(\posElm)$ of $\posElm$, if $\posElm \in
          \PosSet[\OppSym]$.
      \end{enumerate}
      A $\OppSym$-strategy $\strElm \in \StrSet[\OppSym]$ is
      \emph{$\mfElm$-coherent} if $\mfElm(\strElm(\posElm)) + \posElm \leq
      \mfElm(\posElm)$, for all $\OppSym$-positions $\posElm \in
      \PosSet[\OppSym]$.
    \end{definition}

    Assuming a progress measure space with a canonical truncation operator, the
    notion of progress measure actually ensures that any play satisfying this
    progress condition will eventually be trapped in a cycle in which the
    maximal priority is odd, thereby witnessing a win for player $\OppSym$.
    This is established by the following theorem.

    \begin{restatable}{theorem}{prgmsrr}(Progress Measure)
      \label{thm:prgmsr}
      Let $\mfElm \in \MFSet$ be a progress measure \wrt a progress measure
      space $\MsrName$ with a canonical truncation operator.
      Then, $\denot{\mfElm}[\OppSym]$ is a $\OppSym$-dominion for which all
      $\mfElm$-coherent $\OppSym$-strategies are $\OppSym$-winning.
    \end{restatable}

  \end{subsection}

\end{section}

% End of file SectionI.tex

%%****************************************************************************%%
%%                                                                            %%
%% Article Title                                                              %%
%%                                                                            %%
%% SectionII.tex                                                              %%
%%                                                                            %%
%% Revision 0                                                                 %%
%%                                                                            %%
%% Copyright (C) 2020, Massimo Benerecetti, Daniele Dell'Erba, Marco Faella,  %%
%%                     and Fabio Mogavero.                                    %%
%% All rights reserved.                                                       %%
%%                                                                            %%
%%****************************************************************************%%

% Begin of file SectionII.tex

\begin{section}{Solving Parity Games via Quasi-Dominion Measures}
  \label{sec:qsidommsr}

  The framework set forth in the previous section is already sufficient to
  define a sound and complete approach for the solution of parity games, as
  shown in~\cite{Jur00}.
  Here, however, we shall further refine the measure space in order to
  accommodate the notion of quasi dominion into the measure functions.

  \begin{subsection}{Quasi-Dominion-Measure Functions}

    The notion of \emph{quasi dominion} was originally introduced for parity
    games in~\cite{BDM16,BDM18} and, in a slightly different form,
    in~\cite{Fea10a}.
    Here we provide a somewhat more general version that can be easily
    integrated with measure functions.

    \begin{definition}[Quasi Dominion]
      \label{def:qsidom}
      A set of positions $\QSet \subseteq \PosSet$ is a \emph{weak quasi
      $\PlrSym$-dominion} if there exists a $\PlrSym$-strategy $\strElm[\PlrSym]
      \in \StrSet[\PlrSym](\QSet)$, called \emph{$\PlrSym$-witness for $\QSet$},
      such that, for all $\OppSym$-strategies $\strElm[\OppSym] \in
      \StrSet[\OppSym](\QSet)$ and positions $\posElm \in \QSet$, if the induced
      play $\pthElm = \playFun((\strElm[\PlrSym], \strElm[\OppSym]), \posElm)$
      is infinite then $\prtFun(\pthElm)$ is even.
      If the even-parity condition holds also for finite plays, then $\QSet$ is
      called \emph{quasi $\PlrSym$-dominion}.
    \end{definition}

    Essentially, a quasi $\PlrSym$-dominion consists in a set $\QSet$ of
    positions, starting from which player $\PlrSym$ can force plays whose
    maximal prefixes contained in $\QSet$ have even maximal priority.
    Observe
    that, in case the maximal prefix contained in $\QSet$ is infinite, then the
    play is actually winning for player $\PlrSym$.  When the condition holds only for
    infinite plays, the set is called weak quasi $\PlrSym$-dominion.
    Clearly, any quasi $\PlrSym$-dominion is also a weak quasi
    $\PlrSym$-dominion.
    Moreover, the latter are closed under subsets, while the former are not.
    It is an immediate consequence of the definition above that all infinite
    plays induced by the $\PlrSym$-witness, if any, are winning for player $\PlrSym$.
    This also entails that any subset $\QSet[][\star]$ of a weak quasi
    $\PlrSym$-dominion $\QSet$, in which the player $\PlrSym$ can remains and
    from where the opponent cannot escape, is a $\PlrSym$-dominion.
    Indeed, in such a set of positions, player $\PlrSym$ always has moves that
    remain in $\QSet[][\star]$, while the opponent can only choose moves
    remaining in $\QSet[][\star]$.
    Hence, any play compatible with the $\PlrSym$-witness for $\QSet$ that
    starts in $\QSet[][\star]$ is infinite and entirely contained
    in $\QSet[][\star]$.
    We have, so, the following result.

    \begin{corollary}[Quasi Dominion]
      \label{cor:qsidom}
      Let $\QSet \subseteq \PosSet$ be a \emph{weak quasi $\PlrSym$-dominion},
      $\strElm[\PlrSym] \in \StrSet[\PlrSym](\QSet)$ one of its
      $\PlrSym$-witnesses, and $\QSet[][\star] \subseteq \QSet$ a subset such
      that for all positions $\posElm \in \QSet[][\star] \cap \PosSet[\PlrSym]$
      it holds $\strElm[\PlrSym](\posElm) \in \QSet[][\star]$ and
      for all positions $\posElm \in \QSet[][\star] \cap \PosSet[\OppSym]$
      it holds $\MovRel(\posElm) \subseteq \QSet[][\star]$.
      Then, $\QSet[][\star]$ is a $\PlrSym$-dominion.
    \end{corollary}

    The notion of progress measure introduced in the previous section basically
    gives us positions that do satisfy the parity condition and, thus, are
    winning for player $\OppSym$, namely the non-$\top$ positions.
    This is done by enforcing on the measure space the conditions of
    Definition~\ref{def:prgmsrspc} that formally captures the idea that even
    priority positions push further away form the parity property.
    However, the progress measure is an asymmetric notion, centred around one
    of the two players, specifically player $\OppSym$, and does not provide any
    meaningful information on the other player.
    More specifically, no estimation on how far player $\PlrSym$ is from winning the
    game starting in a given position, \ie, from satisfying the dual/even parity
    property, can be extracted from it.

    Quasi dominions, however, are precisely intended to encode the dual
    information for player $\PlrSym$.
    In this case, the odd priority positions are those that push further away
    from satisfying the dual/even parity property.
    A natural way to embed information about quasi dominions into the measures
    is, thus, to enforce the dual conditions of Definition~\ref{def:prgmsrspc},
    which leads us to the notion of \emph{regress measure}.
    Here we constrain the behaviour of the truncation and stretch operators
    \wrt the odd positions, instead of the even ones.

    \begin{definition}[Regress Measure Space]
      \label{def:regmsrspc}
      A measure space $\MsrName$ is a \emph{regress space} if the following
      properties hold true, for each measure $\msrElm \in \MsrSet$ and position
      $\posElm \in \PosSet$:
      \begin{enumerate}
        \item\label{def:regmsrspc(les)}
          $(\msrElm + \posElm) \rst_{\posElm} < \msrElm \rst_{\posElm}$, if
          $\bot < \msrElm \rst_{\posElm} < \top$ and $\prtFun(\posElm)$ is odd;
        \item\label{def:regmsrspc(leq)}
          $(\msrElm + \posElm) \rst_{\uposElm} \leq \msrElm \rst_{\uposElm}$,
          for all odd-priority positions $\uposElm$ with $\prtFun(\posElm) \leq
          \prtFun(\uposElm)$.
      \end{enumerate}
    \end{definition}

    We can now define the notion of \emph{regress measure} as the dual of the
    progress measure.

    \begin{definition}[Regress Measure]
      \label{def:regmsr}
      A measure function $\mfElm \in \MFSet$ is a \emph{regress measure} if the
      following conditions hold, for all positions $\posElm \in
      \denot{\mfElm}[+] \setminus \denot{\mfElm}[\PlrSym]$:
      \begin{enumerate}
        \item\label{def:regmsr(plr)}
          $\mfElm(\posElm) \leq \mfElm(\wposElm) + \posElm$, for some adjacent
          $\wposElm \in \MovRel(\posElm)$ of $\posElm$, if $\posElm \in
          \PosSet[\PlrSym]$;
        \item\label{def:regmsr(opp)}
          $\mfElm(\posElm) \leq \mfElm(\wposElm) + \posElm$, for all adjacents
          $\wposElm \in \MovRel(\posElm)$ of $\posElm$, if $\posElm \in
          \PosSet[\OppSym]$.
      \end{enumerate}
      A $\PlrSym$-strategy $\strElm \in \StrSet[\PlrSym]$ is
      \emph{$\mfElm$-coherent} if $\mfElm(\posElm) \leq \mfElm(\strElm(\posElm))
      + \posElm$, for all $\PlrSym$-positions $\posElm \in \PosSet[\PlrSym]$.
    \end{definition}

    Regress measures are meant to ensure that all the positions whose truncation is neither
    $\bot$ nor $\top$ form a weak quasi dominion for player $\PlrSym$,
    as established by the following theorem.

    \begin{restatable}{theorem}{regmsrr}(Regress Measure)
      \label{thm:regmsr}
      Let $\mfElm \in \MFSet$ be a regress measure \wrt a regress measure space
      $\MsrName$ with a canonical truncation operator.
      Then, $\denot{\mfElm}[+] \setminus \denot{\mfElm}[\PlrSym]$ is a weak
      quasi $\PlrSym$-dominion for which all $\mfElm$-coherent
      $\PlrSym$-strategies are $\PlrSym$-witnesses, once restricted to
      $\denot{\mfElm}[+] \setminus \denot{\mfElm}[\PlrSym]$.
    \end{restatable}
    %
    %proof moved to AppendixB

    A \emph{quasi-dominion measure} is, then, defined as a regress measure with
    the additional property that all the positions with measure $\top$ form a
    $\PlrSym$-dominion, \ie, are indeed winning for player $\PlrSym$.

    \begin{definition}[Quasi-Dominion Measure]
      \label{def:qsidommsr}
      A measure function $\mfElm \!\in\! \MFSet$ is a \emph{quasi-dominion
      measure} (\qdmf, for short) if it is a regress measure for which
      $\denot{\mfElm}[\PlrSym]$ is a $\PlrSym$-dominion.
      $\QDMFSet$ denotes the set of all {\qdmf}s.
    \end{definition}

    The following theorem establishes the main property of quasi-dominion
    measures, namely that the set of non-$\bot$ positions always forms a weak 
    quasi $\PlrSym$-dominion.

    \begin{restatable}{theorem}{qsidommsrir}(Quasi-Dominion Measure I)
      \label{thm:qsidommsri}
      Let $\mfElm \in \MFSet$ be a quasi-dominion measure \wrt a regress measure
      space $\MsrName$ with a canonical truncation operator.
      Then, $\denot{\mfElm}[+]$ is a weak quasi $\PlrSym$-dominion for which all
      $\mfElm$-coherent $\PlrSym$-strategies that are winning on
      $\denot{\mfElm}[\PlrSym]$ are $\PlrSym$-witnesses, once restricted to
      $\denot{\mfElm}[+]$.
    \end{restatable}

  \end{subsection}

  \begin{subsection}{Simple-Measure Functions}

    While the solution of a parity game involves checking the parity condition
    along infinite plays, as the solution algorithm proceeds, the measure
    of a position $\posElm$ encodes a finite horizon approximation
    of that condition for $\posElm$.
    This approximation is progressively refined during the computation, by
    exploring longer and longer finite prefixes of the possible plays starting
    from $\posElm$.
    Clearly, any play that contains a cycle is either winning for the player
    $\PlrSym$ or for the opponent.
    In a sense, the shortest prefix of the play that ends with a repetition of
    some position already provides all the necessary information to assess
    the winning player of the entire infinite play.
    This observation suggests that measures need only encode information of
    finite simple paths in the game, since those are the only prefixes that need
    to be extended to obtain finer approximations.

    From now on, we shall fix a measure space $\MsrName$ and, thus, the induced
    measure-function space $\MFName$.
    Given a set of positions $\XSet \subseteq \PosSet$ and a position $\posElm
    \in \XSet$, we introduce the set
    \[
      \SMsrSet(\posElm, \XSet) \defeq \set{ \msrFun(\pthElm) \in \MsrSet }{
      \pthElm \in \SPthSet(\posElm, \XSet) } \cup \{ \top \}
    \]
    of the \emph{simple measures} of $\posElm$ \wrt $\XSet$.
    This set contains, besides the measure $\top$, only the measures of finite
    simple paths originating at $\posElm$ and composed only of positions in
    $\XSet$.
    It is immediate to observe the following property.

    \begin{proposition}
      \label{prp:fintotord}
      $\tupled {\SMsrSet(\posElm, \XSet)} {<} {\bot} {\top}$ is a finite strict
      total order with minimum and maximum, for all positions $\posElm \in
      \PosSet$ and sets of positions $\XSet \subseteq \PosSet$.
    \end{proposition}

   Following the observations above, \emph{simple measure functions} restrict
   the possible measures of each position to those induced by finite simple
   paths contained in the quasi $\PlrSym$-dominion $\denot{\mfElm}[+]$.
%  We can then restrict the possible measures of each position to those induced
%  by finite simple paths via the notion of \emph{simple measure}, in a way
%  quite similar to that used for Definition~\ref{def:grnmsrfunspc}.

    \begin{definition}[Simple Measure Function]
      \label{def:simmsrfun}
      A measure function $\mfElm \in \MFSet$ is a \emph{simple measure} (\smf,
      for short) if $\mfElm(\posElm) \in \SMsrSet(\posElm, \denot{\mfElm}[+])$,
      for all positions $\posElm \in \PosSet$.
      $\SMFSet$ denotes the set of all {\smf}s.
    \end{definition}

    The next proposition shows that, for simple measures $\mfElm$, if the
    truncation of the measure $\mfElm(\posElm)$ \wrt the position $\posElm$
    itself is $\bot$, then the original measure $\mfElm(\posElm)$ is $\bot$ as
    well, since $\SPthSet(\posElm, \denot{\mfElm}[+]) = \{ \varepsilon \}$ and
    $\msrFun(\varepsilon) = \bot$.

    \begin{proposition}
      \label{prp:botden}
      For every simple measure $\mfElm \in \SMFSet$ and position $\posElm \in
      \denot{\mfElm}[\bot]$, it holds that $\mfElm(\posElm) = \bot$.
    \end{proposition}

    The measure function $\mfElm[\bot]$ is clearly the minimal element \wrt
    $\sqsubseteq$ in the set $\SMFSet$ of simple measures.

    \begin{proposition}
      \label{prp:finparord}
      $\tupleb {\SMFSet} {\sqsubseteq}$ is a finite partial order with
      $\mfElm[\bot] \in \SMFSet$ as unique minimal element.
    \end{proposition}

    Putting together Definitions~\ref{def:qsidommsr} and~\ref{def:simmsrfun},
    we obtain \emph{simple quasi-dominion measures}, which enjoy a stronger
    property than the one stated in Theorem~\ref{thm:qsidommsri}.

    \begin{restatable}{theorem}{qsidommsriir}(Quasi-Dominion Measure II)
      \label{thm:qsidommsrii}
      Let $\mfElm \in \MFSet$ be a simple quasi-dominion measure \wrt a regress
      measure space $\MsrName$ with a canonical truncation operator satisfying
      the equality $\bot \rst_{\posElm} = \bot$, for all odd priority positions
      $\posElm$.
      Then, $\denot{\mfElm}[+]$ is a quasi $\PlrSym$-dominion for which all
      $\mfElm$-coherent $\PlrSym$-strategies that are winning on
      $\denot{\mfElm}[\PlrSym]$ are $\PlrSym$-witnesses, once restricted to
      $\denot{\mfElm}[+]$.
    \end{restatable}

    Since $\mfElm[\bot]$ is the measure function induced by the empty path on
    each position, the following property is immediate.

    \begin{proposition}
      \label{prp:minqdm}
      The minimal measure function $\mfElm[\bot] \in \MFSet$ is a simple
      quasi-dominion measure.
    \end{proposition}

  \end{subsection}

\end{section}

% End of file SectionII.tex

%%****************************************************************************%%
%%                                                                            %%
%% Article Title                                                              %%
%%                                                                            %%
%% SectionIII.tex                                                             %%
%%                                                                            %%
%% Revision 0                                                                 %%
%%                                                                            %%
%% Copyright (C) 2019, Massimo Benerecetti, Daniele Dell'Erba, Marco Faella,  %%
%%                     and Fabio Mogavero.                                    %%
%% All rights reserved.                                                       %%
%%                                                                            %%
%%****************************************************************************%%

% Begin of file SectionIII.tex

\begin{section}{A Concrete Algorithm}
  \label{sec:effalg}

  This section describes an algorithm that solves any parity game by maintaining
  and updating a simple quasi-dominion measure function, until it reaches a
  fixpoint that is both a progress and a quasi-dominion measure.
  At that point, the results in the previous sections ensure that the winning
  positions for both players are determined and easily recovered from the final
  measure by computing the $\PlrSym$- and $\OppSym$-denotations.

  \begin{subsection}{A Concrete Measure Space}

    The measures used by the concrete algorithm associate a non-negative integer
    with each priority in the game, in other words, they are sequences of
    naturals, one for each priority.
    For technical reasons, however, we introduce a more general class of
    concrete measures $\EvlSet \defeq \PrtSet \to \SetZ$, whose range also
    includes the negative integers.
    This allows us to provide algebraic operations on measures, such as addition
    and subtraction, which will prove instrumental in the implementation of the
    basic operators used by the algorithm.
    More specifically:
    \begin{inparaenum}[(i)]
      \item
        the null element $\mathbf{0} \in \EvlSet$ is the distinguished measure
        such that $(\mathbf{0})(k) \defeq 0$, for all indexes $k \in \PrtSet$;
      \item
        the opposite of a measure and the sum of two measures
        are defined point-wise: $(-\evlElm)(k) \defeq -\evlElm(k)$
        and $(\evlElm[1] + \evlElm[2])(k) \defeq \evlElm[1](k) +
        \evlElm[2](k)$, for all $k \in \PrtSet$.
    \end{inparaenum}
    Recall that, according to Definition~\ref{def:msrspc}, measures need to be
    totally ordered.
    For our concrete measures, we employ an \emph{alternate lexicographic
    order}, that is, if we interpret a measure as a sequence of integers, with
    decreasing indexes from left to right, values that are later in the sequence
    are less important than those that come earlier: as in the standard numeric
    representation, the left-most integer is the most-significant digit, while
    the right-most integer is the least-significant one.
    Moreover, values with even indexes are ordered in the natural way, namely by
    increasing magnitude, whereas those with odd indexes are ordered in the
    opposite fashion, \ie, by decreasing magnitude.
    Formally, $<\: \subseteq \EvlSet \times \EvlSet$ is the strict total order
    defined as follows: $\evlElm[1] < \evlElm[2]$ if there exists an index $k
    \in \PrtSet$ such that
    \begin{inparaenum}[(i)]
      \item
        $k$ is the greatest index for which $\evlElm[1](k) \neq
        \evlElm[2](k)$ and
      \item
        $\evlElm[1](k) < \evlElm[2](k)$, if $k$ is even, and
        $\evlElm[2](k) < \evlElm[1](k)$, otherwise.
    \end{inparaenum}
    A special family of measures is given by the Kronecker delta $\delta
    \colon \PrtSet \to (\PrtSet \to \{ 0, 1 \})$, where $\delta_{i}(i) \defeq
    1$, and $\delta_{i}(j) \defeq 0$, for all $j \neq i$.
    Obviously, $\delta_{i} \in \EvlSet$, for every index $i \in \PrtSet$.
    It is quite immediate to show the following property.

    \begin{proposition}
      \label{prp:naievlstr}
      The structure $\EvlName \defeq \tuplee {\EvlSet} {<} {\mathbf{0}} {-} {+}$
      is a totally-ordered Abelian group.
    \end{proposition}

    We can now define the measure space used by the concrete algorithm.
    It suffices to restrict the measures to only assign non-negative values to
    the existing priorities in the game, and then define the appropriate
    canonical truncation and stretch operators.
    Let $\EvlSet[][+]$ be the subset containing all the concrete measures 
    $\evlElm \in \EvlSet$ such that:
    \begin{inparaenum}[(i)]
      \item
        the highest priority $k$ for which $\evlElm(k) \geq 0$ is even;
      \item
        $\evlElm(k) \geq 0$, for all $k \in \PrtSet$.
    \end{inparaenum}

    \begin{definition}[Concrete Measure Structure]
      \label{def:naimsrstr}
      The \emph{concrete measure structure} is the tuple $\MsrName \defeq
      \tuplef {\MsrSet} {<} {\bot} {\top} {\rst} {+}$, whose components are
      defined as follows:
      \begin{enumerate}
        \item\label{def:naimsrstr(set)}
          $\MsrSet \defeq \EvlSet[][+] \cup \{ \top \}$, where $\top$ is a
          distinguished fresh element and $\bot \defeq \mathbf{0}$;
        \item\label{def:naimsrstr(ord)}
          $<\: \subseteq \MsrSet \times \MsrSet$ is the order on $\EvlSet[][+]$
          extended with $\msrElm < \top$, for every measure $\msrElm \in
          \EvlSet[][+]$;
        \item\label{def:naimsrstr(rst)}
          $\rst \colon \MsrSet \times \PosSet \to \MsrSet$ is the operator such
          that, for all positions $\posElm \in \PosSet$, the following holds:
          \begin{inparaenum}[(i)]
            \item\label{def:naimsrstr(rst:top)}
              $\top \rst_{\posElm} \defeq \top$;
            \item\label{def:naimsrstr(rst:prt)}
              $(\msrElm \rst_{\posElm})(\prtElm) \defeq \msrElm(\prtElm)$, if
              $\prtElm \geq \prtFun(\posElm)$, and $(\msrElm
              \rst_{\posElm})(\prtElm) \defeq 0$, otherwise, for all $\msrElm
              \in \EvlSet[][+]$ and $\prtElm \in \PrtSet$;
          \end{inparaenum}
        \item\label{def:naimsrstr(str)}
          $+ \colon \MsrSet \times \PosSet \to \MsrSet$ is the operator such
          that $\msrElm + \posElm \defeq \max \{ \bot, \msrElm +
          \delta_{\prtFun(\posElm)} \}$, for all positions $\posElm \in
          \PosSet$.
      \end{enumerate}
    \end{definition}

    Truncating a measure \wrt a position $\posElm$ consists in setting to zero
    the value associated with all priorities smaller than the priority of
    $\posElm$.
    Stretching a measure \wrt a position $\posElm$, instead, means incrementing
    the value associated with the priority of $\posElm$, unless the result is
    lower than $\bot$, in which case the stretch is $\bot$.
    For example, consider a game with priorities from $0$ to $4$ and the three
    measures $\msrElm[1] \defeq (0, 0, 1, 0, 1)$, $\msrElm[2] \defeq (0, 0, 1,
    1, 0)$, and $\msrElm[3] \defeq (0, 0, 1, 1, 1)$.
    Then, we have $\msrElm[2] < \msrElm[3] < \msrElm[1]$.
    Indeed, $0$ is the greatest priority in which $\msrElm[2]$ and $\msrElm[3]$
    differ, it is even, and $\msrElm[2](0) = 0 < 1 = \msrElm[3](0)$.
    Moreover, $1$ is the greatest priority in which both $\msrElm[2]$ and
    $\msrElm[3]$ differ from $\msrElm[1]$, it is odd, and $\msrElm[2](1) =
    \msrElm[3](1) = 1 > 0 = \msrElm[1](1)$.
    By truncating the three measures at a position $\posElm$ with priority $1$,
    we obtain $\msrElm[2] \rst_{\posElm} = \msrElm[3] \rst_{\posElm} = (0, 0, 1,
    1, 0) < \msrElm[1] \rst_{\posElm} = (0, 0, 1, 0, 0)$.
    We also have $\msrElm[1] + \uposElm = \msrElm[2] + \uposElm = \msrElm[3] +
    \uposElm = \bot$, if $\uposElm$ is a position with priority $3$ and
    $(\msrElm[2] + \wposElm) + \uposElm = (1, 1, 1, 1, 0) < (1, 0, 1, 2, 0) =
    (\msrElm[2] + \wposElm) + \posElm$, if $\wposElm$ is a position with
    priority $4$.

    The following straightforward result states that the concrete measure
    structure satisfies all the desired properties introduced in the previous
    sections, plus some additional ones.
    Specifically, the truncation operator preserves the $\bot$ measure, the
    truncation of a non-$\bot$ stretch cannot lead to the $\bot$ measure, and
    the stretch operator commutes with the addition on measures.

    \begin{proposition}
      \label{prp:naimsrstr}
      The concrete measure structure is a progress and regress measure space,
      whose restriction and stretch operators are canonical.
      Moreover, the following properties hold true, for all positions $\posElm
      \in \PosSet$, measures $\msrElm \in \MsrSet$, and evaluations $\evlElm \in
      \EvlSet$:
      \begin{inparaenum}[(i)]
        \item\label{prp:naimsrstr(trn)}
          %if $\posElm \in \OddSet$ then
          $\bot \rst_{\posElm} = \bot$;
        \item\label{prp:naimsrstr(str)}
          %if $\posElm \in \OddSet$ and $\msrElm + \posElm \neq \bot$
          if $\msrElm + \posElm \neq \bot$ then $(\msrElm + \posElm)
          \rst_{\posElm} \neq \bot$;
        \item\label{prp:naimsrstr(com)}
          if $\msrElm + \posElm \neq \bot$ then $(\msrElm + \posElm) + \evlElm =
          (\msrElm + \evlElm) + \posElm$.
      \end{inparaenum}
    \end{proposition}

  \end{subsection}

  \begin{subsection}{The Solution Algorithm}

    The algorithm we now propose, which makes an implicit use of the concrete
    measure structure just introduced, is based on repeatedly applying two
    \emph{progress operators}, $\prgFun[+]$ and $\prgFun[\bot]$, to an initial
    quasi-dominion measure function $\mfElm$, until a fixpoint $\mfElm[][\star]$
    is reached.
    At that point, the positions whose measure is $\top$ in $\mfElm[][\star]$
    are winning for player~$\PlrSym$, while all other positions are winning for
    player~$\OppSym$.

    Recall that a function $f$ on an ordered set is \emph{inflationary} if for
    all elements $x$ of its domain it holds $x \leq f(x)$.
    We show in the following that both progress operators are inflationary,
    so that the fixpoint $\mfElm[][\star]$ is the limit of the ascending sequence
    of measures obtained by repeated application of those progress operators
    (\aka the \emph{inflationary fixpoint}).
    We denote by $\solFun(\mu)$ such limit, as a partial mapping from $\MFSet$ to $\MFSet$, when starting from measure $\mu$:
    \[
       \solFun \defeq \ifpFun\, \mfElm \,.\, \prgFun[+](\prgFun[\bot](\mfElm))
       \colon \MFSet \pto \MFSet.
    \]
    Intuitively, all positions with a non-$\bot$ measure in $\mfElm$, \ie,
    $\QSet \defeq \denot{\mfElm}[+]$, form a quasi $\PlrSym$-dominion and the
    $\prgFun[+]$ (\resp, $\prgFun[\bot]$) operator is responsible of enforcing
    the progress condition on the positions inside (\resp, outside) $\QSet$ that
    do not satisfy the proper inequalities between the measures along the
    moves.
    This is done in such a way to preserve the properties of quasi-dominion
    measure function and represents the main point where the classic
    progress-measure approaches and the proposed technique diverge.

    Both operators $\prgFun[+]$ and $\prgFun[\bot]$ internally employ a
    \emph{lift operator} $\liftFun \colon \MFSet \times \pow{\PosSet} \times
    \pow{\PosSet} \to \MFSet$, that adjusts a measure function so that it
    locally satisfies the conditions of both the progress and regress measures
    (Definitions~\ref{def:prgmsr} and~\ref{def:regmsr}).
    The two operators need to selectively adjust the measure of specific sets of
    positions.
    For this reason, besides the current measure function, the lift operator
    carries two additional arguments:
    \begin{inparaenum}[(i)]
      \item
        the set of positions $\SSet$ whose measure we want to update, and
      \item
        the set of successor positions $\TSet$ that the update must be based on.
    \end{inparaenum}
    Formally, we obtain the following definition:
    \vspace{-0.5em}
    \[
      \liftFun(\mfElm, \SSet, \TSet)(\posElm) \defeq
      \begin{cases}
        \max {\set{ \mfElm(\wposElm) + \posElm }{ \wposElm \in \MovRel(\posElm)
        \cap \TSet }},
        & \text{if } \posElm \in \SSet \cap \PosSet[\PlrSym]; \\
        \min {\set{ \mfElm(\wposElm) + \posElm }{ \wposElm \in \MovRel(\posElm)
        \cap \TSet }},
        & \text{if } \posElm \in \SSet \cap \PosSet[\OppSym]; \\
        \mfElm(\posElm),
        & \text{otherwise}.
      \end{cases}
      \vspace{-0.5em}
    \]

    The $\prgFun[\bot]$ operator is tasked with adjusting the measure of the
    positions that currently have the minimal measure $\bot$.
    After the update, some of them will acquire a positive measure, thus
    entering into the current quasi $\PlrSym$-dominion.
    From an operational viewpoint, $\prgFun[\bot]$ consists of a single call to
    $\liftFun$:
    \[
      \prgFun[\bot](\mfElm) \defeq \liftFun(\mfElm, \denot{\mfElm}[\bot],
      \PosSet) : \MFSet \to \MFSet.
    \]

    \begin{wrapfigure}[12]{r}[0em]{0.37\textwidth}
      \vspace{-1.25em}
      \algprg
    \end{wrapfigure}
    Applying the definitions, one can easily see that $\prgFun[\bot]$ raises the
    measure of those positions such that:
    \begin{inparaenum}[(i)]
      \item
        they have the minimal measure $\bot$,
      \item
        they either belong to player $\PlrSym$ and have an adjacent with
        positive measure or belong to player $\OppSym$ and have all adjacents
        with positive measures, and
      \item
        the stretch of the adjacent measures is greater than $\bot$.
    \end{inparaenum}
    The following lemma states the main properties of interest for the
    $\prgFun[\bot]$ operator, where we assume $\MFSet[\bot] \defeq \set{ \mfElm
    \in \MFSet }{ \forall \posElm \in \denot{\mfElm}[\bot] \,.\, \mfElm(\posElm)
    = \bot }$ and $\QDMFSet[\bot] \defeq \QDMFSet \cap \MFSet[\bot]$.
    Observe that, by Proposition~\ref{prp:botden}, it holds that $\SMFSet
    \subseteq \MFSet[\bot]$.

    \begin{restatable}{lemma}{prgbotr}
      \label{lmm:prgbot}
      The progress operator $\prgFun[\bot]$ enjoys the following properties:
      \begin{inparaenum}[(i)]
        \item\label{lmm:prgbot(mfb)}
          it is an inflationary function from $\MFSet[\bot]$ to $\MFSet[\bot]$;
        \item\label{lmm:prgbot(qdm)}
          it maps $\QDMFSet[\bot]$ into $\QDMFSet[\bot]$;
        \item\label{lmm:prgbot(smf)}
          it maps $\SMFSet$ into $\SMFSet$;
        \item\label{lmm:prgbot(fix)}
          every fixpoint $\mfElm \in \MFSet$ of $\prgFun[\bot]$ is a progress
          measure over $\denot{\mfElm}[\bot]$.
      \end{inparaenum}
    \end{restatable}

    We now turn our attention to the progress operator $\prgFun[+]$, whose
    pseudo-code is reported in Algorithm~\ref{alg:prg}.
    Besides the lift function, this operator employs two other functions,
    $\escFun \colon \MFSet \times \pow{\PosSet} \to \pow{\PosSet}$ and $\befFun
    \colon \MFSet \times \pow{\PosSet} \times \PosSet \to \EvlSet$, called,
    respectively, \emph{escape function} and \emph{best-escape position
    function}.
    Given a set of positions $\QSet$, the escape function collects the subset of
    positions in $\QSet$ from which their owner wants or is forced to exit from
    $\QSet$, according to their objective.
    Specifically, those are
    \begin{inparaenum}[(i)]
      \item
        the $\OppSym$-positions having a successor outside $\QSet$ and
      \item
        the $\PlrSym$-positions $\posElm$ such that none of their successors
        $\wposElm$ belonging to $\QSet$ support the measure of $\posElm$ in the
        current measure function.
    \end{inparaenum}
    Formally:

    \vspace{-1em}
    \[
      \escFun(\mfElm, \QSet) \,\defeq\,
      {\set{ \posElm \in \QSet \cap \PosSet[\OppSym] }{ \MovRel(\posElm)
      \setminus \QSet \neq \emptyset }}
      \,\cup\, {\set{ \posElm \in \QSet \cap \PosSet[\PlrSym] }{ \forall
      \wposElm \in \MovRel(\posElm) \cap \QSet \,.\, \mfElm(\wposElm) + \posElm
      < \mfElm(\posElm) }}.
    \]
%   \begin{linenomath}
%   \begin{align*}
%     \escFun(\mfElm, \QSet)
%     & \defeq {\set{ \posElm \in \QSet \cap \PosSet[\OppSym] }{
%     \MovRel(\posElm) \setminus \QSet \neq \emptyset }} \\
%     & \:\cup {\set{ \posElm \in \QSet \cap \PosSet[\PlrSym] }{ \forall
%     \wposElm \in \MovRel(\posElm) \cap \QSet \,.\, \mfElm(\wposElm) + \posElm
%     < \mfElm(\posElm) }}.
%   \end{align*}
%   \end{linenomath}

    All positions belonging to $\escFun(\mfElm, \QSet)$ must be lifted during
    the execution of $\prgFun[+]$.
    However, they must be lifted in the appropriate order: namely, the first
    positions to be lifted are those whose measure will rise the least.
    This is the role of the $\bepFun$ function and its supporting \emph{best
    escape forfeit function} $\befFun$.
    The $\bepFun$ simply collects the positions $\posElm$ having minimal
    \emph{forfeit} $\befFun(\mfElm, \QSet, \posElm)$ (defined below).
    \[
      \bepFun(\mfElm, \QSet) \defeq \underset{\posElm \in \QSet}{\arg\min} \;
      \befFun(\mfElm, \QSet, \posElm) :
      \MFSet \times \pow{\PosSet} \to \pow{\PosSet}.
    \]
    Assuming that $\posElm$ is a position in the escape set $\escFun(\mfElm,
    \QSet)$, its forfeit is the difference between the measure that $\posElm$
    would acquire if lifted and its current measure.
    In the following, we use the difference operation $\evlElm[1] - \evlElm[2]
    \defeq \evlElm[1] + (-\evlElm[2])$ defined as usual.
    \[
      \befFun(\mfElm, \QSet, \posElm) \defeq
      \begin{cases}
        \max {\set{( \mfElm(\wposElm) + \posElm) - \mfElm(\posElm) }{ \wposElm
        \in \MovRel(\posElm) \setminus \QSet }},
        & \text{if } \posElm \in \escFun(\mfElm, \QSet) \cap \PosSet[\PlrSym];
        \\
        \min {\set{ (\mfElm(\wposElm) + \posElm) - \mfElm(\posElm) }{ \wposElm
        \in \MovRel(\posElm) \setminus \QSet }},
        & \text{if } \posElm \in \escFun(\mfElm, \QSet) \cap \PosSet[\OppSym];
        \\
        \top,
        & \text{otherwise}.
      \end{cases}
    \]

    Notice that the following inclusions are an immediate consequence of the
    definitions:
    $\bepFun(\mfElm, \QSet) \subseteq \escFun(\mfElm, \QSet) \subseteq \QSet$.
    The following lemma is the core result of this section, as it proves the key
    properties of the progress operator $\prgFun[+]$.

    \begin{restatable}{lemma}{prgplsr}
      \label{lmm:prgpls}
      The progress operator $\prgFun[+]$ enjoys the following properties:
      \begin{inparaenum}[(i)]
        \item\label{lmm:prgpls(mfb)}
          it maps $\MFSet[\bot]$ into $\MFSet[\bot]$;
        \item\label{lmm:prgpls(qdm)}
          it is an inflationary function from $\QDMFSet$ to $\QDMFSet$;
        \item\label{lmm:prgpls(smf)}
          it maps $\SMFSet$ into $\SMFSet$;
        \item\label{lmm:prgpls(fix)}
          every fixpoint $\mfElm \in \MFSet$ of $\prgFun[+]$ is a progress
          measure over $\denot{\mfElm}[+]$.
      \end{inparaenum}
    \end{restatable}

    \begin{figure}[htbp]
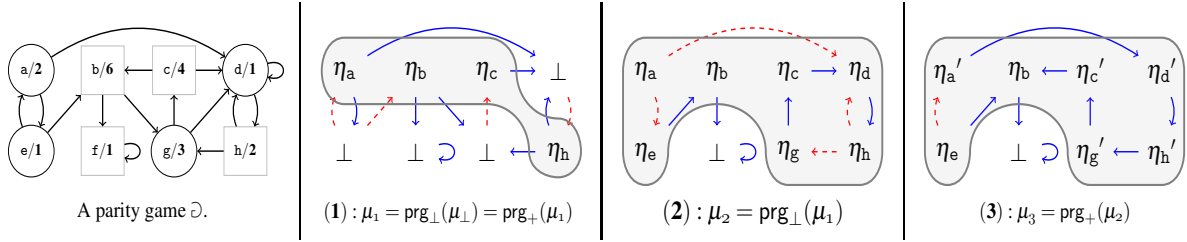

      \vspace{-1.00em}
      \begin{center}
        \tabsim
      \end{center}
      \vspace{-1.00em}
      \caption{\label{fig:exm} Simulating the first steps of the concrete
        algorithm on a simple game.
        Positions of player $\PlrSym$ are circles and positions of player
        $\OppSym$ are squares.
        The label inside each position indicates its name and priority.}
      \vspace{-0.50em}
    \end{figure}

    As an example, consider the simple game $\GamName$ in Figure~\ref{fig:exm}.
    Starting from the minimal measure function $\mfElm[\bot]$, the solution
    algorithm first computes $\mfElm[1] = \prgFun[\bot](\mfElm[\bot])$, by
    lifting the four even-priority positions $\aSym$, $\bSym$, $\cSym$, and
    $\hSym$ to their respective measures $\msrElm[\aSym] = \msrElm[\hSym] \defeq
    (0, 0, 0, 0, 1, 0, 0)$, $\msrElm[\bSym] \defeq (1, 0, 0, 0, 0, 0, 0)$, and
    $\msrElm[\cSym] \defeq (0, 0, 1, 0, 0, 0, 0)$.
    Figure~\ref{fig:exm}.1 reports the situation after this initial phase, where
    the blue (\resp, dashed red) edges indicate the moves satisfying (\resp, not
    satisfying) the progress condition.
    Since inside the quasi $0$-dominion $\denot{\mfElm[1]}[+] = \{ \aSym, \bSym,
    \cSym, \hSym \}$, identified by the grey area, all positions are in
    progress, the $\prgFun[+]$ operator does not change their measures, \ie,
    $\mfElm[1] = \prgFun[+](\mfElm[1])$.
    The three odd-priority positions $\dSym$, $\eSym$, and $\gSym$ outside
    $\denot{\mfElm[1]}[+]$ do not satisfy the progress condition, so the
    $\prgFun[\bot]$ operator lifts their measure to $\msrElm[\dSym] \defeq (0,
    0, 0, 0, 1, 1, 0)$, $\msrElm[\eSym] \defeq (1, 0, 0, 0, 0, 1, 0)$, and
    $\msrElm[\gSym] \defeq (0, 0, 1, 1, 0, 0, 0)$, as reported in
    Figure~\ref{fig:exm}.2.
    Now, positions $\aSym$ and $\hSym$ inside the quasi $0$-dominion do not
    satisfy the progress condition anymore.
    Therefore, the $\prgFun[+]$ operator tries to recover the condition as
    follows.
    It starts by identifying the escape positions $\escFun(\mfElm[2], \QSet[0])
    = \{ \bSym \}$ of $\QSet[0] \defeq \denot{\mfElm[2]}[+] = \{ \aSym, \bSym,
    \cSym, \dSym, \eSym, \gSym, \hSym \}$.
    Since $\bSym$ has a progress move exiting from $\QSet[0]$, its measure
    remains unchanged.
    Now, $\escFun(\mfElm[2], \QSet[1]) = \{ \cSym, \eSym \}$, where $\QSet[1]
    \defeq \QSet[0] \setminus \{ \bSym \}$.
    Since $\eSym$ has a progress move exiting from $\QSet[1]$, while $\cSym$ can
    escape from $\QSet[1]$ only by increasing its measure, we have
    $\bepFun(\mfElm[2][], \QSet[1]) = \{ \eSym \}$ and $\QSet[2] \defeq \QSet[1]
    \setminus \{ \eSym \}$.
    Also in this case, $\eSym$ does not change its measure.
    The process continues by extracting and lifting the measures of all the
    remaining positions in $\QSet[2]$ in the following order:
    \begin{inparaenum}[(i)]
      \item
        $\cSym$ with $\msrElm[\cSym]' \defeq (1, 0, 1, 0, 0, 0, 0)$;
      \item
        $\gSym$ with $\msrElm[\gSym]' \defeq (1, 0, 1, 1, 0, 0, 0)$;
      \item
        $\hSym$ with $\msrElm[\hSym]' \defeq (1, 0, 1, 1, 1, 0, 0)$;
      \item
        $\dSym$ with $\msrElm[\dSym]' \defeq (1, 0, 1, 1, 1, 1, 0)$;
      \item
        $\aSym$ with $\msrElm[\aSym]' \defeq (1, 0, 1, 1, 2, 1, 0)$.
    \end{inparaenum}
    Note that player-$0$ positions $\gSym$ and $\dSym$ are forced to exit the
    quasi $0$-dominion, since their internal moves $(\gSym, \dSym)$ and $(\dSym,
    \dSym)$ do not satisfy the regress condition, as $\msrElm[\dSym] + \gSym =
    \bot < \msrElm[\gSym]$ and $\msrElm[\dSym] + \dSym = (0, 0, 0, 0, 1, 2, 0) <
    \msrElm[\dSym]$; these moves would form, indeed, odd cycles.
    Figure~\ref{fig:exm}.3 reports the situation after the complete execution of
    $\prgFun[+]$, where position $\eSym$ has the non-progress move $(\eSym,
    \aSym)$.
    Another application of $\prgFun[+]$ modifies the measure of $\eSym$ to
    $\msrElm[\eSym]' \defeq (1, 0, 1, 1, 2, 2, 0)$, triggering the non-progress
    move $(\aSym, \eSym)$.
    After a final application of $\prgFun[+]$, positions $\aSym$ and $\eSym$ are
    lifted to $\top$ and the algorithm reaches its fixpoint.
    All positions except $\aSym$ and $\eSym$ satisfy the progress conditions and
    are, thus, winning for player $\OppSym$; $\aSym$ and $\eSym$ are won by
    player $\PlrSym$.

    We can prove that the solver operator is well defined and that, when it is
    applied to a simple measure, it converges in a finite number of iterations,
    at most equal to the depth of the finite partial order
    $\tupleb {\SMFSet} {\sqsubseteq}$.
    A very coarse upper bound on this depth, for a game with $n$ positions, is
    given by $(n + 1)!$, since every non-$\top$ position is associated with the
    measure of a simple path of length less than $n$ and there are at most $n!$
    such paths.
    In the following, we use $\SQDMFSet \defeq \QDMFSet \cap \SMFSet$.

    \begin{theorem}[Termination]
      \label{thm:ter}
      The solver operator $\solFun \colon \SQDMFSet \to \SQDMFSet$ is a
      well-defined function.
      Moreover, for every $\mfElm \in \SQDMFSet$, there exists an index $k \leq
      d$, such that $\solFun(\mfElm) = (\ifpFun[k]\, \nu \,.\,
      \prgFun[+](\prgFun[\bot](\nu)))(\mfElm)$, where $d \in \SetN$ is the depth
      of the finite partial order $\tupleb {\SMFSet} {\sqsubseteq}$.
    \end{theorem}
    \begin{proof}
      By Items~\ref{lmm:prgbot(mfb)}-\ref{lmm:prgbot(smf)} of
      Lemma~\ref{lmm:prgbot} and
      Items~\ref{lmm:prgpls(mfb)}-\ref{lmm:prgpls(smf)} of
      Lemma~\ref{lmm:prgpls}, we have that
      $\prgFun[\bot]$ and $\prgFun[+]$ are inflationary total functions on
      $\SQDMFSet$, which implies that their composition $\prgFun[+] \cmp
      \prgFun[\bot]$ is both inflationary and total on $\SQDMFSet$ as well.
%       Thus, thanks to the inflationary fixpoint theorem~\cite{Bou49,Wit50},
%       $\prgFun[+] \cmp \prgFun[\bot]$ always admits an inflationary fixpoint,
%       starting from any possible measure function in $\QDMFSet[\bot]$.
%       Hence, the solver operator $\solFun$ is well defined.
      Consider now the infinite sequence $\mfElm[0], \mfElm[1], \ldots$ of
      measure functions recursively derived from an arbitrary input element
      $\mfElm \in \SQDMFSet$ as follows: $\mfElm[0] \defeq (\ifpFun[0]\, \nu
      \,.\, \prgFun[+] (\prgFun[\bot](\nu)))(\mfElm) = \mfElm$ and $\mfElm[i +
      1] \defeq (\ifpFun[i + 1]\, \nu \,.\,
      \prgFun[+](\prgFun[\bot](\nu)))(\mfElm) =
      \prgFun[+](\prgFun[\bot](\mfElm[i]))$, for all $i \in \SetN$.
      Obviously, $\mfElm[i] \sqsubseteq \mfElm[i + 1]$.
      Moreover, each element $\mfElm[i]$ is a \smf.
      Since every strict chain in $\tupleb {\SMFSet} {\sqsubseteq}$ can be
      composed of at most $d$ elements, there necessarily exists an index $k
      \leq d$ such that $\mfElm[k + 1] = \mfElm[k]$, as required by the theorem
      statement.
    \end{proof}

    The next theorem stating the soundness and completeness of the solution
    algorithm is a simple consequence of the properties of the $\prgFun[\bot]$
    and $\prgFun[+]$ operators, combined with the general results about the
    measure-function spaces discussed in the previous sections.

  \begin{theorem}[Solution]
    \label{thm:sol}
    $\WinSet[\PlrSym] = \denot{\solFun(\mfElm[\bot])}[\PlrSym]$ and
    $\WinSet[\OppSym] = \denot{\solFun(\mfElm[\bot])}[\OppSym]$.
  \end{theorem}
  \begin{proof}
    Let $\mfElm[][\star] \defeq \solFun(\mfElm[\bot])$ be the result of the
    application of the solver operator to the minimal measure function
    $\mfElm[\bot] \in \SQDMFSet$.
    By the notion of inflationary fixpoint, $\mfElm[][\star]$ is a fixpoint of
    the composition of the two progress operators, \ie, $\mfElm[][\star] =
    \prgFun[+](\prgFun[\bot](\mfElm[][\star]))$, which are inflationary
    functions on $\SQDMFSet$, due to Items~\ref{lmm:prgbot(mfb)}
    and~\ref{lmm:prgbot(qdm)} of Lemma~\ref{lmm:prgbot} and
    Items~\ref{lmm:prgpls(mfb)} and~\ref{lmm:prgpls(qdm)} of
    Lemma~\ref{lmm:prgpls}.
    Therefore, it holds that $\mfElm[][\star] \sqsubseteq
    \prgFun[\bot](\mfElm[][\star]) \sqsubseteq
    \prgFun[+](\prgFun[\bot](\mfElm[][\star])) = \mfElm[][\star]$, which
    implies that $\prgFun[\bot](\mfElm[][\star]) = \mfElm[][\star]$ and, so,
    $\prgFun[+](\mfElm[][\star]) = \mfElm[][\star]$.
    As a consequence of Item~\ref{lmm:prgbot(fix)} of Lemma~\ref{lmm:prgbot}
    and Item~\ref{lmm:prgpls(fix)} of Lemma~\ref{lmm:prgpls}, it holds that
    $\mfElm[][\star]$ is a progress measure.
    Hence, $\denot{\mfElm[][\star]}[\OppSym] \subseteq \WinSet[\OppSym]$
    follows from Theorem~\ref{thm:prgmsr}.
    By Theorem~\ref{thm:ter}, it holds that $\mfElm[][\star] \in \SQDMFSet$,
    which implies that $\denot{\mfElm[][\star]}[\PlrSym] \subseteq
    \WinSet[\PlrSym]$, due to Definition~\ref{def:qsidommsr}.
    Hence, the thesis follows, since $\denot{\mfElm[][\star]}[\PlrSym]$ and
    $\denot{\mfElm[][\star]}[\OppSym]$ partition the set of positions.
  \end{proof}

  \end{subsection}

\end{section}
\vspace*{-.8em}
% End of file SectionIII.tex

%%****************************************************************************%%
%%                                                                            %%
%% A Delayed Promotion Policy for Parity Games                                %%
%%                                                                            %%
%% SectionIV.tex                                                             %%
%%                                                                            %%
%% Revision 0                                                                 %%
%%                                                                            %%
%% Copyright (C) 2016, Massimo Benerecetti, Daniele Dell'Erba, and            %%
%%                     Fabio Mogavero.                                        %%
%% All rights reserved.                                                       %%
%%                                                                            %%
%%****************************************************************************%%

% Begin of file SectionIV.tex

\vspace{-0.75em}
\begin{section}{Experimental Evaluation}
  \label{sec:expevl}

  The algorithm proposed in the paper has been implemented in OINK~\cite{Dij18},
  a C++ framework supporting different parity game solvers\footnote{Experiments
  were carried out on a 64-bit 1.6GHz \textsc{Intel\textregistered} quad-core
  machine, with i5-8250U processor and 8GB of RAM, running
  \textsc{Ubuntu}~18.04.5 with \textsc{Linux} kernel version~3.28.2.
  OINK was compiled with gcc version 7.4.} and providing tools to
  compare their performance on various worst-case families.
  The solvers considered in the experiments include the original priority
  promotion algorithm \emph{PP}~\cite{BDM16} and the progress measure version
  presented in this paper \emph{QDPM}, the optimised version~\cite{LDT14} of the
  Recursive algorithm \emph{Rec}~\cite{Zie98}, the optimised version of the
  Small Progress Measure algorithm \emph{SPM}~\cite{Jur00} and its
  quasi-polynomial version \emph{SSPM}~\cite{JL17}, the quasi-polynomial
  algorithm \emph{QPT}~\cite{FJSSW17}, the Tangle Learning algorithm
  \emph{TL}~\cite{Dij18a}, and the Distraction-based Fixpoint Iteration
  algorithm with justifications \emph{FPJ}~\cite{LBD20}.
  The benchmarks include worst-case games for the considered solvers and
  clustered random games generated with the PGSolver framework~\cite{FL09}.
  The latter are games that exhibit a complex structure \wrt the class of
  randomly generated games.
  Indeed, while most of the random games consist in a single \emph{strongly
  connected component (SCC)} and are easily solved by any attractor-based
  approach, clustered games rely, instead, on a tree-like structure with
  multiple SCCs.

  \begin{figure}[hbtp]
    \vspace{-0.50em}
    \begin{center}
      \footnotesize
      \scalebox{0.95}[0.95]{\tabper}
      %\scalebox{0.95}[0.95]{\tabperx}
    \end{center}
    \vspace{-1.00em}
    \caption{\label{tab:per} \small Biggest (index of the) instance of the
      worst-case families solved within 30s.
      If the 1000th instance is solved, its approximated solution time is
      reported in brackets.}
    \vspace{-1.00em}
  \end{figure}

  Table~\ref{tab:per} displays the results on the following worst-case
  families\footnote{The instances were generated by issuing the following OINK
  commands: tc+ n; counter\_qpt n; counter\_m n; counter\_dp n.
  The Robust Worst Case has been implemented according to~\cite{BDM20}.}: the
  family for TL ``\emph{Two Counters}''~\cite{Dij19}, the one for QPT algorithm
  ``\emph{QPT}''~\cite{FJSSW17}, the family for Zielonka's Optimised algorithm
  ``\emph{Gazda's wc}''~\cite{Gaz16}, the family for the Delayed Priority
  Promotion algorithm ``\emph{DP}''~\cite{BDM18a}, and the Robust Worst Case
  for Divide-et-Impera Algorithms ``\emph{Divide\&Impera}''~\cite{BDM20}.
  Each row reports the biggest instance each solver could solve within the time
  limit of 30 seconds.
  The ``Two Counters'' family proved to be very demanding for all the solvers,
  as none of them could solve the 108th instance within the time limit.
  On the contrary, the ``QPT'' family can easily be solved by all the solvers
  except QPT.
  The proposed solver QDPM performs extremely well on all the families, being
  able to solve the 1000th instance faster than the competitors, except for the
  ``Two Counters'', on which it is outperformed only by PP, and Gazda's
  family, where TL is slightly better.

  \begin{wrapfigure}[13]{r}[0em]{0.55\textwidth}
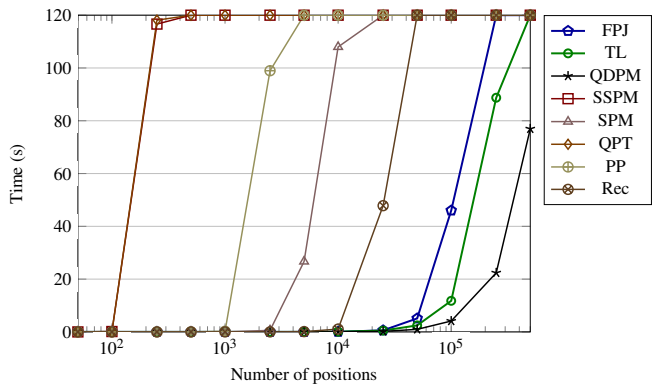

    \vspace{-2.75em}
    \begin{center}
      \figexp
    \end{center}
    \vspace{-1.50em}
    \caption{\label{fig:exp} \small Time on clustered random games with 2 moves
      per position.}
  \end{wrapfigure}
  Figure~\ref{fig:exp} compares the running times on $1300$ random clustered
  games of size ranging from $50$ to $5 \cdot 10^5$ positions and $2$ outgoing
  moves per position~\footnote{The instances were generated by issuing the
  following PGSolver command: clusteredrandomgame $n$ $n/10$ 2 2 5 3 7 3 7.}.
  We set the time-out at 120 seconds.
  Each point in the graph shows the average time over a cluster of $100$
  different games of the same size shown on a logarithmic scale.
  For a game of size $n$, we set the number of priorities to $k = n / 10$.
  The performance of the quasi-polynomial solvers (QPT and SSPM) reaches the
  timeout already for the smaller instances: for games with 250 positions they
  could solve less than the 20\% of the instances.
  Almost all of the quasi-dominion-based algorithms, namely TL, FPJ, and QDPM, instead,
  scale quite well.
  Their behaviour start to differentiate for games with at least $10^5$
  positions.
  On the biggest instances ($5 \cdot 10^5$), QDPM is the only algorithm to
  terminate within 2 minutes, with an average solution time of 76 seconds
  and only 14\% of timeouts.

\end{section}

% End of file SectionIV.tex

%%% Local Variables:
%%% ispell-local-dictionary: "british"
%%% mode: latex
%%% TeX-master: "Article"
%%% End:

  % \input{SectionV}

  % \input{SectionVI}

  % \input{SectionVII}

  % \input{SectionVIII}

%%****************************************************************************%%
%%                                                                            %%
%% A Delayed Promotion Policy for Parity Games                                %%
%%                                                                            %%
%% Discussion.tex                                                             %%
%%                                                                            %%
%% Revision 0                                                                 %%
%%                                                                            %%
%% Copyright (C) 2016, Massimo Benerecetti, Daniele Dell'Erba, and            %%
%%                     Fabio Mogavero.                                        %%
%% All rights reserved.                                                       %%
%%                                                                            %%
%%****************************************************************************%%

% Begin of file Discussion.tex

\vspace{-.75em}
\begin{section}{Discussion}

\vspace{-.4em}
We propose a revisited progress measures-based algorithm for parity games that
integrates progress measures and quasi-dominions.
This integration requires a novel notion of measure to encode the
additional information needed to identify quasi-dominions and a new update
policy that takes advantage of quasi-dominions and often allows to skip
intermediate measures and reach a progress measure much more quickly than the
classic progress measure algorithms.
This motivates the conjecture that the integration significantly accelerates the
convergence to a progress measure.
The experiments show that the proposed approach scales better than any known
algorithm on games with a complex structure, such as clustered random games and
the worst-case families.
In particular, the speed-up can be of several orders of magnitude
when compared to other algorithms based on progress measures.
We believe that this integration approach may also lead to practically efficient
quasi-polynomial algorithms based on succinct progress measures.

\end{section}

% End of file Discussion.tex

%%% Local Variables:
%%% ispell-local-dictionary: "british"
%%% mode: latex
%%% TeX-master: "Article"
%%% End:

  % \input{Conclusion}

  \newpage
  
%%****************************************************************************%%
%%                                                                            %%
%% Article Title                                                              %%
%%                                                                            %%
%% Acknowledgments.tex                                                        %%
%%                                                                            %%
%% Revision 0                                                                 %%
%%                                                                            %%
%% Copyright (C) 20xx, Fabio Mogavero.                                        %%
%% All rights reserved.                                                       %%
%%                                                                            %%
%%****************************************************************************%%

% Begin of file Acknowledgments.tex

\begin{section}*{Acknowledgments}

  F.~Mogavero acknowledges a partial support by GNCS 2020 project ``Ragionamento
  Strategico e Sintesi Automatica di Sistemi Multi-Agente''.
  This project has received funding from the European Union's Horizon 2020
  research and innovation programme under the Marie Sklodowska-Curie grant
  agreement No 101032464.
  The work was supported by EPSRC grant EP/P020909/1.

\end{section}

% End of file Acknowledgments.tex

  \bibliographystyle{eptcs}
  \bibliography{References}

  \newpage
  \appendix

%%****************************************************************************%%
%%                                                                            %%
%% Article Title                                                              %%
%%                                                                            %%
%% AppendixA.tex                                                               %%
%%                                                                            %%
%% Revision 0                                                                 %%
%%                                                                            %%
%% Copyright (C) 2020, Massimo Benerecetti, Daniele Dell'Erba, Marco Faella,  %%
%%                     and Fabio Mogavero.                                    %%
%% All rights reserved.                                                       %%
%%                                                                            %%
%%****************************************************************************%%

% Begin of file AppendixA.tex

\begin{section}{Appendix of Section~\ref{sec:prgmsr}}

  \prgmsrr*
  \begin{proof}
    Let us consider an arbitrary $\mfElm$-coherent $\OppSym$-strategy
    $\strElm[\OppSym] \in \StrSet[\OppSym]$.
    All measures $\mfElm(\posElm)$ of positions $\posElm \in
    \denot{\mfElm}[\OppSym] \cap \PosSet[\OppSym]$ are a progress for
    $\posElm$ \wrt the measures $\mfElm(\strElm[\OppSym](\posElm))$ of their
    adjacents $\strElm[\OppSym](\posElm)$, \ie, formally,
    $\mfElm(\strElm[\OppSym](\posElm)) + \posElm \leq \mfElm(\posElm)$.
    The existence of such a coherent strategy is ensured by the fact that
    $\mfElm$ is a progress measure.
    Indeed, by Condition~\ref{def:prgmsr(opp)} of Definition~\ref{def:prgmsr},
    there necessarily exists a adjacent $\wposElm[][\star] \in
    \MovRel(\posElm)$ of $\posElm$ such that $\mfElm(\wposElm[][\star]) +
    \posElm \leq \mfElm(\posElm)$.

    It can be shown that $\strElm[\OppSym]$ is a winning strategy for
    Player~$\OppSym$ from all the positions in $\denot{\mfElm}[\OppSym]$,
    which implies that $\denot{\mfElm}[\OppSym] \subseteq \WinSet[\OppSym]$.
    To do this, consider a $\PlrSym$-strategy $\strElm[\PlrSym] \in
    \StrSet[\PlrSym]$ and the associated play $\pthElm =
    \playFun((\strElm[\PlrSym], \strElm[\OppSym]), \posElm)$ starting at a
    position $\posElm \in \denot{\mfElm}[\OppSym]$.
    Assume by contradiction that $\pthElm$ is won by Player~$\PlrSym$.
    Since the game $\GamName$ is finite and the strategies are memoryless,
    $\pthElm$ must contain a finite simple cycle, and so a finite simple path,
    and the maximal priority seen infinitely often along it needs to be even.
    In other words, there exist two natural numbers $h \in \SetN$ and $k \in
    \SetN[+]$ such that $(\pthElm)_{h} = (\pthElm)_{h + k}$ and $\prtFun(\rho)
    \equiv_2 0$, where $\rho \defeq ((\pthElm)_{\geq h})_{< h + k}$ is the
    simple path named above.
    Moreover, one can choose the value of the index $h$ in such a way that
    $\prtFun((\pthElm)_{h}) \geq \prtFun((\pthElm)_{i})$, for all $i \in
    \SetN$ with $h < i < h + k$.
    Recall that $((\pthElm)_{i}, (\pthElm)_{i + 1}) \in \MovRel$, for all
    indexes $i \in \SetN$.
    Thanks to the two conditions of Definition~\ref{def:prgmsr} and the notion
    of play, it holds that
    \[
      \mfElm((\pthElm)_{i + 1}) + (\pthElm)_{i} \leq \mfElm((\pthElm)_{i}).
    \]
    By Item~\ref{def:msrspc(trn:mon)} of Definition~\ref{def:msrspc}, for
    all indexes $h \leq i < h + k$, it immediately follows that
    \[
      (\mfElm((\pthElm)_{i + 1}) + (\pthElm)_{i}) \rst_{(\pthElm)_{h}} \leq
      \mfElm((\pthElm)_{i}) \rst_{(\pthElm)_{h}}.
      \tag{$\ast$}
    \]

    At this point, it is important to observe that $\mfElm$ cannot associate
    the maximum value $\top$ with any position in the play, in particular when
    restricted to $(\pthElm)_{h}$, which means that the entire path is
    contained into $\denot{\mfElm}[\OppSym]$.
    The first element of the play trivially satisfies such a constraint, as
    $\mfElm((\pthElm)_{0}) = \mfElm(\posElm) \neq \top$, since $\posElm \in
    \denot{\mfElm}[\OppSym]$.
    Hence, $\mfElm((\pthElm)_{0}) \rst_{(\pthElm)_{h}} \neq \top$, by
    Item~\ref{def:msrspc(trn:top)} of Definition~\ref{def:msrspc}.
    Now, suppose by contradiction that $\mfElm((\pthElm)_{i + 1})
    \rst_{(\pthElm)_{h}} = \top$, for some index $i \in \SetN[+]$.
    By Inequality~($\ast$) and Item~\ref{def:prgmsrspc(geq)} of
    Definition~\ref{def:prgmsrspc}, it follows that
    \[
      \top = \mfElm((\pthElm)_{i + 1}) \rst_{(\pthElm)_{h}} \leq
      (\mfElm((\pthElm)_{i + 1}) + (\pthElm)_{i}) \rst_{(\pthElm)_{h}} \leq
      \mfElm((\pthElm)_{i}) \rst_{(\pthElm)_{h}},
    \]
    being $\prtFun((\pthElm)_{h})$ the maximal priority along the path $\rho$
    that is also even.
    Due to the maximality of $\top$ ensured by Item~\ref{def:msrspc(ord)} of
    Definition~\ref{def:msrspc}, it obviously follows that
    $\mfElm((\pthElm)_{i}) \rst_{(\pthElm)_{h}} = \top$ as well.
    Therefore, by iterating this process until index $0$ is reached, one would
    obtain $\mfElm((\pthElm)_{0}) \rst_{(\pthElm)_{h}} = \top$, which is
    impossible, as previously observed.

    To complete the proof, we can exploit the properties of the progress
    measure space.
    Recall that $\mfElm((\pthElm)_{h + 1}) \rst_{(\pthElm)_{h}} < \top$,
    thanks to the above observation.
    Thus, by Item~\ref{def:prgmsrspc(grt)} of
    Definition~\ref{def:prgmsrspc} and the fact that $\prtFun((\pthElm)_{h})$
    is an even priority, one can derive that
    \[
      \mfElm((\pthElm)_{h + 1}) \rst_{(\pthElm)_{h}} < (\mfElm((\pthElm)_{h +
      1}) + (\pthElm)_{h}) \rst_{(\pthElm)_{h}}.
      \tag{$<$}
    \]
    Moreover, by applying again Item~\ref{def:prgmsrspc(geq)} of
    Definition~\ref{def:prgmsrspc} and due to the fact that
    $\prtFun((\pthElm)_{h})$ is the maximal priority in the cycle, for all
    indexes $h < i < h + k$, it holds that
    \[
      \mfElm((\pthElm)_{i + 1}) \rst_{(\pthElm)_{h}} \leq (\mfElm((\pthElm)_{i
      + 1}) + (\pthElm)_{i}) \rst_{(\pthElm)_{h}}.
      \tag{$\leq$}
    \]
    As a consequence of the transitivity of the order relation between
    measures, by putting together Inequalities~($\ast$), ($<$), and~($\leq$),
    one would therefore obtain
    \[
      \mfElm((\pthElm)_{h + k}) \rst_{(\pthElm)_{h}} < \mfElm((\pthElm)_{h})
      \rst_{(\pthElm)_{h}}.
    \]
    However, $(\pthElm)_{h + k} = (\pthElm)_{h}$, leading to
    $\mfElm((\pthElm)_{h}) \rst_{(\pthElm)_{h}} < \mfElm((\pthElm)_{h})
    \rst_{(\pthElm)_{h}}$, which is clearly impossible, being $<$ an
    irreflexive relation.
  \end{proof}

\end{section}

% End of file AppendixA.tex

%%****************************************************************************%%
%%                                                                            %%
%% Article Title                                                              %%
%%                                                                            %%
%% AppendixB.tex                                                               %%
%%                                                                            %%
%% Revision 0                                                                 %%
%%                                                                            %%
%% Copyright (C) 2020, Massimo Benerecetti, Daniele Dell'Erba, Marco Faella,  %%
%%                     and Fabio Mogavero.                                    %%
%% All rights reserved.                                                       %%
%%                                                                            %%
%%****************************************************************************%%

% Begin of file AppendixB.tex

\begin{section}{Appendix of Section~\ref{sec:qsidommsr}}

  \regmsrr*
  \begin{proof}
    \Mutatismutandis, the proof proceeds similarly to the one previously
    presented for Theorem~\ref{thm:prgmsr}.

    First of all, let us consider an arbitrary $\mfElm$-coherent
    $\PlrSym$-strategy $\strElm[\PlrSym] \in \StrSet[\PlrSym]$.
    All measures $\mfElm(\posElm)$ of positions $\posElm \in
    \denot{\mfElm}[\PlrSym] \cap \PosSet[\PlrSym]$ are a regress for $\posElm$
    \wrt the measures $\mfElm(\strElm[\PlrSym](\posElm))$ of their adjacents
    $\strElm[\PlrSym](\posElm)$, \ie, formally, $\mfElm(\posElm) \leq
    \mfElm(\strElm[\PlrSym](\posElm)) + \posElm$.
    The existence of such a coherent strategy is ensured by the fact that
    $\mfElm$ is a regress measure.
    Indeed, by Condition~\ref{def:regmsr(plr)} of Definition~\ref{def:regmsr},
    there necessarily exists a adjacent $\wposElm[][\star] \in
    \MovRel(\posElm)$ of $\posElm$ such that $\mfElm(\posElm) \leq
    \mfElm(\wposElm[][\star]) + \posElm$.

    To prove that $\QSet \defeq \denot{\mfElm}[+] \setminus
    \denot{\mfElm}[\PlrSym]$ is a weak quasi $\PlrSym$-dominion with
    $\strElm[\PlrSym] \downarrow \QSet \in \StrSet[\PlrSym](\QSet)$ as
    $\PlrSym$-witness, consider a $\OppSym$-strategy $\strElm[\OppSym] \in
    \StrSet[\OppSym](\QSet)$ such that the associated play $\pthElm =
    \playFun((\strElm[\PlrSym] \downarrow \QSet, \strElm[\OppSym]), \posElm)$
    starting at a position $\posElm \in \QSet$ is infinite.
    Now, one need to show that $\prtFun(\pthElm)$ is even.
    Assume by contradiction that this condition on the parity of the priority
    does not hold.
    Since the game $\GamName$ is finite and the strategies are memoryless,
    $\pthElm$ must contain a finite simple cycle, and so a finite simple path,
    and the maximal priority seen infinitely often along it needs to be odd.
    In other words, there exist two natural numbers $h \in \SetN$ and $k \in
    \SetN[+]$ such that $(\pthElm)_{h} = (\pthElm)_{h + k}$ and $\prtFun(\rho)
    \equiv_2 1$, where $\rho \defeq ((\pthElm)_{\geq h})_{< h + k}$ is the
    simple path named above.
    Moreover, one can choose the value of the index $h$ in such a way that
    $\prtFun((\pthElm)_{h}) \geq \prtFun((\pthElm)_{i})$, for all $i \in
    \SetN$ with $h < i < h + k$.
    Recall that $((\pthElm)_{i}, (\pthElm)_{i + 1}) \in \MovRel$, for all
    indexes $i \in \SetN$.
    Thanks to the two conditions of Definition~\ref{def:regmsr} and the notion
    of play, it holds that
    \[
      \mfElm((\pthElm)_{i}) \leq \mfElm((\pthElm)_{i + 1}) + (\pthElm)_{i}.
    \]
    By Item~\ref{def:msrspc(trn:mon)} of Definition~\ref{def:msrspc}, for
    all indexes $h \leq i < h + k$, it immediately follows that
    \[
      \mfElm((\pthElm)_{i}) \rst_{(\pthElm)_{h}} \leq (\mfElm((\pthElm)_{i +
      1}) + (\pthElm)_{i}) \rst_{(\pthElm)_{h}}.
      \tag{$\ast$}
    \]

    At this point, we can exploit the properties of the regress measure space.
    By construction, $\pthElm \in \PthSet(\QSet)$, where we recall that $\QSet
    \defeq \denot{\mfElm}[+] \setminus \denot{\mfElm}[\PlrSym]$.
    Thus, observe that both $\mfElm((\pthElm)_{i}) \rst_{(\pthElm)_{i}} \neq
    \bot$ and $\mfElm((\pthElm)_{i}) \rst_{(\pthElm)_{i}} \neq \top$ hold, for
    all $i \in \SetN$.
    This implies that $\bot < \mfElm((\pthElm)_{i}) \rst_{(\pthElm)_{i}}$ and
    $\mfElm((\pthElm)_{i}) \rst_{(\pthElm)_{h}} < \top$, thanks to
    Item~\ref{def:msrspc(trn:top)} of Definition~\ref{def:msrspc}.
    By Item~\ref{def:regmsrspc(leq)} of Definition~\ref{def:regmsrspc} and the
    fact that $\prtFun((\pthElm)_{h})$ is an odd priority, one can derive that
    \[
      (\mfElm((\pthElm)_{h + 1}) + (\pthElm)_{h}) \rst_{(\pthElm)_{h}} \leq
      \mfElm((\pthElm)_{h + 1}) \rst_{(\pthElm)_{h}}.
    \]
    Hence, by Inequality~($\ast$) and the above observations, it follows that
    \[
      \bot < \mfElm((\pthElm)_{h}) \rst_{(\pthElm)_{h}} \leq
      (\mfElm((\pthElm)_{h + 1}) + (\pthElm)_{h}) \rst_{(\pthElm)_{h}} \leq
      \mfElm((\pthElm)_{h + 1}) \rst_{(\pthElm)_{h}} < \top,
    \]
    which in turn implies
    \[
      \bot < \mfElm((\pthElm)_{h + 1}) \rst_{(\pthElm)_{h}} < \top.
    \]
    Now, by Item~\ref{def:regmsrspc(les)} of Definition~\ref{def:regmsrspc},
    one can obtain that
    \[
      (\mfElm((\pthElm)_{h + 1}) + (\pthElm)_{h}) \rst_{(\pthElm)_{h}} <
      \mfElm((\pthElm)_{h + 1}) \rst_{(\pthElm)_{h}}.
      \tag{$<$}
    \]
    Moreover, by applying again Item~\ref{def:regmsrspc(leq)} of
    Definition~\ref{def:regmsrspc} and due to the fact that
    $\prtFun((\pthElm)_{h})$ is the maximal priority in the cycle, for all
    indexes $h < i < h + k$, it holds that
    \[
      (\mfElm((\pthElm)_{i + 1}) + (\pthElm)_{i}) \rst_{(\pthElm)_{h}} \leq
      \mfElm((\pthElm)_{i + 1}) \rst_{(\pthElm)_{h}}.
      \tag{$\leq$}
    \]
    As a consequence of the transitivity of the order relation between
    measures, by putting together Inequalities~($\ast$), ($<$), and~($\leq$),
    one would derive
    \[
      \mfElm((\pthElm)_{h}) \rst_{(\pthElm)_{h}} < \mfElm((\pthElm)_{h + k})
      \rst_{(\pthElm)_{h}}.
    \]
    However, $(\pthElm)_{h + k} = (\pthElm)_{h}$, leading to
    $\mfElm((\pthElm)_{h}) \rst_{(\pthElm)_{h}} < \mfElm((\pthElm)_{h})
    \rst_{(\pthElm)_{h}}$, which is obviously impossible, being $<$ an
    irreflexive relation.
  \end{proof}

  \qsidommsrir*
  \begin{proof}
    Let $\QSet \defeq \denot{\mfElm}[+]$ and $\strElm[\PlrSym] \in
    \StrSet[\PlrSym]$ be an arbitrary $\mfElm$-coherent $\PlrSym$-strategy
    that is winning on $\denot{\mfElm}[\PlrSym]$.
    To prove that $\QSet$ is a weak quasi $\PlrSym$-dominion with
    $\strElm[\PlrSym] \downarrow \QSet \in \StrSet[\PlrSym](\QSet)$ as
    $\PlrSym$-witness, consider any $\OppSym$-strategy $\strElm[\OppSym] \in
    \StrSet[\OppSym](\QSet)$ such that the play $\pthElm =
    \playFun((\strElm[\PlrSym] \downarrow \QSet, \strElm[\OppSym]), \posElm)$
    from position $\posElm \in \QSet$ is infinite.
    We need to show that $\prtFun(\pthElm)$ is even.
    The following two different cases may arise, where $\HSet \defeq
    \denot{\mfElm}[+] \setminus \denot{\mfElm}[\PlrSym] \subseteq \QSet$:
    \begin{itemize}
      \item
        {[$\pthElm \in \PthSet(\HSet)$].}
        By Theorem~\ref{thm:regmsr}, $\HSet$ is a weak quasi
        $\PlrSym$-dominion with $\strElm[\PlrSym] \downarrow \HSet =
        (\strElm[\PlrSym] \downarrow \QSet) \downarrow \HSet$ as
        $\PlrSym$-witness.
        Moreover, $\pthElm$ is a $(\strElm[\PlrSym] \downarrow \HSet,
        \posElm)$-play in $\HSet$.
        Hence, the thesis immediately follows from the definition of weak
        quasi $\PlrSym$-dominion.
      \item
        {[$\pthElm \not\in \PthSet(\HSet)$].}
        Since $\pthElm \not\in \PthSet(\HSet)$, there clearly exists and index
        $i \in \SetN$ such that $(\pthElm)_{\geq i} \in
        \PthSet(\denot{\mfElm}[\PlrSym])$.
        This follows from the fact that $\denot{\mfElm}[\PlrSym]$ is a
        $\PlrSym$-dominion with $\strElm[\PlrSym]$ as a $\PlrSym$-winning
        strategy, since every play compatible with $\strElm[\PlrSym]$ gets
        necessarily trapped in $\denot{\mfElm}[\PlrSym]$.
        Moreover, $(\pthElm)_{\geq i}$ is a $(\strElm[\PlrSym] \downarrow
        \denot{\mfElm}[\PlrSym], (\pthElm)_{i})$-play in
        $\denot{\mfElm}[\PlrSym]$.
        Hence, we immediately obtain the thesis %immediately follows
        from the definition of $\PlrSym$-dominion, since
        $\prtFun(\pthElm) = \prtFun((\pthElm)_{\geq i})$.
        \qed
    \end{itemize}
    \renewcommand{\qed}{}
  \end{proof}

  \qsidommsriir*
  \begin{proof}
    By Theorem~\ref{thm:qsidommsri}, $\QSet \defeq \denot{\mfElm}[+]$ is a
    weak quasi $\PlrSym$-dominion.
    Therefore, to prove that it is a quasi $\PlrSym$-dominion with
    $\strElm[\PlrSym] \downarrow \QSet \in \StrSet[\PlrSym](\QSet)$ as
    $\PlrSym$-witness, for some arbitrary $\mfElm$-coherent $\PlrSym$-strategy
    $\strElm[\PlrSym] \in \StrSet[\PlrSym]$ that is winning on
    $\denot{\mfElm}[\PlrSym]$, consider a $\OppSym$-strategy $\strElm[\OppSym]
    \in \StrSet[\OppSym](\QSet)$ for which the associated play $\pthElm =
    \playFun((\strElm[\PlrSym] \downarrow \QSet, \strElm[\OppSym]), \posElm)$
    starting at a position $\posElm \in \QSet$ is finite.
    Now, one need to show that $\prtFun(\pthElm)$ is even.

    Suppose by contradiction that $\prtFun(\pthElm)$ is odd.
    Then, there exists an index $h \in \SetN$ with $0 \leq h \leq n \defeq
    \card{\pthElm} - 1$ such  that $\prtFun((\pthElm)_{h}) \equiv_{2} 1$
    and $\prtFun((\pthElm)_{h}) \geq \prtFun((\pthElm)_{i})$, for all $i
    \in \SetN$ with $h \leq i \leq n$.
    Thanks to Definition~\ref{def:regmsr}, Items~\ref{def:msrspc(trn:top)}
    and~\ref{def:msrspc(trn:mon)} of Definition~\ref{def:msrspc},
    Item~\ref{def:regmsrspc(leq)} of Definition~\ref{def:regmsrspc}, and the
    facts that $\mfElm((\pthElm)_{h}) \rst_{(\pthElm)_{h}} \neq \bot$ and
    $\prtFun((\pthElm)_{h})$ is the maximal priority along the finite path
    $(\pthElm)_{\geq h}$ that is also odd, one can obtain the following
    inequality, by applying the same inductive reasoning employed in the
    second half of the proof of Theorem~\ref{thm:regmsr}:
    \[
      \bot < \mfElm((\pthElm)_{h}) \rst_{(\pthElm)_{h}} \leq
      \mfElm((\pthElm)_{n}) \rst_{(\pthElm)_{h}}.
      \tag{$\ast$}
    \]
    Since $\pthElm$ if finite, either one of the two strategies
    $\strElm[\PlrSym]$ and $\strElm[\OppSym]$ have to terminate in
    $(\pthElm)_{n}$, \ie, we necessarily have that $\wposElm[][\star]
    \in \denot{\mfElm}[\bot]$, where $\wposElm[][\star] \defeq
    \strElm[\PlrSym]((\pthElm)_{n})$, if $(\pthElm)_{n} \in \PosSet[\PlrSym]$,
    and  $\wposElm[][\star] \defeq \strElm[\OppSym]((\pthElm)_{n})$,
    otherwise.
    Hence, by Proposition~\ref{prp:botden}, it holds that
    $\mfElm(\wposElm[][\star]) = \bot$, due to the fact that $\mfElm$ is a
    simple measure.
    Moreover, again by Definition~\ref{def:regmsr}, it holds that
    \[
      \mfElm((\pthElm)_{n}) \leq \mfElm(\wposElm[][\star]) +
      (\pthElm)_{n} = \bot + (\pthElm)_{n},
    \]
    from which, by Item~\ref{def:msrspc(trn:mon)} of
    Definition~\ref{def:msrspc}, it follows that
    \[
      \mfElm((\pthElm)_{n}) \rst_{(\pthElm)_{h}} \leq (\bot +
      (\pthElm)_{n}) \rst_{(\pthElm)_{h}}.
      \tag{$\diamond$}
    \]
    At this point, by Item~\ref{def:regmsrspc(leq)} of
    Definition~\ref{def:regmsrspc} and the equality $\bot \rst_{\posElm} =
    \bot$, for the odd-priority position $\posElm$, one can obtain that
    \[
      (\bot + (\pthElm)_{n}) \rst_{(\pthElm)_{h}} \leq \bot
      \rst_{(\pthElm)_{h}} = \bot.
      \tag{$\bot$}
    \]
    Thus, as an immediate consequence of Inequalities~($\ast$),
    ($\diamond$), and~($\bot$), one would derive $\bot < \bot$, which is
    obviously impossible, being $<$ an irreflexive relation.
  \end{proof}

\end{section}

% End of file AppendixB.tex

%%****************************************************************************%%
%%                                                                            %%
%% Article Title                                                              %%
%%                                                                            %%
%% AppendixC.tex                                                               %%
%%                                                                            %%
%% Revision 0                                                                 %%
%%                                                                            %%
%% Copyright (C) 2020, Massimo Benerecetti, Daniele Dell'Erba, Marco Faella,  %%
%%                     and Fabio Mogavero.                                    %%
%% All rights reserved.                                                       %%
%%                                                                            %%
%%****************************************************************************%%

% Begin of file AppendixC.tex

\begin{section}{Appendix of Section~\ref{sec:effalg}}

  \prgbotr*
  \begin{proof}
    We analyze the four properties separately, where we recall that
    $\prgFun[\bot](\mfElm) = \liftFun(\mfElm, \denot{\mfElm}[\bot], \PosSet)$.
    \begin{itemize}
      \item
        \textbf{[\ref{lmm:prgbot(mfb)}].}
        Let $\posElm$ be a position such that $\mfElm[][\star](\posElm) \neq
        \mfElm(\posElm)$, where $\mfElm[][\star] \defeq
        \prgFun[\bot](\mfElm)$.
        By definition of the lift operator, it holds that $\posElm \in
        \denot{\mfElm}[\bot]$.
        Obviously, $\mfElm(\posElm) = \bot$, since $\mfElm \in \MFSet[\bot]$.
        Thus, by Item~\ref{def:msrspc(ord)} of Definition~\ref{def:msrspc}, it
        follows that $\mfElm(\posElm) = \bot < \mfElm[][\star](\posElm)$,
        being $\bot$ the minimal measure.
        Hence, $\mfElm \sqsubseteq \mfElm[][\star]$, due to the arbitrary
        choice of $\posElm$, which means that $\prgFun[\bot]$ is inflationary
        on $\MFSet[\bot]$, as required by the lemma statement.
        In addition, it holds that $\mfElm[][\star] \in \MFSet[\bot]$.
        Indeed, again by definition of the lift operator, there exists an
        adjacent $\wposElm \in \MovRel(\posElm)$ of $\posElm$ such that
        $\mfElm[][\star](\posElm) = \mfElm(\wposElm) + \posElm$.
        If $\posElm$ has even priority, then
        either $\mfElm(\wposElm) \rst_{\posElm} = \top$ and, so,
        $\mfElm[][\star](\posElm) \rst_{\posElm} = \top$, by
        Item~\ref{def:msrspc(ord)} of Definition~\ref{def:msrspc} and
        Item~\ref{def:prgmsrspc(geq)} of Definition~\ref{def:prgmsrspc}, or
        $\mfElm(\wposElm) \rst_{\posElm} < \top$ and $\mfElm[][\star](\posElm)
        \rst_{\posElm} = (\mfElm(\wposElm) + \posElm) \rst_{\posElm} >
        \mfElm(\wposElm) \rst_{\posElm}$, by Item~\ref{def:prgmsrspc(grt)} of
        Definition~\ref{def:prgmsrspc}.
        If $\posElm$ has odd priority, instead, by Item~\ref{prp:naimsrstr(str)}
        of Proposition~\ref{prp:naimsrstr}, it holds that if $\mfElm(\wposElm) +
        \posElm \neq \bot$ then $(\mfElm(\wposElm) + \posElm) \rst_{\posElm}
        \neq \bot$, from which it follows that $\mfElm[][\star](\posElm)
        \rst_{\posElm} = (\mfElm(\wposElm) + \posElm) \rst_{\posElm} \neq
        \bot$, since $\mfElm(\wposElm) + \posElm = \mfElm[][\star](\posElm)
        \neq \mfElm(\posElm) = \bot$.
        Thus, in all cases, $\mfElm[][\star](\posElm) \rst_{\posElm} \neq
        \bot$, \ie, $\posElm \not\in \denot{\mfElm[][\star]}[\bot]$, which
        vacuously satisfies the definitional requirement of $\MFSet[\bot]$.
      \item
        \textbf{[\ref{lmm:prgbot(qdm)}].}
        Let $\mfElm$ be a \qdmf in $\MFSet[\bot]$.
        Thanks to the above item, it is enough to prove that $\mfElm[][\star]
        \defeq \prgFun[\bot](\mfElm)$ is a \qdmf as well.
        To do this, one first needs to show that it is a regress measure.
        Consider an arbitrary position $\posElm \in \denot{\mfElm[][\star]}[+]
        \setminus \denot{\mfElm[][\star]}[\PlrSym]$.
        Then, two cases may arise.
        \begin{itemize}
          \item
            \textbf{[$\posElm \in \denot{\mfElm}[+]$].}
            By definition of the lift operator, it holds that
            $\mfElm[][\star](\posElm) =  \mfElm(\posElm)$, so, $\posElm \in
            \denot{\mfElm}[+] \setminus  \denot{\mfElm}[\PlrSym]$.
            Thus, $\posElm$ satisfies both conditions of
            Definition~\ref{def:regmsr} \wrt $\mfElm$.
            Therefore, it is quite immediate to prove that the same holds \wrt
            $\mfElm[][\star]$ too, by exploiting
            Item~\ref{def:msrspc(str:mon)} of Definition~\ref{def:msrspc},
            since $\mfElm \sqsubseteq \mfElm[][\star]$, as proved in previous
            item.
          \item
            \textbf{[$\posElm \not\in \denot{\mfElm}[+]$].}
            In this case, one needs to analyze the following two subcases.
            If $\posElm \in \PosSet[\PlrSym]$, it holds that
            $\mfElm[][\star](\posElm) = \max{\set{ \mfElm(\wposElm) + \posElm
            }{ \wposElm \in \MovRel(\posElm)}} \leq \max{\set{
            \mfElm[][\star](\wposElm) + \posElm }{ \wposElm \in
            \MovRel(\posElm)}} = \mfElm[][\star](\wposElm) + \posElm$, for
            some adjacent $\wposElm \in \MovRel(\posElm)$ of $\posElm$.
            Thus, Condition~\ref{def:regmsr(plr)} of
            Definition~\ref{def:regmsr} is satisfied.
            If $\posElm \in\PosSet[\OppSym]$, instead, it holds that
            $\mfElm[][\star](\posElm) = \min{\set{ \mfElm(\wposElm) + \posElm
            }{ \wposElm \in \MovRel(\posElm)}} \leq \min{\set{
            \mfElm[][\star](\wposElm) + \posElm }{ \wposElm \in
            \MovRel(\posElm)}} \leq \mfElm[][\star](\wposElm) + \posElm $, for
            all adjacents $\wposElm \in \MovRel(\posElm)$ of $\posElm$.
            Hence, Condition~\ref{def:regmsr(opp)} of
            Definition~\ref{def:regmsr} is satisfied as well.
            Observe that, to prove both conditions, we applied again
            Item~\ref{def:msrspc(str:mon)} of Definition~\ref{def:msrspc} and
            the fact that $\mfElm \sqsubseteq \mfElm[][\star]$.
        \end{itemize}
        \hspace{1em}
        At this point, it only remains to prove that
        $\denot{\mfElm[][\star]}[\PlrSym]$ is a $\PlrSym$-dominion.
        By hypothesis, it is known that $\denot{\mfElm}[\PlrSym]$ is a
        $\PlrSym$-dominion, being $\mfElm$ a \qdmf.
        Therefore, let us consider a position $\posElm \in
        \denot{\mfElm[][\star]}[\PlrSym] \setminus \denot{\mfElm}[\PlrSym]$.
        We can show that, again by definition of the lift operator, there
        exists an adjacent $\wposElm \in \MovRel(\posElm)$ of $\posElm$ such
        that $\wposElm \in \denot{\mfElm}[\PlrSym]$,
        if $\posElm \in \PosSet[\PlrSym]$, and all adjacents $\wposElm \in
        \MovRel(\posElm)$ of $\posElm$ satisfy $\wposElm \in
        \denot{\mfElm}[\PlrSym]$, otherwise.
        Indeed, if $\posElm \in \PosSet[\PlrSym]$, there exists an adjacent
        $\wposElm \in \MovRel(\posElm)$ of $\posElm$ such that
        $\mfElm(\wposElm) + \posElm = \mfElm[][\star](\posElm) = \top$.
        This implies that $\mfElm(\wposElm) = \top$, due to
        Item~\ref{def:msrspc(str:top)} of Definition~\ref{def:msrspc}.
        Hence, $\wposElm \in \denot{\mfElm}[\PlrSym]$.
        Similarly, if $\posElm \in \PosSet[\OppSym]$, all adjacents $\wposElm
        \in \MovRel(\posElm)$ of $\posElm$ satisfy the equality
        $\mfElm(\wposElm) + \posElm = \mfElm[][\star](\posElm) = \top$.
        Thus, again by Item~\ref{def:msrspc(str:top)} of
        Definition~\ref{def:msrspc}, it holds that $\mfElm(\wposElm) = \top$,
        which means that $\wposElm \in \denot{\mfElm}[\PlrSym]$.
        As an immediate consequence, every play starting at $\posElm$ and
        compatible with the $\PlrSym$-winning strategy on
        $\denot{\mfElm}[\PlrSym]$, suitably extended to
        $\denot{\mfElm[][\star]}[\PlrSym]$, is won by Player~$\PlrSym$.
        Therefore, $\denot{\mfElm[][\star]}[\PlrSym]$ is necessarily a
        $\PlrSym$-dominion, as required by the definition of \qdmf.
      \item
        \textbf{[\ref{lmm:prgbot(smf)}].}
        Let $\mfElm$ be a \smf and $\mfElm[][\star] \defeq
        \prgFun[\bot](\mfElm)$ the result of the $\prgFun[\bot]$ operator.
        One needs to prove that the latter is a \smf too.
        To do this, let us focus on a position $\posElm$ such that
        $\mfElm[][\star](\posElm) \neq \mfElm(\posElm)$.
        If $\mfElm[][\star](\posElm) = \top$, there is nothing more to show,
        as $\top \in \SMsrSet(\posElm, \denot{\mfElm[][\star]}[+])$, as
        required by Definition~\ref{def:simmsrfun}.
        Therefore, assume $\mfElm[][\star](\posElm) \neq \top$.
        By definition of the lift operator, there exists an adjacent $\wposElm
        \in \MovRel(\posElm)$ of $\posElm$ such that $\mfElm[][\star](\posElm)
        = \mfElm(\wposElm) + \posElm$.
        Now, by Item~\ref{def:msrspc(str:top)} of Definition~\ref{def:msrspc},
        it follows that $\mfElm(\wposElm) \neq \top$.
        Thus, thanks to the fact that $\mfElm$ is a \smf, it holds that
        $\mfElm(\wposElm) \in \SMsrSet(\wposElm, \denot{\mfElm}[+])$, which
        means that there exists a simple path $\pthElm \in \SPthSet(\wposElm,
        \denot{\mfElm}[+])$ such that $\mfElm(\wposElm) = \msrFun(\pthElm)$.
        Obviously, $\mfElm[][\star](\posElm) = \mfElm(\wposElm) + \posElm =
        \msrFun(\pthElm) + \posElm = \msrFun(\posElm \cdot \pthElm)$.
        Moreover, $\posElm \cdot \pthElm$ is a simple path passing through
        positions in $\{ \posElm \} \cup \denot{\mfElm}[+]$, \ie, $\posElm
        \cdot \pthElm \in \SPthSet(\posElm, \{ \posElm \} \cup
        \denot{\mfElm}[+])$, since $\posElm \not\in \denot{\mfElm}[+]$.
        As shown at the end of the proof of the first item of this lemma,
        $\mfElm[][\star](\posElm) \rst_{\posElm} \neq \bot$, so, $\posElm
        \in \denot{\mfElm[][\star]}[+]$.
        Thus, as an immediate consequence, one obtains that $\posElm \cdot
        \pthElm \in \SPthSet(\posElm, \denot{\mfElm[][\star]}[+])$, being $\{
        \posElm \} \cup \denot{\mfElm}[+] \subseteq
        \denot{\mfElm[][\star]}[+]$, which implies that
        $\mfElm[][\star](\posElm) \in \SMsrSet(\posElm,
        \denot{\mfElm[][\star]}[+])$.
        Hence, $\mfElm[][\star]$ is a \smf.
      \item
        \textbf{[\ref{lmm:prgbot(fix)}].}
        Let $\mfElm$ be a fixpoint of $\prgFun[\bot]$, \ie, $\mfElm =
        \liftFun(\mfElm, \denot{\mfElm}[\bot], \PosSet)$, and $\posElm \in
        \denot{\mfElm}[\bot]$ an arbitrary position.
        If $\posElm \in \PosSet[\PlrSym]$, it holds that $\mfElm(\wposElm) +
        \posElm \leq \max{\set{ \mfElm(\wposElm) + \posElm }{ \wposElm \in
        \MovRel(\posElm)}} = \mfElm(\posElm)$, for all adjacents $\wposElm \in
        \MovRel(\posElm)$ of $\posElm$.
        Thus, Condition~\ref{def:prgmsr(plr)} of Definition~\ref{def:prgmsr}
        is satisfied on $\denot{\mfElm}[\bot]$.
        If $\posElm \in\PosSet[\OppSym]$, instead, it holds that
        $\mfElm(\wposElm) + \posElm = \min{\set{ \mfElm(\wposElm) + \posElm }{
        \wposElm \in \MovRel(\posElm)}} = \mfElm(\posElm)$, for some adjacent
        $\wposElm \in \MovRel(\posElm)$ of $\posElm$.
        Hence, Condition~\ref{def:prgmsr(opp)} of Definition~\ref{def:prgmsr}
        is satisfied on $\denot{\mfElm}[\bot]$ as well.
        \qed
    \end{itemize}
    \renewcommand{\qed}{}
  \end{proof}

  \prgplsr*
  \begin{proof}
    Let us assume $\denot{\mfElm}[+] \neq \emptyset$, since there is nothing
    to prove, otherwise, being $\mfElm[][\star] \defeq \prgFun[+](\mfElm) =
    \mfElm$, and consider the three (potentially) infinite sequences
    $\QSet[0], \QSet[1], \ldots$, $\ESet[0], \ESet[1], \ldots$, and
    $\mfElm[0], \mfElm[1], \ldots$ generated by Algorithm~\ref{alg:prg}, which
    explicitly implements the progress operator $\prgFun[+]$.
    These sequences are defined as follows:
    \begin{inparaenum}[(i)]
      \item
        $\QSet[0] \defeq \denot{\mfElm}[+]$ and $\mfElm[0] \defeq \mfElm$;
      \item
        $\QSet[i + 1] \defeq \QSet[i] \setminus \ESet[i]$ and $\mfElm[i + 1] =
        \liftFun(\mfElm[i], \ESet[i], \dual{\QSet[i]})$, where $\ESet[i]
        \defeq \bepFun(\mfElm[i], \QSet[i]) \subseteq \escFun(\mfElm[i],
        \QSet[i])$, for all $i \in \SetN$.
    \end{inparaenum}
    Since $\card{\QSet[0]} < \infty$ and $\QSet[i + 1] \subseteq \QSet[i]$,
    there necessarily exists an index $k \in \SetN$ such that $\QSet[k + 1] =
    \QSet[k]$, $\mfElm[k + 1] = \mfElm[k]$, $\ESet[k] = \emptyset$, and
    $\ESet[j] \neq \emptyset$, for all $j < k$.
    Moreover, observe that $\mfElm[][\star] = {\mfElm[k]}[\QSet[k] \mapsto
    \top]$.
    At this point, we analyze the four properties separately.
    \begin{itemize}
      \item
        \textbf{[\ref{lmm:prgpls(mfb)}].}
        To prove that $\mfElm[][\star] \in \MFSet[\bot]$, whenever $\mfElm \in
        \MFSet[\bot]$, one can focus on those positions $\posElm$ that changed
        their measure from $\mfElm$ to $\mfElm[][\star]$, \ie, such that
        $\mfElm[][\star](\posElm) \neq \mfElm(\posElm)$.
        If $\posElm \in \denot{\mfElm[][\star]}[+]$, there is nothing to
        prove, as $\posElm$ vacuously satisfies the definitional requirement
        of $\MFSet[\bot]$.
        If $\posElm \not\in \denot{\mfElm[][\star]}[+]$, instead, it holds
        that $\mfElm[][\star](\posElm) \rst_{\posElm} = \bot$.
        Due to the fact that the position changed its measure, there is an
        index $i \in \numcc{0}{k}$ such that $\posElm \in \ESet[i]$ and, so,
        $\mfElm[][\star](\posElm) = \mfElm[i + 1](\posElm)$.
        Therefore, by definition of the lift operator, there exists an
        adjacent $\wposElm \in \MovRel(\posElm) \setminus \QSet[i]$ such that
        $\mfElm[i + 1](\posElm) = \mfElm[i](\wposElm) + \posElm$.
        At this point, we can observe that $\posElm$ has odd priority.
        Indeed, if by contradiction $\posElm$ has even priority, by
        Item~\ref{def:msrspc(ord)} of Definition~\ref{def:msrspc} and
        Item~\ref{def:prgmsrspc(geq)} of Definition~\ref{def:prgmsrspc},
        one would have $\mfElm[i](\wposElm) \rst_{\posElm} = \bot \neq \top$,
        since $\mfElm[i](\wposElm) \rst_{\posElm} \leq (\mfElm[i](\wposElm) +
        \posElm) \rst_{\posElm} = \mfElm[][\star](\posElm) \rst_{\posElm} =
        \bot$, which would in turn imply $\bot = \mfElm[i](\wposElm)
        \rst_{\posElm} < (\mfElm[i](\wposElm) + \posElm) \rst_{\posElm} =
        \mfElm[][\star](\posElm) \rst_{\posElm} = \bot$, due to
        Item~\ref{def:prgmsrspc(grt)} of the same definition, which is
        obviously impossible, being $<$ an irreflexive relation.
        Thus, as a consequence of Item~\ref{prp:naimsrstr(str)} of
        Proposition~\ref{prp:naimsrstr}, $\mfElm[i](\wposElm) + \posElm =
        \bot$, since $(\mfElm[i](\wposElm) + \posElm) \rst_{\posElm} = \bot$,
        which means that $\mfElm[][\star](\posElm) = \bot$, as required by the
        definition of the set $\MFSet[\bot]$.
        Hence, $\mfElm[][\star] \in \MFSet[\bot]$
      \item
        \textbf{[\ref{lmm:prgpls(qdm)}].}
        We first prove the inflationary property of the progress operator.
        To do this, consider the sequence of forfeit values $\evlElm[0],
        \ldots, \evlElm[k - 1] \in \EvlSet$ defined as follows: $\evlElm[i]
        \defeq \min[\posElm \in {\QSet[i]}] \befFun(\mfElm[i], \QSet[i],
        \posElm)$, for all indexes $i \in \numco{0}{k}$.
        In addition, let $\iota \colon \denot{\mfElm}[+] \setminus \QSet[k]
        \to \numco{0}{k}$ be the function associating each position $\posElm
        \in \denot{\mfElm}[+] \setminus \QSet[k]$ with the index
        $\iota(\posElm) \in \numco{0}{k}$ such that $\posElm \in
        \ESet[\iota(\posElm)]$.
        Due to the way the sequence of measure functions $\mfElm[0],
        \mfElm[1], \ldots$ is constructed, it is immediate to observe that,
        for all positions $\posElm \in \PosSet$ and indexes $i \in
        \numcc{0}{k}$, it holds that
        \[
          \text{ if } \posElm \in \denot{\mfElm}[+] \setminus \QSet[k] \text{
          and } \iota(\posElm) < i \text{ then } \mfElm[i](\posElm) =
          \mfElm[\iota(\posElm)](\posElm) + \evlElm[\iota(\posElm)] \text{
          else } \mfElm[i](\posElm) = \mfElm(\posElm).
          \tag{$\ast$}
        \]
        At this point, by induction on the index $i \in \numco{0}{k}$ and in
        that specific order, we can prove the following four auxiliary
        properties:
        \begin{inparaenum}[(a)]
          \item\label{lmm:prgpls(qdm:inf)}
            $\mfElm \sqsubseteq \mfElm[i]$;
          \item\label{lmm:prgpls(qdm:adj)}
            for each position $\posElm \in \denot{\mfElm}[+] \setminus \QSet[i
            + 1]$, there exists an adjacent $\wposElm \in \MovRel(\posElm)
            \setminus \QSet[i]$ such that $\mfElm[i + 1](\posElm) =
            \mfElm[i](\wposElm) + \posElm \neq \top$;
          \item\label{lmm:prgpls(qdm:mon)}
            if $i > 0$ then $\evlElm[i - 1] \leq \evlElm[i]$;
          \item\label{lmm:prgpls(qdm:bot)}
            $\mathbf{0} \leq \evlElm[i]$.
        \end{inparaenum}
        \begin{itemize}
          \item
            \textbf{[\ref{lmm:prgpls(qdm:inf)}].}
            If $i = 0$ then $\mfElm \sqsubseteq \mfElm[i]$, since $\mfElm[0] =
            \mfElm$.
            If $i > 0$, instead, by the Inductive
            Hypotheses~\ref{lmm:prgpls(qdm:inf)}
            and~\ref{lmm:prgpls(qdm:bot)}, it holds that $\mfElm \sqsubseteq
            \mfElm[i - 1]$ and $\mathbf{0} \leq \evlElm[i - 1]$.
            Moreover, by the previous Observation~($\ast$), it follows that
            $\mfElm[i](\posElm) \neq \mfElm[i - 1](\posElm)$ only if $\posElm
            \in \ESet[i - 1]$ and, in this case, $\mfElm[i](\posElm) =
            \mfElm[i - 1](\posElm) + \evlElm[i - 1]$, since $\iota(\posElm) =
            i - 1$.
            As a consequence, if $\mfElm[i](\posElm) \neq \mfElm[i -
            1](\posElm)$ then $\mfElm[i](\posElm) > \mfElm[i - 1](\posElm)$,
            thanks to Proposition~\ref{prp:naievlstr}, which implies, in
            general, that $\mfElm[i](\posElm) \geq \mfElm[i - 1](\posElm) \geq
            \mfElm(\posElm)$.
            Hence, $\mfElm \sqsubseteq \mfElm[i]$ holds true.
          \item
            \textbf{[\ref{lmm:prgpls(qdm:adj)}].}
            Let $\posElm \in \denot{\mfElm}[+] \setminus \QSet[i + 1]$.
            If $\posElm \in \denot{\mfElm}[+] \setminus \QSet[i]$ then $i >
            0$, since $\QSet[0] = \denot{\mfElm}[+]$.
            By the Inductive Hypothesis~\ref{lmm:prgpls(qdm:adj)}, there
            exists an adjacent $\wposElm \in \MovRel(\posElm) \setminus
            \QSet[i - 1]$ such that $\mfElm[i](\posElm) = \mfElm[i -
            1](\wposElm) + \posElm \neq \top$.
            Thanks to Observation~($\ast$), we have that $\mfElm[i +
            1](\posElm) = \mfElm[i](\posElm)$ and $\mfElm[i](\wposElm) =
            \mfElm[i - 1](\wposElm)$, since $\iota(\posElm) < i$ and
            $\iota(\wposElm) < i - 1$.
            Thus, $\mfElm[i + 1](\posElm) = \mfElm[i](\wposElm) + \posElm \neq
            \top$, as required.
            If $\posElm \not\in \denot{\mfElm}[+] \setminus \QSet[i]$,
            instead, it holds that $\posElm \in \ESet[i]$.
            As a consequence, by definition of the lift operator, there exists
            an adjacent $\wposElm \in \MovRel(\posElm) \setminus \QSet[i]$
            such that $\mfElm[i + 1](\posElm) = \mfElm[i](\wposElm) +
            \posElm$.
            Notice now that $\mfElm[i](\wposElm) \neq \top$.
            Indeed, if $\wposElm \in \denot{\mfElm}[+]$ then
            $\mfElm[i](\wposElm) \neq \top$ directly follows from the
            Inductive Hypothesis~\ref{lmm:prgpls(qdm:adj)}.
            Otherwise, $\wposElm \in \denot{\mfElm}[\bot]$, which implies that
            $\mfElm(\wposElm) \neq \top$, thanks to
            Item~\ref{def:msrspc(trn:top)} of Definition~\ref{def:msrspc}.
            Moreover, $\mfElm[i](\wposElm) = \mfElm(\wposElm)$, by
            Observation~($\ast$).
            To conclude, $\mfElm[i + 1](\posElm) = \mfElm[i](\wposElm) +
            \posElm \neq \top$, due to Item~\ref{def:msrspc(str:top)} of
            Definition~\ref{def:msrspc}.
          \item
            \textbf{[\ref{lmm:prgpls(qdm:mon)}].}
            Let $i > 0$ and $\posElm \in \ESet[i]$.
            Thanks to Observation~($\ast$), we have that $\mfElm[i +
            1](\posElm) = \mfElm[i](\posElm) + \evlElm[i]$ and, so,
            $\evlElm[i] = \mfElm[i + 1](\posElm) - \mfElm[i](\posElm)$, due to
            Proposition~\ref{prp:naievlstr}.
            To continue, we need to consider the following case analysis in
            three parts, which is partially based on ownership of the
            positions $\posElm$.
            \begin{itemize}
              \item
                \textbf{[$\posElm \in \PosSet[\PlrSym] \cap \escFun(\mfElm[i -
                1], \QSet[i - 1])$].}
                Since $\posElm \in \PosSet[\PlrSym] \cap \escFun(\mfElm[i -
                1], \QSet[i - 1])$, it holds that $\mfElm[i - 1](\uposElm) +
                \posElm < \mfElm[i - 1](\posElm)$, for all adjacents $\uposElm
                \in \MovRel(\posElm) \cap \QSet[i - 1]$.
                However, by Condition~\ref{def:regmsr(plr)} of
                Definition~\ref{def:regmsr}, there exists an adjacent
                $\wposElm \in \MovRel(\posElm)$ such that $\mfElm(\posElm)
                \leq \mfElm(\wposElm) + \posElm$.
                By Observation~($\ast$), $\mfElm[i - 1](\posElm) =
                \mfElm(\posElm)$, as $\iota(\posElm) = i$.
                Thus, $\mfElm[i - 1](\posElm) = \mfElm(\posElm) \leq
                \mfElm(\wposElm) + \posElm \leq \mfElm[i - 1](\wposElm) +
                \posElm$, thanks to the Inductive
                Hypothesis~\ref{lmm:prgpls(qdm:inf)} and
                Item~\ref{def:msrspc(str:mon)} of
                Definition~\ref{def:msrspc}.
                As a consequence, $\wposElm \not\in \QSet[i - 1]$ and, so,
                $\wposElm \not\in \QSet[i]$, since all adjacents of $\posElm$
                inside $\QSet[i - 1]$ falsify the inequality, as shown before.
                Moreover, $\posElm \not\in \ESet[i - 1]$, as $\posElm \in
                \ESet[i]$.
                Hence, $\evlElm[i - 1] < \befFun(\mfElm[i - 1], \QSet[i - 1],
                \posElm)$.
                At this point, the following inequalities hold:
                \begin{linenomath}
                \begin{align*}
                  \evlElm[i - 1]
                  & <
                    \befFun(\mfElm[i - 1], \QSet[i - 1], \posElm) \\
                  & =
                    {\min \set{ (\mfElm[i - 1](\uposElm) + \posElm) - \mfElm[i
                    - 1](\posElm) }{ \uposElm \in \MovRel(\posElm) \setminus
                    \QSet[i - 1] }} \\
                  & \leq
                    (\mfElm[i - 1](\wposElm) + \posElm) - \mfElm[i -
                    1](\posElm) \\
                  & =
                    (\mfElm[i](\wposElm) + \posElm) - \mfElm[i](\posElm) \\
                  & \leq
                    {\max \set{ \mfElm[i](\uposElm) + \posElm }{ \uposElm \in
                    \MovRel(\posElm) \setminus \QSet[i] }} -
                    \mfElm[i](\posElm) \\
                  & =
                    \liftFun(\mfElm[i], \ESet[i], \dual{\QSet[i]})(\posElm) -
                    \mfElm[i](\posElm) \\
                  & =
                    \mfElm[i + 1](\posElm) - \mfElm[i](\posElm) \\
                  & =
                    \evlElm[i].
                \end{align*}
                \end{linenomath}
                Notice that the first equality follows from the definition of
                the best escape forfeit function, while the second one is due
                to Observation~($\ast$).
                Indeed, $\mfElm[i - 1](\posElm) = \mfElm[i](\posElm)$ and
                $\mfElm[i - 1](\wposElm) = \mfElm[i](\wposElm)$, since
                $\iota(\posElm) = i$ and $\iota(\wposElm) < i - 1$.
                Finally, from the last inequality onward, we applied the
                definition of the lift operator.
              \item
                \textbf{[$\posElm \in \PosSet[\PlrSym] \setminus
                \escFun(\mfElm[i - 1], \QSet[i - 1])$].}
                Since $\posElm \in \PosSet[\PlrSym] \setminus \escFun(\mfElm[i
                - 1], \QSet[i - 1])$, there exists an adjacent $\wposElm \in
                \MovRel(\posElm) \cap \QSet[i - 1]$ such that $\mfElm[i -
                1](\posElm) \leq \mfElm[i - 1](\wposElm) + \posElm$.
                Moreover, by Observation~($\ast$), $\mfElm[i - 1](\posElm) =
                \mfElm[i](\posElm)$, as $\iota(\posElm) = i$, from which we
                derive $\mfElm[i](\posElm) \leq \mfElm[i - 1](\wposElm) +
                \posElm$ and, so, $(\mfElm[i - 1](\wposElm) + \posElm) -
                \mfElm[i](\posElm) \geq \mathbf{0}$, due to
                Proposition~\ref{prp:naievlstr}.
                Observe now that $\iota(\wposElm) = i - 1$, which implies
                $\mfElm[i](\wposElm) = \mfElm[i - 1](\wposElm) + \evlElm[i -
                1]$, again due to Observation~($\ast$).
                On the contrary, we would have had $\iota(\wposElm) > i - 1$,
                since $\wposElm \in \QSet[i - 1]$,  and, thus,
                $\mfElm[i](\posElm) \leq \mfElm[i - 1](\wposElm) + \posElm =
                \mfElm[i](\wposElm) + \posElm$, contradicting the fact that
                $\posElm \in \ESet[i] \subseteq \escFun(\mfElm[i], \QSet[i])$.
                Obviously, $\wposElm \in \ESet[i - 1]$ and, so, $\wposElm
                \not\in \QSet[i]$
                At this point, the following holds:
                \begin{linenomath}
                \begin{align*}
                  \evlElm[i - 1]
                  & \leq
                    ((\mfElm[i - 1](\wposElm) + \posElm) - \mfElm[i](\posElm))
                    + \evlElm[i - 1] \\
                  & =
                    ((\mfElm[i - 1](\wposElm) + \evlElm[i - 1]) + \posElm) -
                    \mfElm[i](\posElm) \\
                  & =
                    (\mfElm[i](\wposElm) + \posElm) - \mfElm[i](\posElm) \\
                  & \leq
                    {\max \set{ \mfElm[i](\uposElm) + \posElm }{ \uposElm \in
                    \MovRel(\posElm) \setminus \QSet[i] }} -
                    \mfElm[i](\posElm) \\
                  & =
                    \liftFun(\mfElm[i], \ESet[i], \dual{\QSet[i]})(\posElm) -
                    \mfElm[i](\posElm) \\
                  & =
                    \mfElm[i + 1](\posElm) - \mfElm[i](\posElm) \\
                  & =
                    \evlElm[i].
                \end{align*}
                \end{linenomath}
                Notice that the first two derivation steps follow from the
                Abelian group properties of the evaluation structure stated in
                Proposition~\ref{prp:naievlstr} and from
                Item~\ref{prp:naimsrstr(com)} of
                Proposition~\ref{prp:naimsrstr}.
                Moreover, from the last inequality onward, we applied the
                definitions of the lift operator and forfeit values.
              \item
                \textbf{[$\posElm \in \PosSet[\OppSym]$].}
                Since $\posElm \in \ESet[i]$, it holds that $\posElm \in
                \denot{\mfElm}[+] \setminus \QSet[i + 1]$.
                By the Inductive Hypothesis~\ref{lmm:prgpls(qdm:adj)}, there
                exists an adjacent $\wposElm \in \MovRel(\posElm) \setminus
                \QSet[i]$ such that $\mfElm[i + 1](\posElm) =
                \mfElm[i](\wposElm) + \posElm$, which implies $\evlElm[i] =
                (\mfElm[i](\wposElm) + \posElm) - \mfElm[i](\posElm)$.
                Now, consider the nested case analysis on the membership of
                the position $\wposElm$ \wrt $\ESet[i - 1]$.
                \begin{itemize}
                  \item
                    \textbf{[$\wposElm \not\in \ESet[i - 1]$].}
                    Since $\wposElm \not\in \ESet[i - 1]$ and $\wposElm
                    \not\in \QSet[i]$, it holds that $\wposElm \not\in \QSet[i
                    - 1]$.
                    This means that $\wposElm \in \MovRel(\posElm) \setminus
                    \QSet[i - 1] \neq \emptyset$, which in turn implies that
                    $\posElm \in \escFun(\mfElm[i - 1], \QSet[i - 1])$, since
                    $\posElm \in \PosSet[\OppSym]$.
                    Moreover, $\posElm \not\in \ESet[i - 1]$, as $\posElm \in
                    \ESet[i]$.
                    Thus, $\evlElm[i - 1] < \befFun(\mfElm[i - 1], \QSet[i -
                    1], \posElm)$.
                    At this point, the following inequalities hold:
                    \begin{linenomath}
                    \begin{align*}
                      \evlElm[i - 1]
                      & <
                        \befFun(\mfElm[i - 1], \QSet[i - 1], \posElm) \\
                      & =
                        {\min \set{ (\mfElm[i - 1](\uposElm) + \posElm) -
                        \mfElm[i - 1](\posElm) }{ \uposElm \in
                        \MovRel(\posElm) \setminus \QSet[i - 1] }} \\
                      & \leq
                        (\mfElm[i - 1](\wposElm) + \posElm) - \mfElm[i -
                        1](\posElm) \\
                      & =
                        (\mfElm[i](\wposElm) + \posElm) - \mfElm[i](\posElm)
                        \\
                      & =
                        \evlElm[i].
                    \end{align*}
                    \end{linenomath}
                    Notice that the first equality follows from the definition
                    of the best escape forfeit function, while the second one
                    is due to Observation~($\ast$).
                    Indeed, $\mfElm[i - 1](\posElm) = \mfElm[i](\posElm)$ and
                    $\mfElm[i - 1](\wposElm) = \mfElm[i](\wposElm)$, since
                    $\iota(\posElm) = i$ and $\iota(\wposElm) < i - 1$.
                  \item
                    \textbf{[$\wposElm \in \ESet[i - 1]$].}
                    Since $\wposElm \in \ESet[i - 1]$, by
                    Observation~($\ast$), it follows that $\mfElm[i](\wposElm)
                    = \mfElm[i - 1](\wposElm) + \evlElm[i - 1]$, as
                    $\iota(\wposElm) = i - 1$.
                    Similarly, $\mfElm[i](\posElm) = \mfElm(\posElm)$, as
                    $\iota(\posElm) = i$.
                    Now, by Condition~\ref{def:regmsr(opp)} of
                    Definition~\ref{def:regmsr}, $\mfElm[i](\posElm) =
                    \mfElm(\posElm) \leq \mfElm(\uposElm) + \posElm \leq
                    \mfElm[i - 1](\uposElm) + \posElm$, for all adjacents
                    $\uposElm \in \MovRel(\posElm)$, thanks to the Inductive
                    Hypothesis~\ref{lmm:prgpls(qdm:inf)} and
                    Item~\ref{def:msrspc(str:mon)} of
                    Definition~\ref{def:msrspc}.
                    As an immediate consequence, $\mfElm[i](\posElm) \leq
                    \mfElm[i - 1](\wposElm) + \posElm$, \ie, $(\mfElm[i -
                    1](\wposElm) + \posElm) - \mfElm[i](\posElm) \geq
                    \mathbf{0}$.
                    At this point, the following holds:
                    \begin{linenomath}
                    \begin{align*}
                      \evlElm[i - 1]
                      & \leq
                        ((\mfElm[i - 1](\wposElm) + \posElm) -
                        \mfElm[i](\posElm)) + \evlElm[i - 1] \\
                      & =
                        ((\mfElm[i - 1](\wposElm) + \evlElm[i - 1]) + \posElm)
                        - \mfElm[i](\posElm) \\
                      & =
                        (\mfElm[i](\wposElm) + \posElm) - \mfElm[i](\posElm)
                        \\
                      & =
                        \evlElm[i].
                    \end{align*}
                    \end{linenomath}
                    Notice that the first two derivation steps follow from the
                    Abelian group properties of the evaluation structure
                    stated in Proposition~\ref{prp:naievlstr} and from
                    Item~\ref{prp:naimsrstr(com)} of
                    Proposition~\ref{prp:naimsrstr}.
                \end{itemize}
            \end{itemize}
            Summing up, in all cases we have $\evlElm[i - 1] \leq \evlElm[i]$.
          \item
            \textbf{[\ref{lmm:prgpls(qdm:bot)}].}
            If $i > 0$, by the Inductive Hypotheses~\ref{lmm:prgpls(qdm:bot)}
            and~\ref{lmm:prgpls(qdm:mon)}, it holds that $\mathbf{0} \leq
            \evlElm[i - 1]$ and $\evlElm[i - 1] \leq \evlElm[i]$.
            Hence, $\mathbf{0} \leq \evlElm[i]$.
            If $i = 0$, instead, let $\posElm \in \ESet[0]$.
            Thanks to Observation~($\ast$), one has that $\mfElm[1](\posElm) =
            \mfElm[0](\posElm) + \evlElm[0]$ and, so, $\evlElm[0] =
            \mfElm[1](\posElm) - \mfElm[0](\posElm)$, due to
            Proposition~\ref{prp:naievlstr}.
            To continue, we need to consider the following case analysis on
            the ownership of the position $\posElm$.
            \begin{itemize}
              \item
                \textbf{[$\posElm \in \PosSet[\PlrSym]$].}
                By Condition~\ref{def:regmsr(plr)} of
                Definition~\ref{def:regmsr}, there exists an adjacent
                $\wposElm \in \MovRel(\posElm)$ such that $\mfElm[0](\posElm)
                \leq \mfElm[0](\wposElm) + \posElm$, as $\mfElm[0] = \mfElm$.
                Since $\posElm \in \ESet[0]$, thanks to the definitions of
                both the best escape forfeit and escape functions, it holds
                that $\mfElm[0](\uposElm) + \posElm < \mfElm[0](\posElm)$, for
                all adjacents $\uposElm \in \MovRel(\posElm) \cap \QSet[0]$.
                Hence, as an immediate consequence, $\wposElm \not\in
                \QSet[0]$.
                Now, by definition of the lift operator, $\mfElm[1](\posElm) =
                \max \set{ \mfElm[0](\uposElm) + \posElm }{ \uposElm \in
                \MovRel(\posElm) \setminus \QSet[0] } \geq \mfElm[0](\wposElm)
                + \posElm \geq \mfElm[0](\posElm)$, from which it follows that
                $\evlElm[0] = \mfElm[1](\posElm) - \mfElm[0](\posElm) \geq
                \mathbf{0}$, due to Proposition~\ref{prp:naievlstr}.
              \item
                \textbf{[$\posElm \in \PosSet[\OppSym]$].}
                By the Inductive Hypothesis~\ref{lmm:prgpls(qdm:adj)}, there
                exists an adjacent $\wposElm \in \MovRel(\posElm) \setminus
                \QSet[0]$ such that $\mfElm[1](\posElm) = \mfElm[0](\wposElm)
                + \posElm$, which implies $\evlElm[0] = (\mfElm[0](\wposElm) +
                \posElm) - \mfElm[0](\posElm)$.
                Moreover, by Condition~\ref{def:regmsr(opp)} of
                Definition~\ref{def:regmsr}, $\mfElm[0](\posElm) \leq
                \mfElm[0](\uposElm) + \posElm$, for all adjacents $\uposElm
                \in \MovRel(\posElm)$, since $\mfElm[0] = \mfElm$.
                Thus, as an obvious consequence, $\mfElm[0](\wposElm) +
                \posElm \geq \mfElm[0](\posElm)$, which immediately implies
                $\evlElm[0] = (\mfElm[0](\wposElm) + \posElm) -
                \mfElm[0](\posElm) \geq \mathbf{0}$, as required, again due
                to Proposition~\ref{prp:naievlstr}.
            \end{itemize}
            Summing up, in both cases we have $\mathbf{0} \leq \evlElm[0]$.
        \end{itemize}
        We now have the necessary tool to show that the progress operator is
        inflationary.
        Indeed, by Property~\ref{lmm:prgpls(qdm:inf)}, $\mfElm \sqsubseteq
        \mfElm[k]$.
        Moreover, $\mfElm[][\star]$ and $\mfElm[k]$ differ only on positions
        $\posElm \in \denot{\mfElm[][\star]}[+]$ such that
        $\mfElm[][\star](\posElm) = \top$.
        Hence, it easily follows that $\mfElm \sqsubseteq \mfElm[][\star]$,
        thanks to Item~\ref{def:msrspc(ord)} of Definition~\ref{def:msrspc}.
        Observe also that $\denot{\mfElm}[+] \subseteq \denot{\mfElm[k]}[+]
        \subseteq \denot{\mfElm[][\star]}[+]$ and $\denot{\mfElm}[\PlrSym]
        \subseteq \denot{\mfElm[k]}[\PlrSym] \subseteq
        \denot{\mfElm[][\star]}[\PlrSym]$, due to
        Item~\ref{def:msrspc(trn:mon)} of the same definition.

        \hspace{1em}
        It remains to prove that $\mfElm[][\star]$ is a \qdmf.
        First notice that, for any position $\posElm \in \PosSet$, if
        $\mfElm[k](\posElm) \neq \mfElm(\posElm)$ then $\posElm \in
        \denot{\mfElm}[+] \setminus \QSet[k]$, due to Observation~($\ast$).
        Thus, by Property~\ref{lmm:prgpls(qdm:adj)}, for such a position
        $\posElm$, there exists an adjacent $\wposElm \in \MovRel(\posElm)
        \setminus \QSet[k - 1] \subseteq \MovRel(\posElm) \setminus \QSet[k]$
        such that $\mfElm[k](\posElm) = \mfElm[k - 1](\wposElm) + \posElm \neq
        \top$.
        Since $\wposElm \not\in \QSet[k - 1]$, it holds that $\iota(\wposElm)
        < k - 1$, so, $\mfElm[k](\wposElm) = \mfElm[k - 1](\wposElm)$, again
        due to Observation~($\ast$).
        Hence, for all positions $\posElm \in \PosSet$, either one of the
        following two possibilities holds:
        \[
          \begin{rcases}
            \posElm \!\not\in\! \denot{\mfElm}[+] \!\!\setminus\! \QSet[k]
            \text{ and } \mfElm[k](\posElm) \!=\! \mfElm(\posElm); \\
            \posElm \!\in\! \denot{\mfElm}[+] \!\!\setminus\! \QSet[k] \text{
            and there exists an adjacent } \wposElm \!\in\! \MovRel(\posElm)
            \!\setminus\! \QSet[k] \text{ such that } \mfElm[k](\posElm) \!=\!
            \mfElm[k](\wposElm) \!+\! \posElm \!\neq\! \top.
          \end{rcases}
          \tag{$\sharp$}
        \]
        From this, we easily derive that $\denot{\mfElm}[+] =
        \denot{\mfElm[k]}[+] = \denot{\mfElm[][\star]}[+]$ and
        $\denot{\mfElm}[\PlrSym] = \denot{\mfElm[k]}[\PlrSym] \subseteq
        \denot{\mfElm[][\star]}[\PlrSym] = \QSet[k]$.
        Indeed, the only positions that can change their measure are those in
        $\denot{\mfElm}[+]$ and they are not set to $\top$ in $\mfElm[k]$.
        Moreover, $\denot{\mfElm}[\PlrSym] \subseteq \denot{\mfElm}[+]$, so,
        $\denot{\mfElm[k]}[\PlrSym] \subseteq \denot{\mfElm}[+]$.
        Now, $(\denot{\mfElm}[+] \setminus \QSet[k]) \cap
        \denot{\mfElm[k]}[\PlrSym] = \emptyset$, hence,
        $\denot{\mfElm[k]}[\PlrSym] \subseteq \QSet[k]$.
        Finally, $\denot{\mfElm[][\star]}[\PlrSym] = \QSet[k]$, since, by
        construction, $\denot{\mfElm[][\star]}[\PlrSym] =
        \denot{\mfElm[k]}[\PlrSym] \cup \QSet[k]$.

        \hspace{1em}
        We can now show that $\mfElm[][\star]$ is a regress measure, \ie, that
        every position $\posElm \in \denot{\mfElm[][\star]}[+] \setminus
        \denot{\mfElm[][\star]}[\PlrSym]$ satisfies the suitable condition of
        Definition~\ref{def:regmsr}.
        We do this, via a case analysis on the ownership of $\posElm$.
        \begin{itemize}
          \item
            \textbf{[$\posElm \in \PosSet[\PlrSym]$].}
            Since $\denot{\mfElm[][\star]}[+] \setminus
            \denot{\mfElm[][\star]}[\PlrSym] = \denot{\mfElm}[+] \setminus
            \QSet[k]$, by Observation~($\sharp$), there exists an adjacent
            $\wposElm \in \MovRel(\posElm) \setminus \QSet[k]$ such
            that $\mfElm[k](\posElm) = \mfElm[k](\wposElm) + \posElm$.
            Obviously, $\mfElm[][\star](\posElm) = \mfElm[k](\posElm)$ and
            $\mfElm[][\star](\wposElm) = \mfElm[k](\wposElm)$, since $\posElm,
            \wposElm \not\in \QSet[k]$.
            Thus, $\mfElm[][\star](\posElm) = \mfElm[][\star](\wposElm) +
            \posElm$, as required by Condition~\ref{def:regmsr(plr)}.
          \item
            \textbf{[$\posElm \in \PosSet[\OppSym]$].}
            Suppose by contradiction that $\mfElm[][\star]$ does not satisfy
            Condition~\ref{def:regmsr(opp)} on $\posElm$.
            Then, there exists one of its adjacents $\wposElm \in
            \MovRel(\posElm)$ such that $\mfElm[][\star](\wposElm) + \posElm <
            \mfElm[][\star](\posElm)$.
            Obviously, $\mfElm[][\star](\wposElm) + \posElm \neq \top$, so,
            $\mfElm[][\star](\wposElm) \neq \top$, due to
            Item~\ref{def:msrspc(str:top)} of Definition~\ref{def:msrspc},
            which in turn implies that $\wposElm \not\in \QSet[k]$.
            As a consequence, $\posElm \not\in \QSet[k]$ too, since we would
            have had, otherwise, $\MovRel(\posElm) \subseteq \QSet[k]$, as
            $\escFun(\mfElm[k], \QSet[k]) = \emptyset$.
            Thus, $\posElm \in \denot{\mfElm}[+] \setminus \QSet[k]$, from
            which it follows that $\mfElm[][\star](\posElm) =
            \mfElm[k][](\posElm) = \mfElm[\iota(\posElm) + 1](\posElm) =
            \mfElm[\iota(\posElm)](\posElm) + \evlElm[\iota(\posElm)] =
            \mfElm(\posElm) + \evlElm[\iota(\posElm)]$, due to
            Observation~($\ast$).
            To proceed, we now need the following nested case analysis, which
            allows to prove that $\wposElm \not\in \QSet[\iota(\posElm)]$ and
            $\mfElm[][\star](\wposElm) = \mfElm[\iota(\posElm)](\wposElm)$.
            \begin{itemize}
              \item
                \textbf{[$\wposElm \in \denot{\mfElm[][\star]}[\bot]$].}
                Notice that $\wposElm \not\in \QSet[0] = \denot{\mfElm}[+] =
                \denot{\mfElm[][\star]}[+]$ and, so, $\wposElm \not\in
                \QSet[\iota(\posElm)]$, as $\QSet[\iota(\posElm)] \subseteq
                \QSet[0]$.
                Moreover, $\mfElm[][\star](\wposElm) = \mfElm[k][](\wposElm) =
                \mfElm[\iota(\posElm)](\wposElm)$, due to
                Observations~($\ast$).
              \item
                \textbf{[$\wposElm \in \denot{\mfElm[][\star]}[+]$].}
                By Observations~($\ast$), it holds that
                $\mfElm[][\star](\wposElm) = \mfElm[k][](\wposElm) =
                \mfElm[\iota(\wposElm)](\wposElm) + \evlElm[\iota(\wposElm)] =
                \mfElm(\wposElm) + \evlElm[\iota(\wposElm)]$.
                Now, by substituting in $\mfElm[][\star](\wposElm) + \posElm <
                \mfElm[][\star](\posElm)$ both $\mfElm[][\star](\wposElm)$
                and $\mfElm[][\star](\posElm)$ with $\mfElm(\wposElm) +
                \evlElm[\iota(\wposElm)]$ and $\mfElm(\posElm) +
                \evlElm[\iota(\posElm)]$, respectively, and exploiting
                Proposition~\ref{prp:naievlstr} and
                Item~\ref{prp:naimsrstr(com)} of
                Proposition~\ref{prp:naimsrstr}, one can obtain
                $(\mfElm(\wposElm) + \posElm) - \mfElm(\posElm) <
                \evlElm[\iota(\posElm)] - \evlElm[\iota(\wposElm)]$.
                Since $\mfElm$ is a \qdmf and, so, a regress measure, we have
                that it satisfies Condition~\ref{def:regmsr(opp)} on
                $\posElm$, \ie, $\mfElm(\posElm) \leq \mfElm(\wposElm) +
                \posElm$, which implies $(\mfElm(\wposElm) + \posElm) -
                \mfElm(\posElm) \geq \mathbf{0}$.
                Hence, $\evlElm[\iota(\posElm)] - \evlElm[\iota(\wposElm)] >
                \mathbf{0}$ and, so, $\iota(\posElm) > \iota(\wposElm)$,
                thanks to Property~\ref{lmm:prgpls(qdm:mon)}.
                From this we can derive that $\wposElm \not\in
                \QSet[\iota(\posElm)]$ and $\mfElm[][\star](\wposElm) =
                \mfElm[k][](\wposElm) = \mfElm[\iota(\posElm)](\wposElm)$, due
                to Observations~($\ast$).
            \end{itemize}
            At this point, the following impossible inequality chain
            should hold:
            \begin{linenomath}
            \begin{align*}
              \mfElm[][\star](\posElm)
              & =
                \mfElm[\iota(\posElm) + 1](\posElm) \\
              & =
                \liftFun(\mfElm[\iota(\posElm)], \ESet[\iota(\posElm)],
                \dual{\QSet[\iota(\posElm)]})(\posElm) \\
              & =
                \min {\set{ \mfElm[\iota(\posElm)](\uposElm) + \posElm }{
                \uposElm \in \MovRel(\posElm) \setminus
                \QSet[\iota(\posElm)] }} \\
              & \leq
                \mfElm[\iota(\posElm)](\wposElm) + \posElm \\
              & =
                \mfElm[][\star](\wposElm) + \posElm \\
              & <
                \mfElm[][\star](\posElm).
            \end{align*}
            \end{linenomath}
        \end{itemize}

        \hspace{1em}
        Finally, we can conclude the proof of this item by showing that
        $\denot{\mfElm[][\star]}[\PlrSym]$, which we now know to be equal to
        $\QSet[k]$, is a $\PlrSym$-dominion.
        Let $\strElm[\PlrSym] \in \StrSet[\PlrSym]$ be a $\mfElm$-coherent
        $\PlrSym$-strategy winning on $\denot{\mfElm}[\PlrSym]$ such that
        $\strElm[\PlrSym](\posElm) \in \QSet[k]$, for all $\PlrSym$-positions
        $\posElm \in \QSet[k] \cap \PosSet[\PlrSym]$.
        Such a strategy surely exists, since $\mfElm$ is, by hypothesis, a
        \qdmf and $\QSet[k]$ is closed, \ie, $\escFun(\mfElm, \QSet[k]) =
        \escFun(\mfElm[k], \QSet[k]) = \emptyset$.
        To state the first equality we exploited the fact that
        $\mfElm[k](\posElm) = \mfElm(\posElm)$, for all positions $\posElm \in
        \QSet[k]$, due to Observation~($\sharp$).
        Now, by Theorem~\ref{thm:qsidommsri}, $\denot{\mfElm}[+]$ is a weak
        quasi $\PlrSym$-dominion, for which $\strElm[\PlrSym] \downarrow
        \denot{\mfElm}[+] \in \StrSet[\PlrSym](\denot{\mfElm}[+])$ is a
        $\PlrSym$-witness.
        Hence, $\QSet[k]$ and, so, $\denot{\mfElm[][\star]}[\PlrSym]$, is a
        $\PlrSym$-dominion, thanks to Corollary~\ref{cor:qsidom}.
      \item
        \textbf{[\ref{lmm:prgpls(smf)}].}
        To show that $\mfElm[][\star]$ is a \smf, whenever $\mfElm$ is a \smf,
        we first prove the following statement: $\mfElm[i](\posElm) \in
        \SMsrSet(\posElm, \denot{\mfElm}[+] \setminus \QSet[i]) \setminus \{
        \top \}$, for all indexes $i \in \numcc{0}{k}$ and positions $\posElm
        \in \PosSet \setminus \QSet[i]$.

        \hspace{1em}
        The proof proceeds by induction on $i$.
        The base case $i = 0$ trivially follows from the hypothesis, since
        $\mfElm[0] = \mfElm$.
        Indeed, it holds that $\PosSet \setminus \QSet[0] =
        \denot{\mfElm}[\bot]$ and $\mfElm[i](\posElm) = \bot$, for all
        positions $\posElm \in \denot{\mfElm}[\bot]$, due to
        Proposition~\ref{prp:botden}.
        For the inductive case $i > 0$, let $\posElm \in \PosSet \setminus
        \QSet[i]$ and assume that $\mfElm[i - 1](\wposElm) \in
        \SMsrSet(\wposElm, \denot{\mfElm}[+] \setminus \QSet[i - 1]) \setminus
        \{ \top \}$, for all positions $\wposElm \in \PosSet \setminus \QSet[i
        - 1]$.
        If $\posElm \in \PosSet \setminus \QSet[i - 1]$, there is nothing more
        to prove, since $\mfElm[i](\posElm) = \mfElm[i - 1](\posElm)$ and,
        by inductive hypothesis, $\mfElm[i - 1](\posElm) \in \SMsrSet(\posElm,
        \denot{\mfElm}[+] \setminus \QSet[i - 1]) \setminus \{ \top \}
        \subseteq \SMsrSet(\posElm, \denot{\mfElm}[+] \setminus \QSet[i])
        \setminus \{ \top \}$, being $\denot{\mfElm}[+] \setminus \QSet[i - 1]
        \subseteq \denot{\mfElm}[+] \setminus \QSet[i]$.
        Therefore, consider the case $\posElm \not\in \PosSet \setminus
        \QSet[i - 1]$, which implies that $\posElm \in \ESet[i - 1] \subseteq
        \QSet[i - 1]$.
        By definition of the lift operator, there exists an adjacent $\wposElm
        \in \MovRel(\posElm) \setminus \QSet[i - 1] \subseteq \PosSet
        \setminus \QSet[i - 1]$ such that  $\mfElm[i](\posElm) = \mfElm[i -
        1](\wposElm) + \posElm$.
        By the inductive hypothesis, $\mfElm[i - 1](\wposElm) \neq \top$.
        Thus, by Item~\ref{def:msrspc(str:top)} of
        Definition~\ref{def:msrspc}, it holds that $\mfElm[i](\posElm) \neq
        \top$.
        Moreover, there exists a simple path $\pthElm \in \SPthSet(\wposElm,
        \denot{\mfElm}[+] \setminus \QSet[i - 1])$ such that $\mfElm[i -
        1](\wposElm) = \msrFun(\pthElm)$.
        Obviously, $\mfElm[i](\posElm) = \mfElm[i - 1](\wposElm) + \posElm =
        \msrFun(\pthElm) + \posElm = \msrFun(\posElm \cdot \pthElm) \neq
        \top$.
        Now, it is quite clear that $\posElm \cdot \pthElm$ is a simple path
        passing through positions in $\{ \posElm \} \cup \denot{\mfElm}[+]
        \setminus \QSet[i - 1]$, \ie, $\posElm \cdot \pthElm \in
        \SPthSet(\posElm, \denot{\mfElm}[+] \setminus \QSet[i])$, since
        $\posElm \in \QSet[i - 1]$.
        Hence, $\mfElm[i](\posElm) \in \SMsrSet(\posElm, \denot{\mfElm}[+]
        \setminus \QSet[i]) \setminus \{ \top \}$.
        This conclude the inductive proof.

        \hspace{1em}
        At this point, it immediately follows, from what we have just proved,
        that $\mfElm[k]$ is \smf.
        In addition, $\mfElm[][\star]$ potentially differs from $\mfElm[k]$
        only on positions $\posElm \in \denot{\mfElm[][\star]}[+]$ such that
        $\mfElm[][\star](\posElm) = \top$.
        Consequently, $\mfElm[][\star]$ is a \smf as well.
      \item
        \textbf{[\ref{lmm:prgpls(fix)}].}
        By hypothesis, $\mfElm[][\star] = \mfElm$, which implies that
        $\mfElm[i] = \mfElm$, for all $i \in \SetN$, due to the way the
        sequence $\mfElm[0], \mfElm[1], \ldots$ is constructed.
        Now, let us consider an arbitrary position $\posElm \in
        \denot{\mfElm}[+]$.
        Due to the definition of the sequence $\QSet[0], \QSet[1], \ldots$, it
        obviously holds that either $\posElm \in \QSet[k]$ or there is a
        unique index $i \in \numco{0}{k}$ such that $\posElm \in \QSet[i]
        \setminus \QSet[i + 1]$, \ie, $\posElm \in \ESet[i]$.
        In the first case, we have $\mfElm[](\posElm) = \top$, due to the
        assignment $\mfElm[][\star] = {\mfElm[k]}[\QSet[k] \mapsto \top]$.
        Therefore, $\posElm$ is a progress position, \ie, it satisfies both
        conditions of Definition\ref{def:prgmsr}.
        In the other case, the proof proceeds by a case analysis on the
        ownership of the position $\posElm$ itself.
        \begin{itemize}
          \item
            \textbf{[$\posElm \in \PosSet[\PlrSym]$].}
            First recall that $\ESet[i] = \bepFun(\mfElm, \QSet[i]) \subseteq
            \escFun(\mfElm, \QSet[i])$.
            Thus, due to the definition of the escape function, we have that
            $\mfElm[](\wposElm) + \posElm < \mfElm[](\posElm)$, for all
            positions $\wposElm \in \MovRel(\posElm) \cap \QSet[i]$.
            Now, by definition of the lift operator, we have that
            $\mfElm[](\wposElm) + \posElm \leq \max \set{ \mfElm[](\wposElm) +
            \posElm }{ \wposElm \in \MovRel(\posElm) \cap \dual{\QSet[i]} } =
            \mfElm[](\posElm)$, for all adjacents $\wposElm \in
            \MovRel(\posElm)\cap \dual{\QSet[i]}$ of $\posElm$.
            Thus, $\mfElm[](\wposElm) + \posElm \leq \mfElm[](\posElm)$, for
            all positions $\wposElm \in \MovRel(\posElm)$, as required by
            Condition~\ref{def:prgmsr(plr)} of Definition~\ref{def:prgmsr} on
            $\denot{\mfElm}[+]$.
          \item
            \textbf{[$\posElm \in \PosSet[\OppSym]$].}
            Again by definition of the lift operator, we have that
            $\mfElm[](\wposElm) + \posElm \leq \min \set{
            \mfElm[](\wposElm) + \posElm }{ \wposElm \in \MovRel(\posElm)
            \cap \dual{\QSet[i]} } = \mfElm[](\posElm)$, for some adjacent
            $\wposElm \in \MovRel(\posElm) \cap \dual{\QSet[i]} \subseteq
            \MovRel(\posElm)$ of $\posElm$.
            Hence, Condition~\ref{def:prgmsr(opp)} of
            Definition~\ref{def:prgmsr} is satisfied on $\denot{\mfElm}[+]$ as
            well.
            \qed
        \end{itemize}
    \end{itemize}
    \renewcommand{\qed}{}
  \end{proof}

\end{section}

% End of file AppendixC.tex

  % \input{AppendixD}

\end{document}

% End of file Article.tex